%% file: main.tex
\documentclass[submission,copyright,creativecommons]{eptcs}
\providecommand{\event}{QPL 2024} 

\usepackage[utf8]{inputenc}
\usepackage{multicol}
\usepackage{multirow}
\usepackage{xcolor}
\usepackage{graphicx}
\usepackage{cite}
\usepackage{adjustbox}
\usepackage{float}
\usepackage{amsmath}
\usepackage{amssymb}
\usepackage{amsthm}
\usepackage{diagbox}
\usepackage{url}
\usepackage{xspace}
\usepackage{tikzit}
\usepackage{listings}
\usepackage[colorinlistoftodos,
            prependcaption,
            textsize=scriptsize,
            backgroundcolor=yellow,
            linecolor=lightgray,
            bordercolor=lightgray]{todonotes}

\input{quantum.tikzdefs}
\input{quantum.tikzstyles}

\renewcommand{\tr}[1]{\operatorname{tr}\left(#1\right)}
\newcommand{\ptr}[2]{\operatorname{tr}_{#1}\left(#2\right)}

\makeatletter
\newtheorem*{rep@theorem}{\rep@title}
\newcommand{\newreptheorem}[2]{%
\newenvironment{rep#1}[1]{%
 \def\rep@title{#2 \ref{##1}}%
 \begin{rep@theorem}}%
 {\end{rep@theorem}}}
\makeatother

\newtheorem*{proposition}{Proposition}

\newtheorem{definition}{Definition}
\newtheorem{theorem}{Theorem}
\newreptheorem{theorem}{Theorem}

\newcommand{\figscale}{1.0}
\newcommand{\blockfig}[1]{$$\scalebox{\figscale}{\includegraphics{figs/#1.pdf}}$$}

\title{Quantum Algorithms for Compositional Text Processing}

\author{
    Tuomas Laakkonen, Konstantinos Meichanetzidis, Bob Coecke
    \institute{Quantinuum, 17 Beaumont Street, Oxford OX1 2NA, United Kingdom}
    \email{\{tuomas.laakkonen, k.mei, bob.coecke\}@quantinuum.com}
}

\def\titlerunning{Quantum Algorithms for Compositional Text Processing}
\def\authorrunning{T. Laakkonen, K. Meichanetzidis, B. Coecke}

\begin{document}

\maketitle

\input{sections/intro}

\input{sections/qdiscocirc}

\input{sections/hardness}
\input{sections/grover}
\input{sections/other-tasks}
\input{sections/conclusion}

\paragraph{Acknowledgments} We thank Matty Hoban for discussions on complexity theory, Harry Buhrman, Richie Yeung, Luuk Coopmans and Ilyas Khan for comments on the manuscript, the QDisCoCirc team (Saskia Bruhn, Tiffany Duneau, Gabriel Matos, Anna Pearson, and Katerina Saiti) for helpful discussions refining the model, and the parser-pipeline team (Jonathon Liu, Razin Shaikh, and Benjamin Rodatz) for answering all of our questions.


\bibliographystyle{eptcs}
\bibliography{refs}

\clearpage
\appendix
\input{sections/frame-decomp}

\input{sections/proofs}
\input{sections/oracle}

\end{document}

%% file: quantum.tikzstyles
\tikzstyle{gate}=[shape=rectangle, text height=1.5ex, text depth=0.25ex, yshift=0.5mm, fill=white, draw=black, minimum height=3mm, yshift=-0.5mm, minimum width=3mm, font={\small}, tikzit category=circuit, inner sep=2pt]
\tikzstyle{big gate}=[shape=rectangle, text height=1.5ex, text depth=0.25ex, yshift=0.5mm, fill=white, draw=black, minimum height=10mm, yshift=-0.5mm, minimum width=5mm, font={\small}, tikzit category=circuit]
\tikzstyle{Z dot}=[inner sep=0mm, minimum size=2mm, shape=circle, draw=black, fill={rgb,255: red,221; green,255; blue,221}, tikzit category=zx]
\tikzstyle{Z phase dot}=[minimum size=5mm, font={\footnotesize\boldmath}, shape=rectangle, rounded corners=2mm, inner sep=0.2mm, outer sep=-2mm, scale=0.8, tikzit shape=circle, draw=black, fill={rgb,255: red,221; green,255; blue,221}, tikzit draw=blue, tikzit category=zx]
\tikzstyle{X dot}=[Z dot, shape=circle, draw=black, fill={rgb,255: red,255; green,136; blue,136}, tikzit category=zx]
\tikzstyle{X phase dot}=[Z phase dot, tikzit shape=circle, tikzit draw=blue, fill={rgb,255: red,255; green,136; blue,136}, font={\footnotesize\boldmath}, tikzit category=zx]
\tikzstyle{hadamard}=[fill=yellow, draw=black, shape=rectangle, inner sep=0.6mm, minimum height=1.5mm, minimum width=1.5mm, tikzit category=zx]
\tikzstyle{Q Z dot}=[fill={rgb,255: red,221; green,255; blue,221}, draw=black, shape=circle, thick, tikzit category=zx, tikzit draw=red, minimum size=2mm, inner sep=0mm]
\tikzstyle{Q X dot}=[fill={rgb,255: red,255; green,136; blue,136}, draw=black, shape=circle, tikzit category=zx, tikzit draw=red, thick, inner sep=0mm, minimum size=2mm]
\tikzstyle{Q hadamard}=[fill=yellow, draw=black, thick, shape=rectangle, inner sep=0.6mm, minimum height=1.5mm, minimum width=1.5mm, tikzit category=zx, tikzit draw=red]
\tikzstyle{paulibox}=[fill={rgb,255: red,221; green,221; blue,255}, draw=black, shape=rectangle, inner sep=0.6mm, minimum height=5mm, minimum width=5mm, font={\footnotesize}, text height=1.5ex, text depth=0.25ex, tikzit category=zx]
\tikzstyle{vertex}=[inner sep=0mm, minimum size=1mm, shape=circle, draw=black, fill=black, tikzit category=misc]
\tikzstyle{vertex set}=[inner sep=0mm, minimum size=1mm, shape=circle, draw=black, fill=white, font={\footnotesize\boldmath}, tikzit category=misc]
\tikzstyle{small black dot}=[fill=black, draw=black, shape=circle, inner sep=0pt, minimum width=1.2mm, tikzit category=circuit]
\tikzstyle{cnot ctrl}=[fill=black, draw=black, shape=circle, inner sep=0pt, minimum width=1.2mm, tikzit category=circuit]
\tikzstyle{cnot targ}=[fill=white, draw=white, shape=circle, tikzit category=circuit, label={center:$\oplus$}, inner sep=0pt, minimum width=2.1mm, tikzit fill={rgb,255: red,102; green,204; blue,255}, tikzit draw=black]
\tikzstyle{ket}=[fill=white, draw=black, shape=regular polygon, regular polygon sides=3, regular polygon rotate=-30, scale=0.7, inner sep=1pt, tikzit category=circuit, tikzit shape=rectangle, tikzit fill=green]
\tikzstyle{Q ket}=[fill=white, draw=black, thick, shape=regular polygon, regular polygon sides=3, regular polygon rotate=-30, scale=0.7, inner sep=1pt, tikzit category=circuit, tikzit shape=rectangle, tikzit fill=green, tikzit draw=red]
\tikzstyle{bra}=[fill=white, draw=black, shape=regular polygon, regular polygon sides=3, regular polygon rotate=30, scale=0.7, inner sep=1pt, tikzit category=circuit, tikzit shape=rectangle, tikzit fill=red]
\tikzstyle{scalar}=[shape=rectangle, text height=1.5ex, text depth=0.25ex, yshift=0.5mm, fill=white, draw=black, minimum height=5mm, yshift=-0.5mm, minimum width=5mm, font={\small}]
\tikzstyle{clabel}=[fill=white, draw=none, shape=rectangle, tikzit fill={rgb,255: red,56; green,255; blue,242}, font={\footnotesize}, inner sep=1pt, tikzit category=labels]
\tikzstyle{empty diagram}=[draw={gray!40!white}, dashed, shape=rectangle, minimum width=1cm, minimum height=1cm, tikzit category=misc]
\tikzstyle{amap}=[fill=white, draw=black, shape=NEbox, tikzit category=asymmetric, tikzit fill=yellow, tikzit shape=rectangle]
\tikzstyle{amap conj}=[fill=white, draw=black, shape=NWbox, tikzit category=asymmetric, tikzit fill=green, tikzit shape=rectangle]
\tikzstyle{amap adj}=[fill=white, draw=black, shape=SEbox, tikzit category=asymmetric, tikzit fill=red, tikzit shape=rectangle]
\tikzstyle{amap trans}=[fill=white, draw=black, shape=SWbox, tikzit category=asymmetric, tikzit fill=orange, tikzit shape=rectangle]
\tikzstyle{astate}=[fill=white, draw=black, shape=NEtriangle, tikzit category=asymmetric, tikzit shape=circle, tikzit fill=yellow]
\tikzstyle{astate conj}=[fill=white, draw=black, shape=NWtriangle, tikzit category=asymmetric, tikzit shape=circle, tikzit fill=green]
\tikzstyle{astate adj}=[fill=white, draw=black, shape=SEtriangle, tikzit category=asymmetric, tikzit shape=circle, tikzit fill=red]
\tikzstyle{astate trans}=[fill=white, draw=black, shape=SWtriangle, tikzit category=asymmetric, tikzit shape=circle, tikzit fill=orange]
\tikzstyle{white dot}=[inner sep=0mm, minimum size=2mm, shape=circle, draw=black, fill={rgb,255: red,250; green,250; blue,250}]
\tikzstyle{white phase dot}=[minimum size=5mm, font={\footnotesize\boldmath}, shape=rectangle, rounded corners=2mm, inner sep=0.2mm, outer sep=-2mm, scale=0.8, tikzit shape=circle, draw=black, fill={rgb,255: red,250; green,250; blue,250}, tikzit draw=blue]
\tikzstyle{hbox}=[shape=rectangle, text height=2mm, fill={rgb,255: red,255; green,235; blue,61}, draw=black, minimum height=2mm, minimum width=2mm, font={\small}, tikzit category=zh, inner sep=0pt, rounded corners=0.5mm]
\tikzstyle{Z dot (zh)}=[inner sep=0mm, minimum size=2mm, shape=circle, draw=black, fill={rgb,255: red,250; green,250; blue,250}, tikzit category=zh]
\tikzstyle{X dot (zh)}=[Z dot, shape=circle, draw=black, fill={rgb,255: red,193; green,193; blue,193}, tikzit category=zh]
\tikzstyle{triangle}=[fill=yellow, draw=black, shape=isosceles triangle, isosceles triangle apex angle=60, minimum size=2.5mm, inner sep=0mm]
\tikzstyle{labelled hbox}=[shape=rectangle, text height=1.75ex, text depth=0.5ex, fill={rgb,255: red,255; green,235; blue,61}, draw=black, minimum height=3mm, minimum width=4mm, font={\small}, tikzit category=zh, inner sep=1.3pt, rounded corners=0.5mm]
\tikzstyle{Z phase dot (zh)}=[Z phase dot, tikzit shape=circle, tikzit draw=blue, fill={rgb,255: red,250; green,250; blue,250}, font={\footnotesize\boldmath}, tikzit category=zh]
\tikzstyle{X phase dot (zh)}=[Z phase dot, tikzit shape=circle, tikzit draw=blue, fill={rgb,255: red,193; green,193; blue,193}, font={\footnotesize\boldmath}, tikzit category=zh]
\tikzstyle{W node}=[fill=black, draw=black, shape=regular polygon, regular polygon sides=3, minimum size=2mm]
\tikzstyle{Z dot (zw)}=[fill=white, draw=black, shape=circle, minimum width=1.2mm, inner sep=0pt]
\tikzstyle{W tri}=[fill=black, draw=black, shape=isosceles triangle, inner sep=0pt, minimum size=3.5mm, isosceles triangle apex angle=60, rotate=180]

\tikzstyle{hadamard edge}=[-, dashed, dash pattern=on 2pt off 0.5pt, thick, draw={rgb,255: red,68; green,136; blue,255}]
\tikzstyle{Q edge}=[-, thick, tikzit draw=red]
\tikzstyle{box edge}=[-, dashed, dash pattern=on 2pt off 0.5pt, thick, draw={rgb,255: red,203; green,192; blue,225}]
\tikzstyle{brace edge}=[-, tikzit draw=blue, decorate, decoration={brace,amplitude=1mm,raise=-1mm}]
\tikzstyle{diredge}=[->, thick]
\tikzstyle{double edge}=[-, double, shorten <=-1mm, shorten >=-1mm, double distance=2pt]
\tikzstyle{gray edge}=[-, {gray!60!white}]
\tikzstyle{pointer edge}=[->, very thick, gray]
\tikzstyle{boldedge}=[-, line width=1.6pt, shorten <=-0.17mm, shorten >=-0.17mm]
\tikzstyle{bidir edge}=[<->, very thick, draw={rgb,255: red,191; green,191; blue,191}]
\tikzstyle{purple edge}=[->, thick, draw={rgb,255: red,225; green,117; blue,216}]
\tikzstyle{green edge}=[->, thick, draw={rgb,255: red,167; green,231; blue,137}]
\tikzstyle{orange edge}=[->, thick, draw={rgb,255: red,245; green,170; blue,63}]
\tikzstyle{blue edge}=[->, thick, draw={rgb,255: red,68; green,136; blue,255}]
\tikzstyle{any edge}=[->, thick, draw=cyan]
\tikzstyle{red edge}=[->, thick, draw={rgb,255: red,255; green,136; blue,136}]
\tikzstyle{bidiredge}=[<->, thick]
\tikzstyle{dashed diredge}=[->, dashed, dash pattern=on 1pt off 0.5pt]
\tikzstyle{bidashed diredge}=[<->, dashed, dash pattern=on 1pt off 0.5pt]

%% file: sections/intro.tex
\begin{abstract}
    Quantum computing and AI have found a fruitful intersection in the field of natural language processing. We focus on the recently proposed DisCoCirc framework for natural language, and propose a quantum adaptation, QDisCoCirc. This is motivated by a compositional approach to rendering AI interpretable: the behavior of the whole can be understood in terms of the behavior of parts, and the way they are put together. 

    For the model-native primitive operation of text similarity, we derive quantum algorithms for fault-tolerant quantum computers to solve the task of  question-answering within QDisCoCirc, and show that this is BQP-hard; note that we do not consider the complexity of question-answering in other natural language processing models. Assuming widely-held conjectures, implementing the proposed model classically would require super-polynomial resources. Therefore, it could provide a meaningful demonstration of the power of practical quantum processors.
    
    The model construction builds on previous work in compositional quantum natural language processing.  Word embeddings are encoded as parameterized quantum circuits, and compositionality here means that the quantum circuits compose according to the linguistic structure of the text. We outline a method for evaluating the model on near-term quantum processors, and elsewhere we report on a recent implementation of this on quantum hardware.  
    
    In addition, we adapt a quantum algorithm for the closest vector problem to obtain a Grover-like speedup in the fault-tolerant regime for our model. This provides an unconditional quadratic speedup over any classical algorithm in certain circumstances, which we will verify empirically in future work.
\end{abstract}

\section{Introduction}
\label{sec:intro}


Artificial intelligence permeates a wide range of areas of activity, from academia to industrial applications, with natural language processing (NLP) standing center stage. In parallel, quantum computing has seen a surge in development, and practical quantum processors are reaching scales where small-scale quantum algorithms are becoming feasible. The merging of these two fields has given rise to the young area of research on quantum natural language processing (QNLP).  Here we focus on the `first wave' of QNLP, starting around 2016 \cite{WillC, QNLP-foundations, GrammarAwareSentenceClassification, lorenz2021qnlp, kartsaklis2021lambeq}, although there are also other approaches \cite{widdows2023nearterm, widdows2024natural}.  One important feature that distinguishes that first wave from other work, and from the broad field of quantum machine learning \cite{dunjko2017machine,Benedetti_2019,Schuld2014}, is that it provides a path towards \emph{explainability} and \emph{interpretability}.


Despite all of the impressive advancements of contemporary artificial intelligence  -- which is synonymous with machine learning methods -- these advancements have been achieved by training black-box deep-learning models on large amounts of data. In order to arrive at general applicability and high performance, what is often compromised by such setups, at least in terms of model architecture, is exactly explainability and interpretability.  When things go wrong unexpectedly, you typically don't know why.  Therefore, several attempts have been put forward recently to render black-box deep-learning models more explainable and interpretable. For example, methods for \emph{post-hoc} mechanistic interpretability have recently formed an active area of research \cite{bricken2023monosemanticity}. Conversely, first-wave QNLP made use of an earlier quantum-inspired model for language \cite{CSC, teleling} that had been crafted precisely with the features of explainability and interpretability in mind. These features were achieved through \emph{compositionality}.


For us, compositionality means that the behavior of the whole can be understood in terms of the behavior of parts, and the way they are put together \cite{compositionality}.  In the case of language this includes linguistic meaning as well as linguistic structure such as grammar and coreference. An early compositional framework for natural language that combined linguistic meaning and linguistic structure was DisCoCat \cite{CSC, FrobMeanI, GrefSadr, KartSadr}. Somewhat surprisingly, this was enabled by the foundational research program of categorical quantum mechanics \cite{AC1, CKbook}: DisCoCat emerged from the observation that the category-theoretic structure of categorical quantum mechanics perfectly matched Lambek's model of linguistic structure in terms of pregroups \cite{Lambek1, LambekBook}. In this way, linguistic structure mediates how word-meanings interact in order to form phrase- and sentence-meanings -- see \cite{teleling, QNLP-foundations} for more details. Just as it is the case in categorical quantum mechanics, sentences in DisCoCat take the form of diagrams:
\blockfig{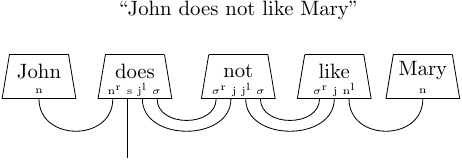}
where the boxes capture word-meanings and the wires represent the grammatical interactions between those word-meanings, yielding the meaning of the sentence as a whole. The fact that DisCoCat essentially was derived from a quantum mechanical formalism, enabled one to use it to easily craft quantum NLP models, which we will refer to as QDisCoCat \cite{WillC, QPL-QNLP, QNLP-foundations}. This was the foundation of the `first wave' of QNLP.


QDisCoCat models have been implemented on quantum hardware for the task of sentence classification \cite{GrammarAwareSentenceClassification,lorenz2021qnlp, kartsaklis2021lambeq}, but they have some shortcomings. Firstly, QDisCoCat diagrams have to be converted into quantum circuits in order to fit them on a quantum computer, and consequently, quantum algorithms derived from them require post-selection -- no lower bounds are known for the post-selection probability, so these algorithms may require exponential time in the worst case.\footnote{Sample-efficient sentence-level QNLP models can be constructed \cite{harvey2023sequence}, but they are not compositional in the above sense.} A related issue is that the reliance of DisCoCat on `categorial grammars' \cite{Lambek0, Lambek1}, which cannot natively represent texts larger than individual sentences. Both problems are addressed in the DisCoCirc framework \cite{CoeckeText, wang2023distilling}, where sentences are represented by circuits, and these circuits can be further composed to represent texts. We call these \emph{text circuits}. This bypasses the need for conversion of diagrams into circuits, making post-selection unnecessary.  Besides these, DisCoCirc enjoys a number other advantages that are discussed in \cite{wang2023distilling}, and has been studied within a wide variety of contexts \cite{duneau_solving_2021, rodatz2021conversational}.

In this work, we propose QDisCoCirc, an adaptation of DisCoCirc which provides \emph{efficient quantum algorithms} for several NLP tasks within the model. In particular, in Section \ref{sec:hardness} we show that there exists an NLP task, namely question-answering, which is BQP-hard to solve within the QDisCoCirc framework. Therefore, evaluating this model can provide a demonstration of the power of practical quantum processors, as simulating the model classically would require super-polynomial resources (assuming widely-held hardness conjectures). Furthermore, this would be less abstract compared to experiments such as those based on random circuit sampling \cite{arute_quantum_2019, decross2024computational}, which are convincing demonstrations of quantum hardware, but do not attempt to solve a widely applicable task. Importantly, we do not claim that the underlying NLP task of question-answering cannot be solved efficiently with a classical machine, only that the QDisCoCirc model in particular would be hard to evaluate in this case. 

We believe that QDisCoCirc models are worth considering because of the interpretability and explainability they afford. However, other concrete implementations of the DisCoCirc framework could also be proposed -- for example, based on classical neural networks -- and in any quantum model, it is necessary to justify why it should be preferred over a classical model, given the extreme difficulty of constructing quantum devices. We argue that a quantum model is most natural for this framework. DisCoCirc, via DisCoCat, is ultimately derived from pregroup grammars \cite{LambekBook}, where the composition of spaces is their tensor product. This has been argued as essential to fully capture the semantics of language \cite{GrefSadr,CSC}. In fact, we have the following result:

\begin{proposition}[{Informal, \cite[Theorem 2.1]{coecke2018uniqueness} \& \cite[Theorem 5.49]{CKbook}}]
    If a model has the same abstract structure as pregroup grammars, composition of spaces must be analogous to the tensor product.
\end{proposition}

This can be made precise in several ways, as discussed in \cite{coecke2018uniqueness, CKbook}. This fixes the compositional structure of our model, but leaves the underlying Hilbert space undefined. As is customary in machine learning, we take this to be numerical vectors with the usual dot product, and in this case these models become tensor networks. Since large general tensor networks are not feasible to evaluate, both for quantum and classical devices, we must pick a subset which can be efficiently evaluated. We choose to base our model on quantum circuits, which represent the largest such subset currently known \cite{bernstein1997quantum}.

Finally, in Section \ref{sec:grover} we give an algorithm that achieves for QDisCoCirc a Grover-like quadratic speedup over any classical algorithm for tasks such as text similarity, under certain technical assumptions about the dataset which can be verified empirically. This algorithm improves over the equivalent algorithm for QDisCoCat \cite{WillC}, which considers a technological regime where fault-tolerant scalable quantum computers and large quantum random access memories (QRAM) are available, while this is unlikely to be feasible in the near term. In contrast, our algorithm uses exponentially fewer bits of QRAM. Additionally, in Section \ref{sec:qdiscocirc} we present algorithms that do not use a QRAM, but are quadratically slower than the algorithm given in Section \ref{sec:grover}. They require much fewer quantum resources, and so are practical to implement on near-term machines. In Section \ref{sec:qdiscocirc}, we show how text circuits can be mapped into parameterized quantum circuits to realize these algorithms.


\section{DisCoCirc}
\label{sec:discocirc}



In the DisCoCirc framework, texts are the result of the composition of sentences, analogously to how sentences are compositions of words. Texts are represented by text circuits, which are read top to bottom. Boxes represent processes, and every box has input wires and output wires. In such circuits, only connectivity matters and so the information flow is explicit in the circuit.
A circuit is composed of \emph{generators}, or atomic processes, each of which has a linguistic interpretation -- see \cite{liu2023pipeline} for more information.
Text circuits are composed of the following set of generators, which we call \emph{states}, \emph{effects}, \emph{boxes}, and \emph{frames}:
\blockfig{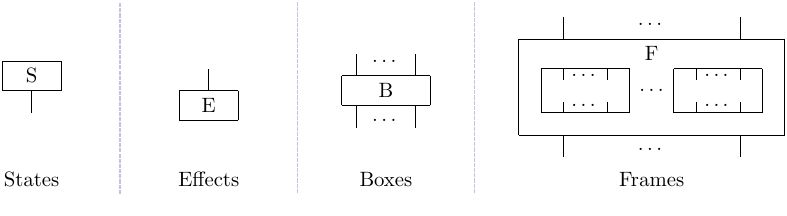}
Generators can be composed -- sequential composition of processes can be done by connecting the output wires of one process with the input wires of another process, and parallel composition is performed by placing boxes side by side.  An effect is to be understood as a `test', or the inverse, of a state. Applying an effect to a state, by joining all of their wires, results in a circuit with no open wires, which we call a \emph{scalar}. 

In this approach to modeling text as circuits, nouns are `first class citizens', and are represented as states. The idea is that noun states are carried along wires, and so all wires here carry the same type, the noun type. Boxes are processes that transform states. Examples of these are adjectives or intransitive and transitive verbs, that act on nouns and pairs of nouns respectively. In general, generators with compatible shapes can be serially composed:
\blockfig{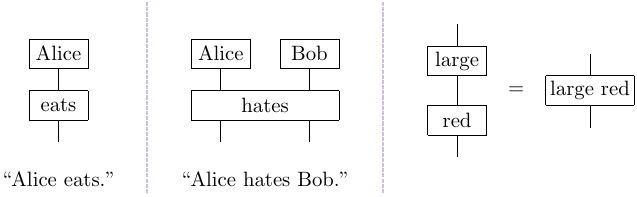}
We also have some special generators -- the identity, the swap, and the discard effect:
\blockfig{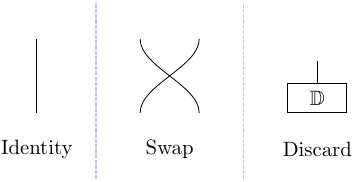}
These interact with each other in a well-behaved way \cite{CKbook}, so that \emph{only connectivity matters} -- we can bend and move wires freely so that equality of circuits is considered up to connectivity between inputs and outputs of generators. We consider any circuit to be state, effect, box, or frame if it has a compatible pattern of input and output wires. Additionally, we associate each generator with an inverse, which when composed serially yields the identity. Inverses for composite circuits are constructed by mirroring the whole circuit vertically and replacing every individual generator with its inverse.

Frames are super-maps that transform boxes, for example, intensifiers that act on adjectives, adverbs that act on verbs, or conjunctions that act on phrases:
\blockfig{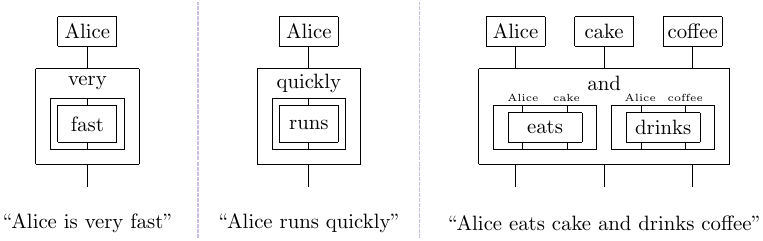}
Note that not all wires may occur in all the holes in each frame. In this case, the interior wires can be labeled with exterior wires to which they correspond (note that this need not be a single wire). Frames can also be serially composed:
\blockfig{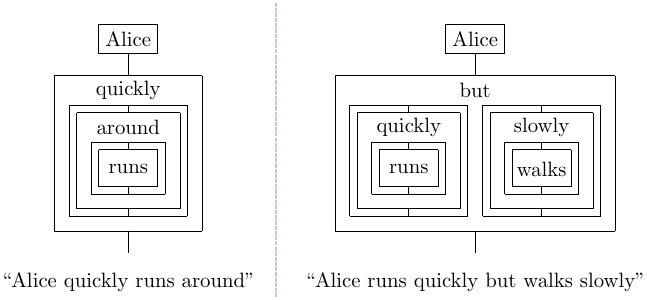}
The composition of the generators of text circuits is not free; not every composition is allowed -- for instance, wires may not pass through frames. Compositions are restricted in the sense that text circuit are the output of a text-to-circuit \emph{parser} \cite{liu2023pipeline}, which for this work we assume as given -- the parser generates parse trees for each sentence that are converted to text circuits, and the wires corresponding to common nouns are connected using coreference resolution so that generators concerning the same nouns throughout the text act on the same wires in the circuit. Arbitrary text circuits are \emph{local}, as the number of input and output wires of every box is upper bounded by a constant that is independent of the text size. Furthermore, we assume that text circuits derived from a text via the parser are acyclic.

\subsection{Native NLP Tasks}
\label{sec:native-tasks}


The range of NLP tasks is wide and diverse. Here, we consider tasks that are \emph{native} to DisCoCirc. That is, we formulate simple NLP tasks using the compositional structure of the model.
The key idea is that text circuits constructed by the parser pipeline are states which carry the information of the text. Then this information is retrieved by constructing appropriate \emph{tests} in terms of effects. In general, the effects are constructed as (inverses of) text circuits themselves. The primitive operation that we will use to define NLP tasks is that of \emph{similarity}, or overlap, which is tested by serially composing effects onto states. In this way, we can represent each task as a single, scalar, text circuit.


First we consider the text-level analogue of sentence similarity as presented in the DisCoCat framework \cite{WillC}, ie \emph{text similarity}. Suppose $T_1$ and $T_2$ are text circuits defining states obtained from the parser, defined on the wires generated by same set of nouns $N$. We form a scalar text circuit corresponding to their similarity by serially composing one with the inverse of the other. For example: 
\blockfig{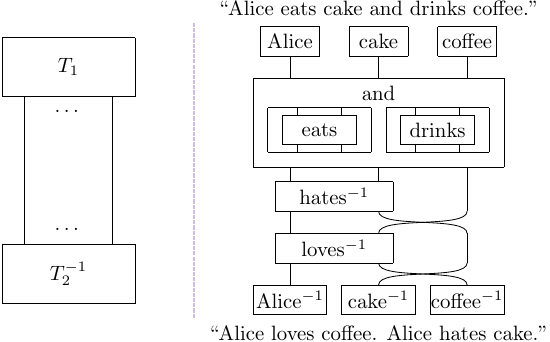}
In the case that the circuits are not defined on the same set of nouns, we can still compute the similarity. Suppose $T_1$ and $T_2$ be defined on the sets of nouns $N_1$ and $N_2$. Let $U_1, U_2, \dots, U_{|N_2 \setminus N_1|}$ be the noun states present in $T_2$ but not $T_1$, and similarly $V_1, V_2, \dots, V_{|N_1 \setminus N_2|}$ be the noun states present in $T_1$ but not $T_2$. By parallel composing $T_1$ with $U_i$ and $T_2$ with $V_i$, we form two text circuits that are defined on the same wires.  For example:
\blockfig{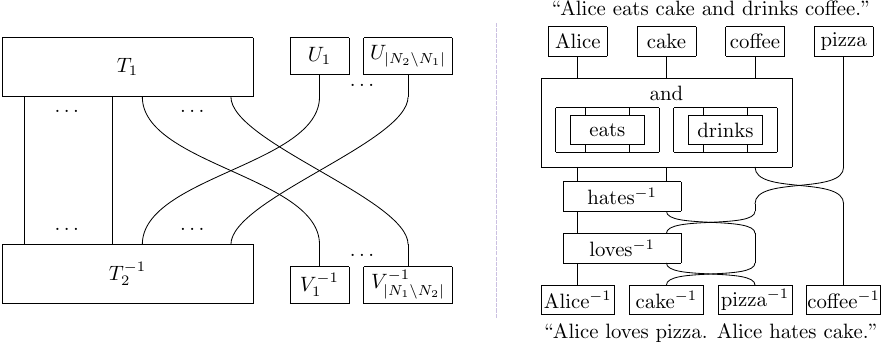}
While this task compares the whole of both texts, we can extract more specific information about a text by only comparing against a subset of the nouns that we wish to know about. We call this task \emph{question answering}. In this task, we want to ask a question about a subset of the nouns in the text. To do this, we formulate the question as a statement (for example, `Is Alice home?' becomes `Alice is home.') and then we compute the similarity of the text with the inverse of the question in the same way as above, except that we discard the noun wires in the text that are not in the subset used by the question. For example:
\blockfig{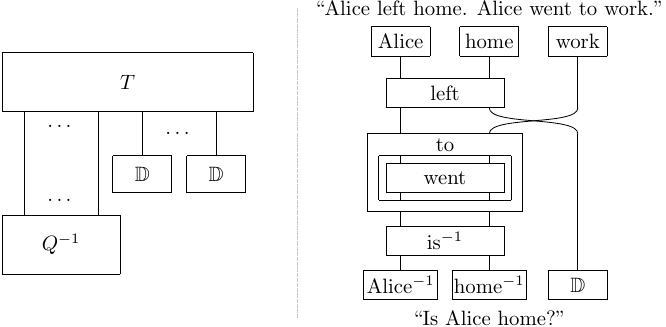}
In the case of questions which have several different possible answers, we formulate such a question statement and thus text circuit for each possible answer. For example, `Is Alice home?' becomes a pair of question statements `Alice is home.' and `Alice is not home.' Then we compare the scalar text circuits formed in each case to determine the answer to the question.

In Section \ref{sec:other-tasks}, we show a collection of other model-native NLP tasks that enjoy a natural representation as text circuits.
In general, DisCoCirc itself does not prescribe a particular semantics for any of these tasks (for instance, we leave the way that scalar circuits are computed or compared unspecified) --- it provides a method for constructing the text circuits representing the input texts, as well as an intuitive framework for designing circuits corresponding to a particular task. Moreover, once an implementation of the model is decided, the word-embeddings that form the boxes will be derived to solve a specific task. The aim of DisCoCirc is only to incorporate linguistic structure into the model, so it is not imperative that these task-circuits are defined rigorously (for example, as a statement of discourse representation theory), but rather that they are `intuitively correct' so that the imparted structure is helpful in practice. In the next section, we will construct a specific implementation of DisCoCirc that fills in these details.


%% file: sections/qdiscocirc.tex
\section{QDisCoCirc Models}
\label{sec:qdiscocirc}


To construct a concrete DisCoCirc model that performs a model-native NLP task for any potential text and task, we apply a structure-preserving map to the corresponding text circuit. This is achieved by defining a map that is applied to every generator individually, in a way that respects the relations defined on the generators. As discussed in the introduction, a natural choice of space in which such circuits may exist is tensor networks. As tensor networks are hard to evaluate in general, we propose QDisCoCirc, a family of models based on quantum circuits, which are a restricted family of tensor networks, and can be efficiently evaluated on a quantum computer.

In QDisCoCirc, text circuits are translated to quantum circuits with postselections by replacing each generator with a quantum operation or state and each wire with some fixed number of qubits. Furthermore, the inverses of generators are given by their adjoints. The generators are mapped as follows:
\blockfig{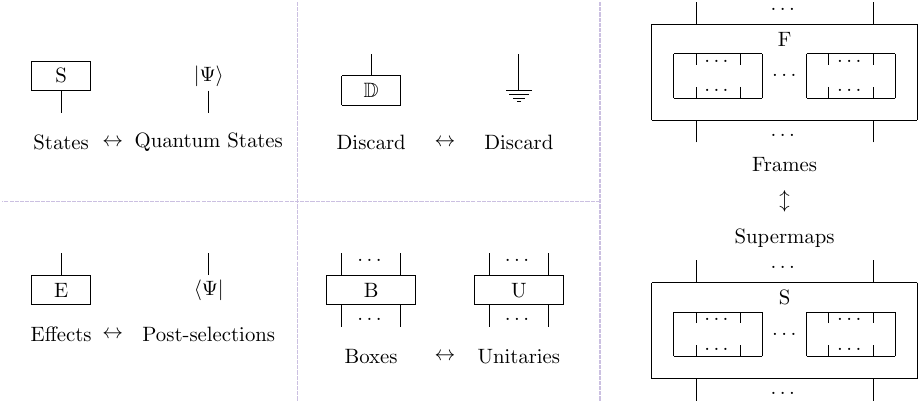}
The specific form of supermaps obtained via this mapping will be discussed shortly. We identify post-selected quantum circuit diagrams with the density matrices they represent (which are not normalized in general); thus, scalar text circuits are indeed transformed to scalars representing the probability of successfully post-selecting the given circuit.

To solve the model-native tasks defined above using this model, we observe that all the scalar text circuits for the tasks take the form of $A$ serially composed with $B^{-1}$ for some circuits $A$ and $B$ (this is true also for the tasks presented in Section \ref{sec:other-tasks}). This corresponds to a text similarity operation.
Through the quantum semantic functor, this operation maps to quantum state overlap. According to the Born rule, these result to expressions of the form $F(\Psi, \Phi) = |\langle \Psi \mid \Phi \rangle|^2$ for some quantum states $\Psi$ and $\Phi$ in the pure case, and in the general case $F(\rho, \sigma) = \tr{\rho\sigma}$ for density matrices $\rho$ and $\sigma$.

However, in the case that both $\rho$ and $\sigma$ are mixed, this does not satisfy the properties we would expect from a similarity operation, such as $F(\rho, \rho) = 1$ and $F(\rho, \sigma) < 1$ if $\rho \neq \sigma$. There are a few quantities \cite{liang2019quantum} such as the Uhlmann-Josza fidelity $\tr{\sqrt{\sqrt{\rho}\sigma\sqrt{\rho}}}^2$ which do satisfy these properties for mixed states, but unfortunately they cannot be computed efficiently using a quantum computer. Therefore, we will ensure text circuits obtained from the parser pipeline are translated to pure quantum states, so that for the tasks defined previously, \emph{at most one} of the states representing $A$ and $B$ is mixed.

Since the parser does not generate discards, we just need to make sure none of the generators are mapped to mixed operations. This is trivially true for states, effects, and boxes, but frames are more problematic.
Since states map to quantum states and boxes to unitaries, frames become maps between unitaries -- supermaps -- by analogy. However, some frames in text circuits may contain more wires for one of their arguments than are input and output to the frame -- this arises in many cases, for instance in reflexive structures. Furthermore, general quantum supermaps can be written as pairs of unitaries conjugating their arguments, along with a side channel consisting of some auxiliary qubits. In either case, we cannot use auxiliary qubits or wires because these would then be subsequently discarded, making the state mixed. Therefore, we instead constrain define our mapping from text circuits such that these situations do not occur, see Appendix \ref{sec:frame-decomp} for details.

\subsection{Solving Model-Native Tasks}
\label{subsec:qdisco-native}

After decomposing the frames in a text circuit, the resulting QDisCoCirc quantum circuit will contain only quantum states and unitaries labeled with the words from which they originated. For each task, we assume we are given these states and unitaries as quantum circuits, which we call \emph{quantum word embeddings}, such that the circuits defined for the task are evaluated as expected. We will consider in future work how to pick word embeddings such that this is the case. In particular, for text-similarity tasks, that the degree of semantic similarity between the two texts is indicated by $F(\rho, \sigma) = \tr{\rho\sigma}$. We can solve tasks based on this primitive using the \emph{swap test} \cite{buhrman2001quantum}:

\begin{theorem}
    The text-similarity primitive of QDisCoCirc models between two text circuits $A$ and $B$ can be approximated to precision $\epsilon$ in polynomial time on a quantum computer. In particular, given $\epsilon > 0$ and $\rho$, $\sigma$ the quantum states derived from $A$ and $B$ respectively, the quantity $\tr{\rho\sigma}$ can be determined to precision $\epsilon$ with failure probability $\delta$ by a quantum circuit of size
    $$O\left(\frac{(|A| + |B|)\log(\frac{1}{\delta})}{\epsilon^2}\right)$$
    where $|A|$ and $|B|$ are the number of elementary quantum gates comprising $\rho$ and $\sigma$.
\end{theorem}

\begin{proof}
    We construct a quantum circuit based on the swap test \cite[Figure 1]{buhrman2001quantum}. This gives a circuit of size $O(|A| + |B|)$ for which the first qubit has measurement probability $\frac{1}{2} + \frac{\tr{\rho\sigma}}{2}$. By a Hoeffding bound, $O(\log(\frac{1}{\delta})\epsilon^{-2})$ samples of this circuit suffice to estimate $\tr{\rho\sigma}$ to precision $\epsilon$ with failure probability $\delta$.
\end{proof}

To solve the question-answering task in the QDisCoCirc model, we are given a text $T$ and a set of questions $Q_i$ and wish to find the $Q_i$ which maximizes $\tr{\rho_T, \rho_{Q_i}}$. We can do this by evaluating $\tr{\rho_T, \rho_{Q_i}}$ for each $i$ using the method above, and then taking the maximum of these. We show below that this also runs in polynomial time on a quantum computer. To compare this algorithm to classical algorithms, it is important to consider the scaling of $\epsilon$ carefully. The best-known classical method to calculate $\tr{\rho\sigma}$ would at least require performing statevector simulation of the quantum circuit. In this case, we would expect a total time of at least $O((|A| + |B|)N^n\operatorname{polylog}(\epsilon^{-1}))$ and memory requirement of $O(N^n\log(\epsilon^{-1}))$, where $\log(N)$ is the number of qubits per wire, and $n$ is the number of wires in the text circuits -- this is an exponentially worse dependency on $N$ and $n$. 

On the other hand, the scaling in terms of $\epsilon$ is much better in the classical case. For instance, consider question answering. If the number of nouns referenced in the question is constant $c$, then only $c\log(N)$ qubits will not be discarded. In this case, if the quantum states corresponding to the questions are sufficiently uniformly distributed then $\epsilon = O(N^{-c})$ should suffice to resolve the maximum similarity, by an argument of Wiebe et al \cite{wiebe2015quantum}, which gives exponentially better dependence in $n$. However, considering the case of full text-text similarity with no discarding, the quantities $\tr{\rho\sigma}$ may grow as small as $N^{-n}$ in general, and so if small relative error is needed, $\epsilon = O(N^{-n})$ may be required, which makes the complexity at least as bad as the classical algorithm. Therefore, for any given dataset and word embeddings, it would be important to analyze the scaling of $\epsilon$ in terms of $N$ and $n$, either analytically or empirically.

%% file: sections/hardness.tex
\section{Hardness for Question Answering}
\label{sec:hardness}



We will now consider the problem of question answering in the QDisCoCirc framework more formally. In particular, we will show that for the appropriate choices of texts and word embeddings, question answering is BQP-hard, and that there is an efficient quantum algorithm to solve it. Note these results only apply to question answering strictly as specified by the QDisCoCirc framework, not as a general NLP task. We call this task \textsc{QDisCoCirc-QA} to emphasize the distinction. Proofs for the theorems in this section are given in Appendix \ref{sec:proofs}. 


\begin{table}
    \centering
    \begin{tabular}{|r|ccc|}
        \hline
        \diagbox{Embeddings}{Texts} & \emph{Worst-Case} & \emph{Average-Case} & \emph{Typical} \\
        \hline
        \emph{Worst-Case} & BQP-hard (Thm.\ref{worstworst}) & BQP-hard (Thm.\ref{averagetexts}) & Empirical only \\
        \emph{Average-Case} & BQP-hard (Thm.\ref{haarrandom}) & Unknown & Empirical only \\
        \emph{Typical} & Testable (Thm.\ref{typicalembeddings}) & Unknown & Empirical only \\
        \hline
    \end{tabular}
    
    \caption{The hardness results shown in this paper. The embeddings axis refers to the choice of word embeddings $V$, whereas the texts axis refers to the choice of context $T$ and questions $\{Q_i\}$. `Average-case' hardness refers to a `natural' distribution - for embeddings this is the Haar distribution, whereas for texts we define a distribution. `Typical' hardness refers to texts and embeddings taken from actual human-generated data. Since typical texts are always bounded by definition, only empirical evidence can be given for hardness in this case.\vspace{-0.2cm}}
    \label{tab:hardnessresults}
\end{table}

\begin{definition}
    The problem \textup{\textsc{QDisCoCirc-QA}} is defined as follows: given a set of word embeddings $V$, a context text $T$, and a set of $k$ question texts $\{Q_i\}$, determine any $j$ such that
    $$ \left|\tr{\rho_T(\rho_{Q_j} \otimes I)} -  \max_{i}\tr{\rho_T(\rho_{Q_i} \otimes I)}\right| < \epsilon $$
    where $\rho_T = U_T\ket{0}\bra{0}U^\dagger_T$, $\rho_{Q_i} = U_{Q_i}\ket{0}\bra{0}U^\dagger_{Q_i}$, and $U_T, U_{Q_i}$ are the QDisCoCirc text circuits generated from $T$ and $Q_i$ respectively over $V$.
\end{definition}

\begin{reptheorem}{qabqp}
    \textup{\textsc{QDisCoCirc-QA}} can be solved on a quantum computer in time
    $$ O\left(\frac{k\log\left(\frac{k}{\delta}\right)|V|_w(|T| + \max_i|Q_i|)}{\epsilon^2} \right) $$
    with precision $\epsilon$ and failure probability $\delta$, where $|V|_w$ is the maximum number of elementary gates of any word embedding in $V$, and $|T|$, $|Q_i|$ are the number of elementary gates in $U_T$ and $U_{Q_i}$ respectively.
\end{reptheorem}

Note that if the separation between the values $\tr{\rho_T(\rho_{Q_i} \otimes I)}$ is smaller than $\epsilon$, then the formulation of \textsc{QDisCoCirc-QA} does not guarantee that the `correct' answer (in the sense of maximum similarity) will be returned. We assume that word embeddings exist such that $\epsilon$ need not be too small, as discussed in Section \ref{sec:qdiscocirc}.
Now we tackle the hardness results in the first column of Table \ref{tab:hardnessresults} - that is, we will show that for some sets of word embeddings $V$ it is possible to construct context texts $T$ and questions $\{Q_i\}$ for which \textsc{QDisCoCirc-QA} is BQP-hard.

\begin{reptheorem}{worstworst}
    Suppose that a set of word embeddings $V$ satisfies the following:
    \begin{enumerate}
            \item The operations of $V$ use one qubit for each input wire,
            \item $V$ contains arbitrarily many proper nouns,
            \item $V$ contains at least two adjectives that generate a dense subset of $SU(2)$,
            \item $V$ contains at least one transitive verb that is entangling
    \end{enumerate}
    then for any fixed $\epsilon < \frac{1}{7}$, \textup{\textsc{QDisCoCirc-QA}} is BQP-hard. 
\end{reptheorem}

Supposing that we have word embeddings with enough nouns (and thus enough qubits) to perform the reduction in the previous theorem, the particular unitaries in the word embeddings aren't important -- Haar-random word embeddings suffice with high probability. This shows that \textsc{QDisCoCirc-QA} is hard in the average-case over word embeddings.

\begin{reptheorem}{haarrandom}
    Given a set of word embeddings $V$, suppose that operations in $V$ are independent Haar-random unitaries. Then conditions three and four of Theorem \ref{worstworst} are almost surely satisfied for all word embeddings containing at least two adjectives and one transitive verb.
\end{reptheorem}

Moreover, suppose we are given a set of word embeddings (for instance, that can solve a particular task), then we give a method to test if \textsc{QDisCoCirc-QA} is hard in this case.

\begin{reptheorem}{typicalembeddings}
    Given a specific set of word embeddings $V$ it is possible to check numerically that the conditions of Theorem \ref{worstworst} are satisfied. 
\end{reptheorem}

To show average-case hardness of \textsc{QDisCoCirc-QA} in terms of texts (second column in Table \ref{tab:hardnessresults}), we will first define a distribution over QDisCoCirc circuits of a given size, and then adapt Theorem \ref{worstworst} to show that we can draw polynomially sized texts $T$, $\{Q_i\}$ from this distribution and construct a set of word embeddings $V$ such that solving the corresponding question answering instance is BQP-hard. Note that this distribution is not particularly representative of human-generated texts.

\begin{definition}
    For every size $k$, we define a distribution $\mathcal{D}_k$ of text circuits with $k$ unique states:
    \begin{itemize}
        \item There are no frames, and there are $\Omega(k^{\gamma})$ boxes with high probability for some $\gamma > 4$.
        \item The label of each box is drawn independently from a distribution of words that obey Heap's law and Zipf's law exactly.
        \item The number of inputs to each box is drawn independently from a distribution $A$ such that $P(A = 2)$ is nonzero and each input to a box is selected uniformly randomly from all possible wires.
    \end{itemize}
\end{definition}

We are now ready to show the reduction. However, to make this work we need the number of words that only occur once in a given text (\emph{hapax legomena}) to be bounded below by a polynomial in the total size of the text. Clearly, this cannot hold for all sizes for any finite vocabulary. For typically-sized texts, this assumption is justified statistically by Heap's law and Zipf's law \cite{Lin_2017,Sano2012} --- it has been shown that for large corpora, roughly half of the vocabulary are hapax legomena, and hence their number scales like $\propto N^k$ for some $k \approx 0.5$ \cite{baayen2001word}. However, since we are considering the limit of arbitrarily-sized texts which do not exist in practice (all human texts are bounded), we require that this behavior can be extrapolated.

\begin{reptheorem}{averagetexts}
    There exists a polynomial $f(m, n)$, such that given an oracle to solve \textup{\textsc{QDisCoCirc-QA}} instances with arbitrary word embeddings and text circuits drawn from $\mathcal{D}_{f(m, n)}$, we can perform arbitrary quantum computations with $m$ $2$-local gates on $n$ qubits with high probability in polynomial time. That is, \textup{\textsc{QDisCoCirc-QA}} for worst-case word embeddings is average-case BQP-hard over texts drawn from $\mathcal{D}_{f(m, n)}$.
\end{reptheorem}

Since \textsc{QDisCoCirc-QA} is BQP-hard, we expect that it is not easy to solve with a classical computer, and that quantum algorithms to solve this ought to provide super-polynomial speedups. Therefore, supposing that (a) there exists a set of word embeddings for which a QDisCoCirc model is accurate on a dataset of interest, (b) the resulting quantum circuits are not easy to simulate classically (in the sense of Theorem \ref{typicalembeddings}), (c) a quantum device exists with sufficient fidelity to provide meaningful results, and (d) the precision $\epsilon$ required is not too high (in the sense of Section \ref{subsec:qdisco-native}), then it would be possible to perform an experiment on a quantum device which is hard to simulate classically but solves a task of interest. However, it is important to note that for many datasets this task can be solved classically by other means (notwithstanding the arguments presented in favor of QDisCoCirc in Section \ref{sec:intro}), so this would \emph{not} constitute a demonstration of quantum advantage over \emph{every} classical algorithm.

%% file: sections/grover.tex
\section{Further Quadratic Speedups}
\label{sec:grover}




In this section, we present a quantum algorithm for obtaining polynomial speedups for the task of text question answering that we presented in Sections \ref{sec:qdiscocirc} and \ref{sec:hardness}.
Specifically, our results are obtained by adapting and improving on the approach of the work of Zeng and Coecke \cite{WillC}. These algorithms provide a quadratic improvement over the simpler swap test algorithm, as well as being quadratically faster than any classical algorithm for solving the same task. While they require a QRAM \cite{giovannetti2008quantum}, it is used to store exponentially fewer bits as compared to \cite{WillC}.

Basheer et al \cite{basheer2021quantum} give a quantum algorithm for accelerating the $k$-nearest neighbors problem as follows. Given the vectors $v_0 \in \mathbb{C}^N$ and $\{v_1, \dots, v_k\} \subset \mathbb{C}^N$ such that $||v_i|| = 1$ and quantum amplitude encoding oracles $\mathcal{V}\ket{0} = \ket{v_0}$ and $\mathcal{W}\ket{i}\ket{0} = \ket{i}\ket{v_i}$ with gate complexities $T_V$ and $T_W$, then there is a quantum algorithm to find $i = \arg\max_{j \geq 1} | \langle v_0\mid v_j\rangle |$ with error at most $\epsilon$ and fixed positive success probability in time $\tilde{O}\left(\sqrt{k}\epsilon^{-1}(T_V + T_W + \log^2{N}) \right)$.
The algorithm takes quantum states, encodes the fidelity via a partial swap test, then digitizes this with the quantum analog-to-digital conversion algorithm of Mitarai et al \cite{mitarai2019quantum} before finding the minimum via D\"urr-H\o{}yer search \cite{durr1999quantum}. In the case that $\mathcal{V}(\ket{0}\bra{0}) = \rho_0$ and $\mathcal{W}(\ket{i}\bra{i} \otimes \ket{0}\bra{0}) = \ket{i}\bra{i} \otimes \rho_i$ are actually operators preparing mixed states, then the quantity calculated will be $i = \arg \max_{j \geq 1} \tr{\rho_0 \rho_i}$.

\subsection{Application to DisCoCirc}

Since DisCoCirc is based on a circuit architecture, it is natural to work with the \emph{word embeddings} described as \emph{quantum circuits}, and construct the amplitude encoding of texts directly as quantum circuits. We will apply the nearest-neighbor algorithm to the following task: given a text $T_0$ and a set of texts $T_{i}$ for $k \geq i > 0$, find $i$ such that $T_0$ and $T_i$ are the most similar, and thus that maximizes $\tr{\rho_0 \rho_i}$, where $\rho_i$ is the density matrix corresponding to the QDisCoCirc circuit for $T_i$. This extends both the question-answering and text-similarity tasks given in Section \ref{sec:qdiscocirc}.

Unlike the previous section where we assumed that the boxes in QDisCoCirc were arbitrary unitaries, we now suppose that each box with $d$ wires is a particular instance of a parameterized quantum circuit ansatz. Let us define the following hyperparameters: (a) each noun wire is mapped to a set of qubits with Hilbert space dimension $N$, (b) each QDisCoCirc circuit considered has at most $w$ unitaries and $n$ noun wires total, (c) each unitary in the quantum circuit acts on at most $d$ wires, (d) there are $V$ different word embeddings given, (e) the number of gates in the ansatz for $d$ wires of dimension $N$ scales as $A_g(N, d)$, (f) the number of parameters in the ansatz for $d$ wires of dimension $N$ scales as $A_p(N, d)$, and (g) each parameter for the ansatz is stored with $P$ bits of precision.

To solve this task, we can straightforwardly apply the quantum $k$-nearest neighbors algorithm discussed above, with $\mathcal{V}(\ket{0}\bra{0}) = \rho_0$ and $\mathcal{W}(\ket{i}\bra{i} \otimes \ket{0}\bra{0}) = \ket{i}\bra{i} \otimes \rho_i$. Clearly, $\mathcal{V}$ is given by the quantum circuit preparing $\rho_0$, however for $\mathcal{W}$ we need a significantly more complicated circuit. This is explicitly constructed in Appendix \ref{sec:oracle}, and requires the following QRAM oracles:
\begin{itemize}
    \item $\mathrm{Param}_{i j}\ket{k}\ket{0} = \ket{k}\ket{\phi_{i j k}}$ where $\phi_{i j k}$ is a $P$-bit fixed-precision binary encoding of the $j$th parameter for the ansatz of the $i$th unitary in the circuit preparing $\rho_k$,
    \item $\mathrm{Width}_i \ket{k}\ket{0} = \ket{k}\ket{W_{i k}}$ where $W_{i k}$ is a binary encoding of the width of the $i$th unitary in the circuit preparing $\rho_k$,
    \item $\mathrm{Index}_{i j}\ket{k}\ket{0} = \ket{k}\ket{I_{i j k}}$ where $I_{i j k}$ is an encoding of the position of the $j$th argument to unitary $i$ in the circuit preparing $\rho_k$. Note that the encoding is not a mapping from zero-indexed position to binary, but it is of size $O(\log(w))$ -- see Appendix \ref{sec:oracle}.
\end{itemize}
Each of these oracles can be implemented using a series of calls (one for each bit) to bucket-brigade QRAMs of size $k$. Therefore, $O(kw(PA_p(N, d) + \log(d) + d\log(w)))$ bits of QRAM is required in total, each taking time $O(\log(k))$ to query. 

If we assume that $A_p(N, d) = O(\operatorname{poly}(\log(N), d))$ (this is true for any of the widely-used ansatze defined by Sim et al \cite{Sim_2019}, for example), that $w = O(\operatorname{poly}(n))$ (which is justified by Heap's law), and that $P$ and $d$ are fixed, then this is $O(k\operatorname{poly}(n, \log(N)))$ bits total. By comparison, adapting the strategy employed in \cite{WillC} to QDisCoCirc would require at most $O(kN^n)$ bits total, so our approach provides a significant saving.  

The oracle for $\mathcal{W}$ is constructed by alternating layers of multiplexers, which bring the qubits for each unitary to the top of the circuit, and parameterized ansatze to implement each unitary. Each of these load parameters from QRAM to specify their behavior. The overall construction is as follows:
\blockfig{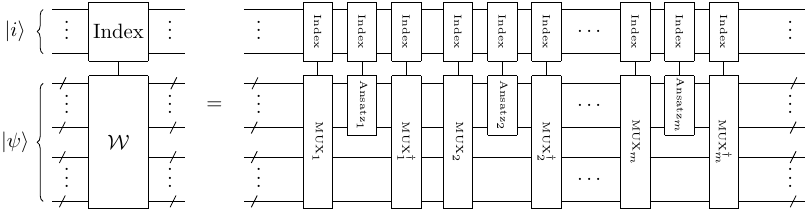}
Using this, we derive the following result:

\renewcommand{\thetheorem}{7}
\begin{theorem}
    Given a text $T_0$ and a set of texts $T_{i}$ for $k \geq i > 0$, we can find $i$ that maximizes $\tr{\rho_0 \rho_i}$, with error at most $\epsilon$, in time
    $$\tilde{O}\left(\frac{\sqrt{k}}{\epsilon} (w d \log(k) (n \log(N) + P A_g(N, d)) + \log^2(N))
    \right)$$
    on a quantum computer with a constant success probability, where $\rho_i$ is the density matrix corresponding to the QDisCoCirc quantum circuit for $T_i$.
\end{theorem}

\begin{proof}
    We apply the quantum $k$-nearest neighbor algorithm of Basheer et al \cite{basheer2021quantum} with $\mathcal{V}$ and $\mathcal{W}$ defined above. These can be implemented with $O(A_g(N, d))$ and $O(w d \log(k) (n \log(N) + P A_g(N, d)))$ gates respectively, see Appendix \ref{sec:oracle} for details. We require controlled versions of these, but this only increases the gate count by a constant factor.
\end{proof}

Comparing this to the naive swap test-based algorithm for solving the question-answering task defined in Section \ref{sec:hardness}, it is clear that this method achieves a quadratic speedup both in $k$ and $\epsilon$, up to extra logarithmic factors needed to construct the oracles. Comparing against a classical algorithm, note that for any given $\epsilon$, this algorithm is better than any classical algorithm, which clearly must be at least $O(k)$, and so for large values of $k$ we would expect to see a speedup as compared to any classical algorithm to solve the same task. However, as discussed in Section \ref{sec:qdiscocirc}, the classical algorithm scales much better in terms of $\epsilon$, namely $O(\log(\frac{1}{\epsilon}))$ vs $O(\epsilon^{-1})$, and so the same concerns about the dependence on $N$ and $n$ apply here. Additionally for this algorithm, for a given dataset, it is necessary to analyze the scaling of $\epsilon$ in terms of $k$ -- for instance, if $\epsilon = o(k^{-\frac{1}{2}})$ is required, then there would be no quadratic speedup in $k$ compared to the classical algorithm.

%% file: sections/other-tasks.tex
\section{Other Model-Native Tasks}
\label{sec:other-tasks}

In the previous sections, we have considered only the fundamental primitive of text-similarity, and its local version that we call question-answering. However, we can also propose some simple model-native text processing tasks regarding `character development'. All tasks we introduce here have a natural representation as text circuits in the DisCoCirc framework.

To do so, we introduce two new generators: the cap and cup. The cap semantically represents a joint state of two nouns which are unspecified but identical. Likewise, the cup is the corresponding effect that tests for this situation.
\blockfig{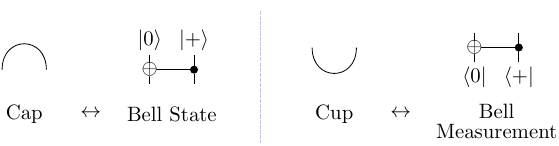}
Their translations into QDisCoCirc are naturally defined as the Bell state and measurement respectively. These generators are both self-inverse, and obey certain relations along with the identity and swap generators to ensure that they are well-behaved \cite{CKbook} - in particular, to maintain the property that only connectivity of the diagram matters.

We can use these to construct several tasks. For example, measuring a `character arc' --- comparing the similarity of a noun state wire before and after it has been acted on by a text, either with itself or between a pair of nouns. We can also form a diagrammatic trace using caps, cups, and identity wires, which we can use to measure a similar quantity in a way that is independent of the initial noun state (i.e. how much \emph{any} noun changes if it was taking that particular role in the text). Let $T$ be a text and $N_1$ and $N_2$ be two nouns present in it, then we can form the following text circuits:
\renewcommand{\figscale}{0.95}\blockfig{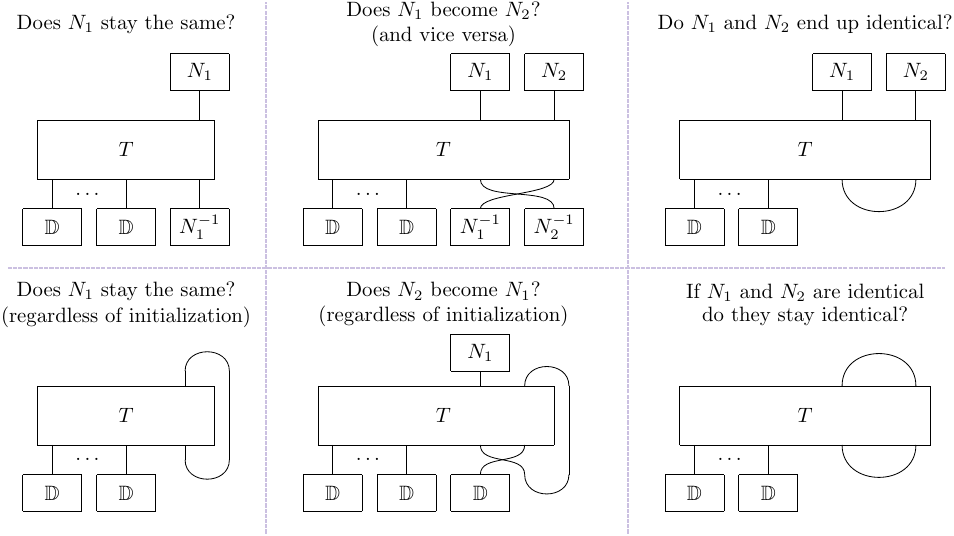}\renewcommand{\figscale}{1.0}

Further, we can have combinations of the above, by applying noun effects and traces in the same circuit. This is to be understood as testing for properties in conjunction. From this, we see the motivation for including cups, caps, and inverses in our framework -- it allows us to represent more intuitively some natural classes of tasks as a single text circuit. Additionally, while the results in Sections \ref{sec:hardness} and \ref{sec:grover} are tailored to the sentence-similarity and question-answering tasks, we would expect that similar theorems hold for all the tasks defined here as well. 

%% file: sections/conclusion.tex
\section{Discussion and Future Work}\label{sec:outlook}

In this paper we focused on the DisCoCirc framework because of its explainability and interpretability in compositional terms. For the primitive task of question-answering within the QDisCoCirc model, we showed BQP-hardness. Further, we adapted known quantum algorithms to achieve for QDisCoCirc a Grover-like quadratic speedup over any classical algorithm for tasks such as text similarity, under certain technical conditions.

Recalling that the way in which we constructed QDisCoCirc is by applying a map from text circuits to quantum circuits, other variants of this mapping could be explored in the future. For example, there are families of quantum circuits that are \emph{classically efficiently simulable}, such as those constructed from nearest-neighbor match-gates \cite{Jozsa_2008}, which have been used successfully to numerically optimize large quantum circuits \cite{Matos_2023} and could be used as `warm start' initializations for QDisCoCirc models. Viewing quantum circuits as special cases of tensor networks, we can also consider models based on tensor networks (although they would be costly to compute). This would constitute compositional tensor-network machine learning, which is an emerging subfield in machine learning\cite{Huggins_2019,tomut2024compactifai}. 

Overall, our approach can be viewed as providing sparse structural priors to NLP models, potentially aiding generalization performance or allowing more efficient training. Such an approach could also provide a structure that may guide how mechanistic interpretability methods can be applied \cite{liu2023seeing}. Most importantly, we prioritized transparent model-building, with performance and broad applicability as a secondary target.

The results here are presented in a way compatible with the parser introduced by Liu, Shaikh, Rodatz, Yeung, and Coecke \cite{liu2023pipeline}, which mechanically transforms a subset of English text into DisCoCirc text circuits. Given a different parser, the definition of the mapping from text circuits to quantum circuits might need modification. However, the results of Section \ref{sec:hardness} and \ref{sec:grover} regarding quantum algorithms would still be applicable. These results concern fault-tolerant quantum computers (either via deep quantum circuits, or by the assumption of quantum random access memories). 

Focusing on more near-term hardware, and given access to a parser and state-of-the-art near-term quantum processors \cite{QuantumVolumeQuantinuum}, allows us to perform heuristic explorations based on the algorithm in Section \ref{subsec:qdisco-native}.  In a separate paper \cite{experiment}, we report on performing tasks such as question-answering within the noisy intermediate-scale quantum regime. We train the quantum word embeddings that encode the semantics of texts in-task, as was done in previous work with DisCoCat \cite{GrammarAwareSentenceClassification,lorenz2021qnlp}. However, it is also promising to experiment with training these in a task-agnostic way.  

Furthermore, QDisCoCirc allows us to \emph{put to the test} the principle of \emph{compositionality} with respect to training, i.e.~to quantify to what extent a compositional model aids generalization to larger instances after being trained on instances up to a certain size. This setup is motivated by the possibility that a compositional approach to quantum machine learning with parameterized circuits bypasses the fundamental problem posed by barren plateaus \cite{ragone2023unified}. Our results in \cite{experiment} are very promising in this respect.

%% file: sections/frame-decomp.tex
\section{Decomposing Frames}
\label{sec:frame-decomp}

Our semantic functor maps
states, effects, and boxes to pure quantum operations trivially.
However, frames, which are mapped to quantum supermaps, in general, need some special treatment.
Some frames in text circuits may contain more wires for one of their arguments than are input and output to the frame -- this arises in many cases, for instance in reflexive structures. We cannot represent these extra wires as auxiliary qubits because these would then be subsequently discarded, making the state mixed.
In general, quantum supermaps can be written as pairs of unitaries conjugating their arguments, along with a side channel consisting of some auxiliary qubits. However, this conflicts with our no-ancilla requirement.

Therefore, we will instead \emph{constrain} our representation of frames as follows:
\begin{itemize}
    \item We lay the frame out vertically, with unitaries before the first argument, between every argument, and after the last argument.
    \item For every wire in each argument, we map it to one of the wires of the unitary: if that wire corresponds to an unused noun wire that is input to the frame, we map it to that wire. Otherwise, we map it to any unused noun wire. If there are none, then we delete the wire.
    \item After this mapping, we substitute the argument between the unitaries of the frame, and any unused wires are connected directly between the unitaries of the frame to serve as side-channels.
\end{itemize}
In order to resolve which wires correspond to which nouns (and thus where to assign each wire of the frame arguments), we rely on coindexing information provided by the parser, which within sentences is provided by finding the grammatical 'head' of each phrase, and between sentences is provided by a coreference tool \cite{liu2023pipeline}. To minimize the number of wires that are deleted, we perform this mapping for each argument independently, not respecting the wire assignments in previous arguments. This may create boxes that are ungrammatical, but we assign them unitaries in the same way as other boxes, so this is not an issue.
\renewcommand{\figscale}{0.95}\blockfig{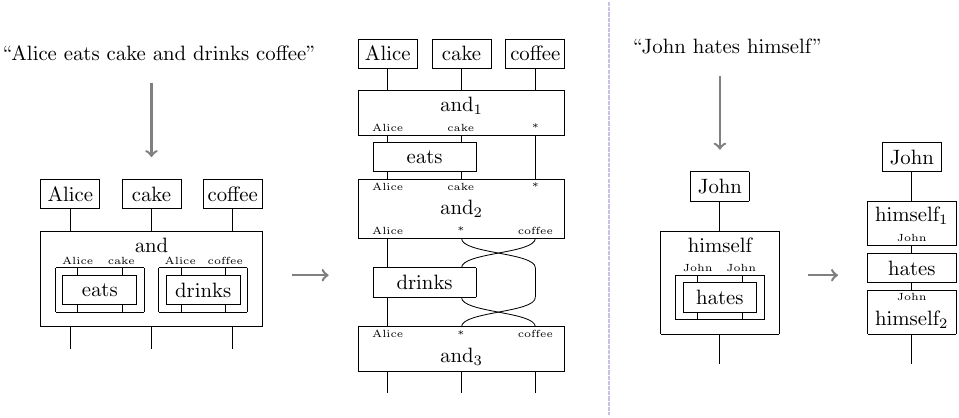}\renewcommand{\figscale}{1.0}
This illustrates two examples of the sandwich construction. In each, the wires inside the unitaries making up the frame are labeled with which nouns they are representing, and wires that are being used as side channels are labeled with `*'. In the first example, we see that swaps may be introduced to map wires to the correct locations. This is an artifact of the notation and does not appear in the compiled quantum circuits. In the second example we see a wire being deleted since there are no unused wires to assign it to, and in the process this creates an ungrammatical box for `hates', since it is a transitive verb.

%% file: sections/proofs.tex
\section{Proofs of Theorems}
\label{sec:proofs}

\renewcommand{\thetheorem}{2}
\begin{theorem}
    \label{qabqp}
    \textup{\textsc{QDisCoCirc-QA}} can be solved on a quantum computer in time
    $$ O\left(\frac{k\log\left(\frac{k}{\delta}\right)|V|_w(|T| + \max_i|Q_i|)}{\epsilon^2} \right) $$
    with failure probability $\delta$, where $|V|_w$ is the maximum size of any word embedding in $V$.
\end{theorem}

\begin{proof}
    To calculate each
    $$ p_j = \tr{\rho_T(\rho_{Q_j} \otimes I)} $$
    we produce a circuit using the swap test where the probability of measuring a zero on the first qubit is $\frac{1}{2} + \frac{p_j}{2}$. Let $\mu_{j, n}$ be the estimate of $p_j$ after $n$ shots. By Hoeffding's inequality, 
    $$ P(|\mu_{j, n} - p_j| \geq \epsilon') \leq 2\exp\left(-\frac{n\epsilon'^2}{2}\right) $$
    thus if we take $n = \frac{2}{\epsilon'^2}\log(\frac{2k}{\delta})$, then we have
    $$ P\left(\bigwedge_{j = 1}^k|\mu_{j, n} - p_j| \leq \epsilon'\right) \geq \left(1 - \frac{\delta}{k}\right)^k \geq 1 - \delta$$
    by Bernoulli's inequality. Note that if $|\mu_{j, n} - p_j| \leq \epsilon'$ for all $j$, then $|\max_i\mu_{n, i} - \max_ip_i| \leq \epsilon'$, and hence if $j = \arg\max_i \mu_{n, i}$ then $|p_j - \max_i p_i| \leq 2\epsilon'$. Therefore, if $\epsilon' = \frac{\epsilon}{2}$, then this $j$ correctly solves the \textsc{QDisCoCirc-QA} instance with probability $1 - \delta$. For each $j$, each shot of the swap test circuit takes $O\left((|T| + |Q_j|)|V|_w\right)$ time since $U_T$ and $U_{Q_j}$ contain $O(|T||V|_w)$ and $O(|Q_j||V|_w)$ gates respectively, and the width of $U_{Q_j}$ must be at most $2|Q_j||V|_w$. Therefore, we have an overall time complexity of:
    $$ O\left(\sum_i (|T| + |Q_i|)|V|_w n\right) = O\left(\frac{k\log\left(\frac{k}{\delta}\right)|V|_w(|T| + \max_i|Q_i|)}{\epsilon^2} \right)$$
\end{proof}

\begin{definition}[\textsc{Approx-QCircuit-Prob}]
    Suppose we are given a quantum circuit $C$ on $n$ qubits with $m$ 2-local gates, where $C$ is taken from a uniform family such that $m$ is polynomial in $n$, along with the promise that either
    \begin{itemize}
        \item measuring the first qubit $C\ket{0}$ has probability at least $\frac{2}{3}$, or
        \item measuring the first qubit $C\ket{0}$ has probability at most $\frac{1}{3}$.
    \end{itemize}
    Then the \textup{\textsc{Approx-QCircuit-Prob}} problem is to distinguish between these two cases. This is BQP-hard.  
\end{definition}

\renewcommand{\thetheorem}{3}
\begin{theorem}
    \label{worstworst}
    Suppose that a set of word embeddings $V$ satisfies the following:
    \begin{enumerate}
            \item The operations of $V$ use one qubit for each input wire,
            \item $V$ contains arbitrarily many proper nouns,
            \item $V$ contains at least two adjectives that generate a dense subset of $SU(2)$,
            \item $V$ contains at least one transitive verb that is entangling
    \end{enumerate}
    then for any fixed $\epsilon < \frac{1}{7}$, \textup{\textsc{QDisCoCirc-QA}} is BQP-hard. 
\end{theorem}

\begin{proof}
    Let $U_1$ and $U_2$ be the adjectives defined by condition three, and $U_3$ be the transitive verb identified by condition four. Then, by the inverse-free Solovay-Kitaev algorithm \cite{bouland2021efficient}, any single-qubit unitary $U$ can be approximated to precision $\epsilon'$ by a sequence of $U_1$ and $U_2$ operations of length $O(\log^{8.62}(\frac{1}{\epsilon'}))$. From a result of Zhang et al \cite{zhang2003exact}, any entangling gate combined with arbitrary single-qubit gates can be used to implement any two-qubit gate exactly in a constant number of gates (where the constant depends only on the definition of the entangling gate). Hence, any two-qubit gate can also be approximated arbitrarily well by $O(\log^{8.62}(\frac{1}{\epsilon'}))$ applications of $U_1$, $U_2$, and $U_3$.

    Now we can reduce from the \textsc{Approx-QCircuit-Prob} problem to \textsc{QDisCoCirc-QA}, by constructing a QDisCoCirc circuit corresponding to the given circuit $C$. First, assign each qubit a proper noun $N_i$ with corresponding unitary $U_{N_i}$ in $V$. To each noun apply the unitary $U_{N_i}^\dagger$. Each $U_{N_i}^\dagger$ is approximated as defined above. Then, for each gate in $C$, we approximate it as discussed above and apply it to the corresponding wires. This yields a QDisCoCirc circuit $U_T$ which approximates the original circuit $C$. 
    
    From this, we can generate $T$ as follows: for each gate in the circuit, we generate a sentence. If the gate acts on a single qubit $i$, we write ``\{$N_i$\} is \{$U_k$\}.'' where $U_k$ is the adjective corresponding to the gate. If the gate acts on two qubits $i$ and $j$, we write ``\{$N_i$\} \{$U_3$\} \{$N_j$\}.'' To produce $Q_1$ and $Q_2$, we first assemble quantum circuits $U_{Q_1} = U_{N_1}^\dagger U_{N_1}$ and $U_{Q_2} = U_XU_{N_1}^\dagger U_{N_1}$ where $U_X$ is an approximation of the X gate, and $U_{N_1}^\dagger$ is also approximated as above. Then $Q_1$ and $Q_2$ can be generated in the same way as $T$.

    Let $\epsilon' \geq \frac{1}{42}$ be such that $\epsilon' + \epsilon < \frac{1}{6}$, and let $P_U(i)$ be the probability of measuring $\ket{i}$ on the first qubit of $U\ket{0}$. Suppose we solve the \textsc{QDisCoCirc-QA} instance associated with $V$, $T$, $\{Q_1, Q_2\}$ to precision $\epsilon$. Note that since $U_{Q_1}\ket{0} \approx \ket{0}$, $\tr{\rho_T(\rho_{Q_1} \otimes I)} \approx P_{U_T}(0)$. Likewise, since $U_{Q_2}\ket{0} \approx \ket{1}$, $\tr{\rho_T(\rho_{Q_2} \otimes I)} \approx P_{U_T}(1)$. As $U_T \approx C$, if we find $j = 1$, then we can conclude that $P_C(0) > P_C(1) = 1 - P_C(0) \implies P_C(0) > \frac{1}{2}$ up to some error $\epsilon + \epsilon' < \frac{1}{6}$. However, in the definition of \textsc{Approx-QCircuit-Prob}, we are guaranteed that $P_C(0) > \frac{2}{3}$ or $P_C(0) < \frac{1}{3}$, which are separated from $\frac{1}{2}$ by $\frac{1}{6}$, so by checking whether $j = 1$ or $j = 2$, we can distinguish between these cases. Therefore, \textsc{QDisCoCirc-QA} is BQP-hard.
\end{proof}

\renewcommand{\thetheorem}{4}
\begin{theorem}
    \label{haarrandom}
    Given a set of word embeddings $V$, suppose that operations in $V$ are independent Haar-random unitaries. Then conditions three and four of Theorem \ref{worstworst} are almost surely satisfied for all word embeddings containing at least two adjectives and one transitive verb.
\end{theorem}

\begin{proof}
    Let $U_1$ and $U_2$ be the adjectives defined by condition three of Theorem 1, and $U_3$ be the transitive verb identified by condition four. Consider $U_3$ as a point in $SU(4)$. If $U_3$ is not entangling it is either separable or locally equivalent to the SWAP gate. The spaces formed in both of these cases are isomorphic to $SU(2) \times SU(2)$ which is a 9-dimensional subspace of the 15-dimensional $SU(4)$. Hence condition four is satisfied almost everywhere in $SU(4)$. Therefore, for Haar random $U_3$ it is almost surely satisfied. 

    Consider arbitrary $U_1$ and $U_2$. If $U_1$ and $U_2$ are non-commuting elements of infinite order in $SU(2)$, then they generate a dense subset of $SU(2)$. This is because the only closed positive-dimension Lie subgroup of $SU(2)$ is itself, so it must be equal to the closure of the group generated by $U_1$ and $U_2$. Consider $U_1$ and $U_2$ as a point in $SU(2) \times SU(2)$. Since elements of $SU(2)$ can also be characterized as rotations of the Bloch sphere, then commuting elements correspond to rotations around the same axis. Thus the subspace of commuting $U_1$ and $U_2$ is isomorphic to $SU(2) \times S$. Likewise, elements of finite order in $SU(2)$ correspond to rotations about any axis by a rational multiple of $\pi$, and the subspace is isomorphic to $S^2 \times \mathbb{Q}$, which has measure zero. Thus almost everywhere in $SU(2) \times SU(2)$, $U_1$ and $U_2$ generate a dense subset of $SU(2)$. Hence, for Haar random $U_1$ and $U_2$, this is almost surely the case.
\end{proof}

\renewcommand{\thetheorem}{5}
\begin{theorem}
    \label{typicalembeddings}
    Given a specific set of word embeddings $V$ it is possible to check numerically that the conditions of Theorem \ref{worstworst} are satisfied. 
\end{theorem}

\begin{proof}
    Let $U_1$ and $U_2$ be the adjectives defined by condition three of Theorem \ref{worstworst}, and $U_3$ be the transitive verb identified by condition four. Consider the state $\ket{U_3}$ obtained via the Choi-Jamio\l kowski isomorphism. If we compute the reduced density matrix $\hat{\rho} = \ptr{2}{\ket{U_3}\bra{U_3^\dagger}}$ (where the partial trace is taken over both qubits corresponding to the second qubit of $U_3$), then this is pure (i.e has rank one) if and only if $U_3$ is separable. Therefore, we can check that $U_3$ is entangling by performing this test on both $U_3$ and $U_3 \cdot SWAP$ (to check that it is not locally equivalent to a SWAP).

    To check that $U_1$ and $U_2$ generate a dense subset of $SU(2)$, first check that they do not commute. We wish to check that the closure of the subgroup generated by $U_1$ and $U_2$ is $SU(2)$ itself. The only infinite closed nonabelian proper subgroup of $SU(2)$ is isomorphic to $S \rtimes_\phi C_2$ (where $\phi_0(x) = x$, $\phi_1(x) = -x$). In this subgroup for all $x$ and $y$, $x^2$ and $y^2$ commute, so to rule out this case it suffices to find a non-commuting pair. This can be done numerically by generating elements of the group by taking products and checking randomly until an example is found.
    
    The only finite closed nonabelian subgroups of $SU(2)$ are the binary dicylic groups, and three groups of order 24, 48 (the single-qubit Cliffords), and 120, by the ADE classification. The binary dicylic groups are subgroups of $S \rtimes_\phi C_2$, so we have already ruled these out. Therefore, it remains to check that the order of the group generated by $U_1$ and $U_2$ is more than 120. This can be done numerically by testing products of group elements until 121 unique elements are generated.
\end{proof}

\renewcommand{\thetheorem}{6}
\begin{theorem}
    \label{averagetexts}
    There exists a polynomial $f(m, n)$, such that given an oracle to solve \textup{\textsc{QDisCoCirc-QA}} instances with arbitrary word embeddings and text circuits drawn from $\mathcal{D}_{f(m, n)}$, we can perform arbitrary quantum computations with $m$ $2$-local gates on $n$ qubits with high probability in polynomial time. That is, \textup{\textsc{QDisCoCirc-QA}} for worst-case word embeddings is average-case BQP-hard over texts.
\end{theorem}

\begin{proof}
    We will reduce from the problem \textsc{Approx-QCircuit-Prob}. We are given an arbitrary quantum circuit on $n$ qubits with $m$ two-qubit gates given by arbitrary unitaries, and we assume that any single-qubit gates is combined into an adjacent two-qubit gate -- the only time this is not possible is when a qubit is unentangled, so it can be disregarded (if it is the measurement qubit, we can calculate the outcome immediately, and if it is not, it doesn't affect the outcome). 

    Suppose that we have sampled a large random text circuit from $\mathcal{D}_k$ for some $k \gg n$. We wish to calculate how many boxes are required so that there exists a subsequence which matches the given circuit. By match, we mean that each gate has a corresponding box in the circuit with a unique label, operating on the same wires as the gate. Let $N_i$ be the gap between gates $i$ and the previous gate (or that start of the circuit) -- that is, the number of boxes between them in the circuit that are not part of the subsequence. Because all boxes in the circuit are chosen in an identical random manner, $P(N_i \geq q) = (1 - p_i)^q$ where $p_i$ is the probability of any particular box being correct for gate $i$. Since inputs to boxes are chosen uniformly, we have that all the $p_i = p$ are identical. 

    Therefore, we can bound the probability that the total sequence length required is at least some constant $q$:
    \begin{align*}
        P\left(\sum_i N_i \geq q\right) &\leq 1 - P\left(\bigwedge_i N_i < \frac{q}{m} \right) = 1 - \left(1 - P\left(N_i \geq \frac{q}{m}\right)\right)^m = 1 - (1 - (1 - p)^{\frac{q}{m}})^m \\
        &\leq 1 - (1 - m(1 - p)^\frac{q}{m}) = m(1 - p)^\frac{q}{m} \leq m e^{-\frac{pq}{m}}
    \end{align*}
    The second line follows from Bernoulli's inequality. Thus letting $q = \frac{m}{p}(\log(m) + 1)$, we have that $P\left(\sum_i N_i \geq q\right) \leq \frac{1}{e}$. We have
    $$ p = P(A = 2) \cdot P(\mathrm{hapax}) \cdot \frac{2}{k(k - 1)} $$
    where $P(\mathrm{hapax})$ is the probability that any particular box has a unique label, since $\frac{2}{k(k - 1)}$ is the probability of selecting any particular pair of wires uniformly, assuming that $k > n$. According to Heap's law, the number of distinct words in a text grows like $\sqrt{w}$ where $w$ is the total number of words. Furthermore from Zipf's law, a constant proportion of distinct words occur only once (\emph{hapax legomena}), usually about half for large corpora. By assumption, this law holds exactly for circuits drawn from $\mathcal{D}_k$, so this implies that
    $$ P(\mathrm{hapax}) \leq \frac{c_1\sqrt{k^\gamma}}{k^\gamma} = \frac{c_1}{\sqrt{k^\gamma}} ~~\text{ and }~~ p \geq \frac{c_2}{k^{2 + \frac{\gamma}{2}}}$$
    for some constants $c_1, c_2 > 0$, since we assume $w \geq c_3 k^\gamma$ for some $c_3$. Therefore, we can set $q \leq \frac{1}{c_2}m(\log(m) + 1)k^{2+\frac{\gamma}{2}}$, but we also require $c_3k^\gamma \geq q$ to ensure the number of boxes required is less than the total in the text. This can be enforced by setting:
    $$ c_3k^\gamma \geq \frac{1}{c_2}m(\log(m) + 1)k^{2+\frac{\gamma}{2}} \implies k \geq \left(\frac{1}{c_2c_3}m(\log(m) + 1)\right)^{\frac{2}{\gamma - 4}}$$
    Finally, in order to make sure $k \geq n$, we can set $k = f(m, n) = \max(n, \lceil (c m(\log(m) + 1))^{\frac{2}{\gamma - 4}}\rceil)$ for some $c > 0$.

    The protocol to solve \textsc{Approx-QCircuit-Prob} is as follows: we are given a circuit $C$ on $n$ qubits with $m$ two-qubit gates (single-qubit gates are eliminated as before). Sample a text circuit from $\mathcal{D}_k$ with $k$ as given above. With probability at least $1 - e^{-1}$, there exists a subsequence of boxes in the circuit that matches $C'$. We construct a set of word embeddings for the circuit by setting each word that does not appear in the subsequence to be the identity, and every word that is part of the subsequence to be the unitary of the corresponding gate in the circuit. Each state is assigned to the $\ket{0}$ quantum state. This defines the text circuit and embeddings for the context text $T$. Similarly, we can obtain the questions $Q_1$ and $Q_2$ by drawing from $\mathcal{D}_1$, setting all boxes to the identity, and assigning the state to $\ket{0}$ and $\ket{1}$ respectively.
    
    In the same way as in Theorem \ref{worstworst}, solving this instance of \textsc{QDisCoCirc-QA} with $\epsilon < \frac{1}{7}$ is sufficient to solve the corresponding \textsc{Approx-QCircuit-Prob} instance. $f(m, n)$ as given above grows polynomially if $\gamma > 4$, and the construction of the word embeddings is done in polynomial time. Since this procedure has fixed positive success probability, it can be amplified to any high success probability with logarithmically many iterations.
\end{proof}

%% file: sections/oracle.tex
\section{Oracle Construction}
\label{sec:oracle}

As a warmup, we will start by constructing oracles for a different task called bAbI1, first proposed in \cite{weston2015aicomplete}. In the DisCoCirc model of the bAbI1 task, we have a context text from which we extract information by asking a series of yes/no questions. An example context text is:

\begin{quote}
    John goes to the bedroom. \\
    Mary walks to the hallway. \\
    John goes back to the bathroom.
\end{quote}

And the task then is to answer a question about the state of the world. For example: ``Where is John?". In DisCoCirc, we do this by stating a potential answer to the question as an affirmative: ``John is in the bathroom." Then we measure which potential answer (``John is in the bathroom", ``John is in the hallway", etc) has the highest overlap with the context text, discarding the nouns that do not occur in the answer. Graphically, this means we wish to find the maximum of the following circuits:
$$\input{./figures/babi1-overview.tikz}$$
We can phrase this as a instance of the closest vector problem by noting that this is equivalent to finding $P_i$ which maximizes the following probability. The set $\{P_1, P_2 \cdots P_i  \cdots\}$ is any set of permutations such that $P_j$ brings the $j$th input to the top.
$$\input{./figures/babi1-split.tikz}$$
To view this as a closest vector problem, we divide the problem along the dotted line, letting the upper portion be the variable set of states $\rho_i$ and the lower portion be $\rho_0 = \ket{v_0}\bra{v_0}$. Since $\rho_0$ is pure, we can apply the quantum closest vector algorithm of Basheer et al \cite{basheer2021quantum} to find $i = \arg \max_{j \geq 1} \mathrm{F}(\rho_j, \rho_0)$. For this task the oracles $\mathcal{V}$ and $\mathcal{W}$ are given by the following circuits:
$$\input{./figures/babi1-oracles.tikz}$$

Here, $P$ is a 1-to-$n$ \emph{multiplexer}, it maps each wire to the top output (leaving the rest undefined) depending on the control register $\ket{i}$. There are several ways of constructing this operation. For instance, we can do a CSWAP controlled on every possible bitstring of the control register $\ket{i} = \ket{i_1 i_2 \cdots i_b}$. Here, the CSWAPs are swapping nouns, so they can be realized with  $O(\log N)$ CSWAPs acting on qubits. This looks like the following circuit: 
$$\input{./figures/multiplexer-option1.tikz}$$
The gate complexity of this is $O(n \log(N) \log(n))$ since we have to put a CX gate with $O(\log(n))$ controls on each of $O(n \log(N))$ qubits, and swaps the noun at the index with the top noun, leaving the rest unchanged. However, we can do better by leaving the other nouns in an unspecified order using a binary search method - first, we swap the first and second half of the nouns depending on the first bit of the index, then the first and second quarter depending on the second bit of the index, and so on, until the last bit controls whether we swap the top two nouns. This requires $n\left(\frac{1}{2} + \frac{1}{4} + \frac{1}{8} + \cdots\right) O(\log(N)) = O(n \log(N))$ singly-controlled Toffolis, so has an overall gate complexity of $O(n \log(N))$. For example, the circuit on eight words looks like the following:
$$\input{./figures/multiplexer-option2.tikz}$$
$P$ can be constructed recursively as:
$$\input{./figures/multiplexer-option2-construction.tikz}$$
Therefore, the oracles are given as follows in pseudocode:
\begin{lstlisting}
fn V(n1, n2) {
    IsIn(n1, n2);
}

fn NounCSWAP(ctrl, n1, n2) {
    for q in 0..log(N) {
        if ctrl {
            SWAP(n1[q], n2[q]);
        }
    }
}

fn P(i, nouns) {
    if nouns.len() == 1 {
        return;
    } else if nouns.len() == 2 { 
        NounCSWAP(i[0], nouns[0], nouns[1])
    } else {
        const m = nouns.len() / 2;

        for j in 0..m {
            NounCSWAP(i[0], nouns[j], nouns[m + j]);
        }

        P(i[1..], nouns[..m]);
    }
}

fn W(i, n1, n2) {
    let ancillas = InitAncillas(log(N) * (n - 2));
    let nouns = [n1, n2, ..ancillas];
    Text(nouns);
    P(i, nouns);
    Discard(ancillas);
}
\end{lstlisting}

The gate complexity of the $\mathcal{W}$ oracle is thus $O\left(w A(N, d) + n \log(N)\right)$, and for the $\mathcal{V}$ oracle it is $O(A(N, 2)) = O(A(N, d))$. It is reasonable to make the assumption that $A(N, d) = O(d \log(N))$. For example, this is the case if we are using standard ansatze like those from Sim et al \cite{Sim_2019}. In the closest vector algorithm for this task, we have $M = O(n)$. Therefore, for a fixed success probability and precision $\epsilon$ for calculating fidelity, the overall gate complexity of the algorithm is:
$$ O\left(\frac{\sqrt{n}}{\epsilon} (w d + n + \log(N)) \log(N)\right) \leq O\left(\frac{\sqrt{n} w d \log^2(N)}{ \epsilon}\right) $$ 
The inequality here follows from $n = O(w)$. Furthermore, the number of ancilla used by the oracles is at most $O(n \log(N))$.

\subsection{The Text-Text Similarity Task}

The text-text similarity task involves finding a text that is most similar to a target text, by maximizing the overlap between the states representing the two texts. This is a versatile primitive that can be used both for full text-text similarity and for some variants of question answering. As discussed in Section \ref{sec:grover}, we can use the closest vector algorithm to solve this problem. Pictorially, we want to find $i$ such that
$$\input{./figures/tt-def.tikz}$$
is maximized. For this task the oracles $\mathcal{V}$ and $\mathcal{W}$ are given by:
\begin{itemize}
    \item $\mathcal{V}\ket{0} = \ket{T_0}$
    \item $\mathcal{W}\ket{i} = \ket{T_i}$
\end{itemize}
Clearly, we have $\mathcal{V} = U_{T_0}$ which is the direct translation of the text circuit into a quantum circuit via QDisCoCirc. However, for $\mathcal{W}$, we need a significantly more complicated circuit. Let us define the following parameters:
\begin{itemize}
    \item Each noun wire has dimension $N$,
    \item Each text circuit considered has at most $w$ unitaries and $n$ noun wires total,
    \item Each unitary in the text circuit acts on at most $d$ wires,
    \item There are $V$ words in the vocabulary,
    \item The number of gates in the ansatz for $d$ wires of dimension $N$ scales as $A_g(N, d)$,
    \item The number of parameters in the ansatz for $d$ wires of dimension $N$ scales as $A_p(N, d)$,
    \item Each parameter for the ansatz is stored with $P$ bits of precision,
    \item We consider $M$ possible texts, so that $i \leq M$.
\end{itemize}

We will construct the circuit for $\mathcal{W}$ by interleaving layers of two operations. One operation is a parameterized partial permutation, which permutes some of its inputs according to information stored in a QRAM. The second is a parameterized ansatz, which applies an ansatz of a given size to a set of qubits with the parameterized angles of the ansatz stored in a QRAM. By combining these two operations, we can reconstruct a text circuit by loading from QRAM the parameters and argument locations for each unitary.

We assume access to the following QRAM oracles:
\begin{itemize}
    \item Param$_{i j}\ket{k}\ket{0} = \ket{k}\ket{\phi_{i j k}}$ where $\phi_{i j k}$ is a $P$-bit fixed-precision binary encoding of the $j$th parameter for the ansatz of the $i$th unitary in $T_k$,
    \item Width$_i \ket{k}\ket{0} = \ket{k}\ket{W_{i k}}$ where $W_{i k}$ is a binary encoding of the width of the $i$th unitary in $T_k$,
    \item Index$_{i j}\ket{k}\ket{0} = \ket{k}\ket{I_{i j k}}$ where $I_{i j k}$ is an encoding of the position of the $j$th argument to unitary $i$ in $T_k$. Note that the encoding here is not a straightforward mapping from zero-indexed position to binary, and we will elaborate more later.
\end{itemize}
Each of these oracles can be implemented using a series of calls (one for each bit) to bucket-brigade QRAMS of size $M$ \cite{giovannetti2008quantum}. In pseudocode, each of these lookups is as follows:

\begin{lstlisting}
fn QRAMLookup(dict, qidx, cidx, out) {
    const ORACLES: [[Unitary]];
    let size = out.len();
    for i in 0..size {
        ORACLES[dict][cidx * size + i](qidx, out[i]);
    }
}
\end{lstlisting}
where \texttt{dict} specifies the chosen oracle, \texttt{qidx} the quantum index $\ket{k}$ and \texttt{cidx} the classical index (i.e $i j$ or $i$).

For the parameterized ansatz, we can construct a circuit out of some building blocks:
\begin{itemize}
    \item Firstly, we have a parameterized controlled rotation gate, which implements a controlled-$Z$ rotation based on an angle specified in auxiliary qubits.
    \item For a given ansatz of a specific size, we can implement a parameterized controlled circuit by replacing all gates except for parameterized rotations with their controlled versions. For the parameterized rotations, we then load the parameters from QRAM and implement each rotation using a parameterized controlled rotation circuit.
    \item This gives a parameterized ansatz of a specific width. To cover all possible widths, first we can load the width from QRAM, and perform the parameterized controlled ansatz of all possible widths, controlled on the width loaded from QRAM. This gives a general parameterized ansatz.
\end{itemize}
For the parameterized controlled rotation gate, we have:
\begin{lstlisting}
fn PCRz(ctrl, angle, q) {
    const MAX_PREC;
    for i in 0..=MAX_PREC {
        CCRz(pi / (2^i), angle[i], ctrl, q);
    }
}
\end{lstlisting}
which is diagrammatically:
$$\input{./figures/tt-oracle/pcrz.tikz}$$
This has a gate complexity of $O(P)$.

Then for a specific realization of an ansatz of a given width, we can convert it to a parameterized controlled ansatz as follows. Here \texttt{angle\_base} refers to the location in QRAM of the parameters for this ansatz, \texttt{text\_idx} is the quantum register specifying the text index, and \texttt{angle} is a ancilla register for storing a parameter.
\begin{lstlisting}
fn PCAnsatz(ansatz, ctrl, text_idx, angle_base, angle, qs) {
    for gate in ansatz {
        match gate {
            CX(a, b) => CCX(ctrl, qs[a], qs[b]),
            RzParam(idx, a) => {
                QRAMLookup(
                    ANGLES, text_idx, angle_base + idx, angle
                );
                PCRz(ctrl, angle, qs[a]);
                QRAMLookup(
                    ANGLES, text_idx, angle_base + idx, angle
                );
            },
            RzFixed(p, a) => CRz(p, ctrl, qs[a]),
            H(a) => CH(ctrl, qs[a])
        }
    }
}
\end{lstlisting}
For example, for an ansatz of width one (parameterized Euler rotations), corresponding to $d = 1, N = 2$, we have:
$$\input{./figures/tt-oracle/pcansatz1.tikz}$$
And for a one-layer of ansatz 9 from Sim et al. \cite{Sim_2019} of width four (i.e for $d = 4$ and $N = 2$ or $d = 2$ and $N = 4$) \cite{Sim_2019}, we have the following:
$$\input{./figures/tt-oracle/pcansatz.tikz}$$
This has a gate complexity of $O(A_g(N, d) + P \log(M) A_p(N, d))$. These ansatze of differing widths can be combined into one parameterized ansatz that controls all of these based on a width loaded from QRAM. We wish to do the following operation
\begin{lstlisting}
fn PAnsatz(width, text_idx, angle_base, angle, qs) {
    const MAX_WIDTH;
    const ANSATZE = [
        AnsatzGates(w) for w in 1..=MAX_WIDTH
    ];

    QRAMLookup(WIDTHS, text_idx, i, width);
    for w in 1..=MAX_WIDTH {
        let ctrl = w == width;
        PCAnsatz(ANSATZE[w], ctrl, text_idx, angle_base, angle, qs)
        uncompute ctrl;
    }
    QRAMLookup(WIDTHS, text_idx, i, width);
}
\end{lstlisting}
which is represented diagrammatically as:
$$\input{./figures/tt-oracle/pansatz-naive.tikz}$$
However, compiling this circuit naively results in a gate complexity of 
$$O(d \log(d) \cdot (A_g(N, d) + P \log(M) A_p(N, d)))$$
since each parameterized ansatz has $O(\log(d))$ controls that must be combined. By making use of the unary iteration method, we can compile this more efficiently with gate complexity of $O(d (A_g(N, d) + P \log(M) A_p(N, d)))$. Given a series of unitaries $U_{0 \cdots 0}$ to $U_{1 \cdots 1}$ which we want to apply conditionally given an index register of size $\log(d)$, this technique synthesizes a circuit using singly-controlled versions of $U_{i \cdots j}$ with only $O(d)$ gates. It is based on the following recurrences:
$$\input{./figures/tt-oracle/unary-iteration-0.tikz}$$
$$\input{./figures/tt-oracle/unary-iteration-2.tikz}$$
This second equation is applied until only single controls on each $U_{i \cdots j}$ remain. In pseudocode, this process is given by:
\begin{lstlisting}
fn UnaryIteration(ctrls, ancillas, func) {
    X(ctrls[0]);
    CUnaryIteration(ctrls[1..], 0, ctrls[0], ancillas, func);
    X(ctrls[0]);
    CUnaryIteration(ctrls[1..], 1, ctrls[0], ancillas, func);
}

fn CUnaryIteration(ctrls, base, prev, ancillas, func) {
    if ctrls.len() > 0 {
        X(ctrls[0]);
        let fresh = ancillas[0];
        CCX(prev, ctrls[0], fresh);
        CUnaryIteration(
            ctrls[1..], 2*base + 0, fresh, ancillas[1..], func
        );
        CX(prev, fresh);
        CUnaryIteration(
            ctrls[1..], 2*base + 1, fresh, ancillas[1..], func
        );
        CCX(prev, ctrls[0], fresh);
        Discard(fresh);
    } else {
        func(base, prev);
    }
}
\end{lstlisting}
Applying this to create a parameterized ansatz, we have the following pseudocode:
\begin{lstlisting}
fn PAnsatz(width, text_idx, angle_base, angle, ancillas, qs) {
    const MAX_WIDTH;
    const ANSATZE = [
        AnsatzGates(w) for w in 1..=MAX_WIDTH
    ];

    QRAMLookup(WIDTHS, text_idx, i, width);
    UnaryIteration(width, ancillas, |idx, ctrl| {
        if idx > 0 && idx <= MAX_WIDTH {
            PCAnsatz(
                ANSATZE[idx], ctrl, text_idx, angle_base, angle, qs
            );
        }
    });
    QRAMLookup(WIDTHS, text_idx, i, width);
}
\end{lstlisting}
For the case that \texttt{MAX\_WIDTH = 4} and hence $\log(d) = 2$, this can be represented diagrammatically as follows (although note that we use a more optimized form of unary iteration on two index qubits). 
$$\input{./figures/tt-oracle/pansatz.tikz}$$
Note that in this setup, the index $0 \dots 0$ is mapped to an ansatz with width one, and in general $x$ in binary maps to the ansatz with width $x + 1$. Note that if we wish to apply the identity operator, we can apply an ansatz of width one with all angles set to zero. This has an overall gate complexity of $O(d (A_g(N, d) + P \log(M) A_p(N, d)))$.

For the parameterized partial permutation, we can reuse the multiplexer given before to implement each bit of the partial permutation. This puts a bit of our choice at the top of the multiplexer and leaves the rest permuted. By stacking $d$ of these together, we can specify which qubits should form the top $d$ qubits that the ansatz is applied too. We can calculate the required indices by propagating the desired index forward through all previous changes to calculate the actual index required at each step. Noting that each swap in the multiplexer is actually a swap of noun wires (which consist of multiple qubits), we can represent the multiplexer in pseudocode as:
\begin{lstlisting}
fn NounCSWAP(ctrl, n1, n2) {
    const N;

    for q in 0..log(N) {
        CSWAP(ctrl, n1[q], n2[q]);
    }
}

fn P(idx, nouns) {
    if nouns.len() == 1 {
        return;
    } else if nouns.len() == 2 { 
        NounCSWAP(idx[0], nouns[0], nouns[1])
    } else {
        const mf = floor(nouns.len() / 2);
        const mc = ceil(nouns.len() / 2);

        for j in 0..mf {
            NounCSWAP(i[0], nouns[j], nouns[mc + j]);
        }

        P(idx[1..], nouns[..mc]);
    }
}
\end{lstlisting}
This has a gate complexity of $O(n \log(N))$. From this the parameterized partial permutation can be constructed as
\begin{lstlisting}
fn NounMux(i, idx, text_idx, nouns) {
    const MAX_ARITY;

    for j in 0..MAX_ARITY {
        QRAMLookup(INDICES, text_idx, i * MAX_ARITY + j, idx);
        P(idx, nouns[j..]);
        QRAMLookup(INDICES, text_idx, i * MAX_ARITY + j, idx);
    }
}
\end{lstlisting}
and diagrammatically, this looks as follows:
$$\input{./figures/tt-oracle/nounmux.tikz}$$
This has an overall gate complexity of $O(d (n \log(N) + \log(n)\log(M)))$.

Finally, by layering these two parts, we can construct the oracle. In this way, we first move the wires for each unitary to the top of the circuit, apply an ansatz (which applies the given unitary), and apply the inverse of the permutation. Then we repeat until no unitaries are remaining. In pseudocode, we have
\begin{lstlisting}
fn Oracle(text_idx, nouns) {
    const MAX_PARAMS;
    const MAX_WORDS;
    const MAX_WIDTH;
    const MAX_NOUNS;
    const MAX_PREC;
    
    let idx = InitAncillas(MAX_NOUNS.log2());
    let angle = InitAncillas(MAX_PREC);
    let width = InitAncillas(MAX_WIDTH.log2());
    let ui_ancillas = InitAncillas(MAX_WIDTH.log2());
    for i in 0..MAX_WORDS {
        NounMux(i, idx, text_idx, nouns);
        
        PAnsatz(
            width, text_idx, i * MAX_PARAMS, 
            angle, ui_ancillas, nouns[..MAX_ARITY]
        );

        NounMux'(i, idx, text_idx, nouns);
    }
    Discard(idx);
    Discard(angle);
    Discard(width);
    Discard(ui_ancillas);
}
\end{lstlisting}
and diagrammatically this looks like:
$$\input{./figures/tt-oracle/oracle.tikz}$$
This has an overall gate complexity of
$$ O(w d (n \log(N) + \log(n)\log(M) + A_g(N, d) + P \log(M) A_p(N, d))) $$
which can be simplified to:
$$ O(w d \log(M) (n \log(N) + P A_g(N, d))) $$

%% file: figures/babi1-overview.tikz
\begin{tikzpicture}
	\begin{pgfonlayer}{nodelayer}
		\node [style=none] (0) at (2.5, 0) {};
		\node [style=none] (1) at (2.5, -2) {};
		\node [style=none] (2) at (-3.5, -2) {};
		\node [style=none] (3) at (-3.5, 0) {};
		\node [style=none] (4) at (-0.5, -1) {Text};
		\node [style=none] (6) at (2, 0) {};
		\node [style=none] (8) at (-1, 0) {};
		\node [style=none] (10) at (-3, 0) {};
		\node [style=none] (13) at (0.5, 0) {};
		\node [style=none] (14) at (2, -2) {};
		\node [style=none] (15) at (0.5, -2) {};
		\node [style=none] (16) at (-1, -2) {};
		\node [style=none] (17) at (-3, -2) {};
		\node [style=none] (18) at (-3, -3) {};
		\node [style=none] (19) at (-1, -3) {};
		\node [style=none] (20) at (-0.65, -3) {};
		\node [style=none] (21) at (-1.375, -3) {};
		\node [style=none] (22) at (-0.75, -3.25) {};
		\node [style=none] (23) at (-1.25, -3.25) {};
		\node [style=none] (24) at (-0.9, -3.5) {};
		\node [style=none] (25) at (-1.15, -3.5) {};
		\node [style=none] (26) at (-3, -3) {};
		\node [style=none] (27) at (-2.65, -3) {};
		\node [style=none] (28) at (-3.375, -3) {};
		\node [style=none] (29) at (-2.75, -3.25) {};
		\node [style=none] (30) at (-3.25, -3.25) {};
		\node [style=none] (31) at (-2.9, -3.5) {};
		\node [style=none] (32) at (-3.15, -3.5) {};
		\node [style=none] (33) at (-2, -2.5) {$\dots$};
		\node [style=none] (34) at (0.5, -4) {};
		\node [style=none] (35) at (2, -4) {};
		\node [style=none] (36) at (2.5, -4) {};
		\node [style=none] (37) at (0, -4) {};
		\node [style=none] (38) at (0, -6) {};
		\node [style=none] (39) at (2.5, -6) {};
		\node [style=none] (40) at (1.25, -5) {Is In$^\dagger$};
		\node [style=none] (42) at (2, -6) {};
		\node [style=none] (44) at (0.5, -6) {};
		\node [style=none] (45) at (11.5, 0) {};
		\node [style=none] (46) at (11.5, -2) {};
		\node [style=none] (47) at (5.5, -2) {};
		\node [style=none] (48) at (5.5, 0) {};
		\node [style=none] (49) at (8.5, -1) {Text};
		\node [style=none] (51) at (11, 0) {};
		\node [style=none] (53) at (8, 0) {};
		\node [style=none] (55) at (6, 0) {};
		\node [style=none] (58) at (9.5, 0) {};
		\node [style=none] (59) at (11, -2) {};
		\node [style=none] (60) at (9.5, -2) {};
		\node [style=none] (61) at (8, -2) {};
		\node [style=none] (62) at (6, -2) {};
		\node [style=none] (63) at (6, -3) {};
		\node [style=none] (71) at (6, -3) {};
		\node [style=none] (72) at (6.35, -3) {};
		\node [style=none] (73) at (5.625, -3) {};
		\node [style=none] (74) at (6.25, -3.25) {};
		\node [style=none] (75) at (5.75, -3.25) {};
		\node [style=none] (76) at (6.1, -3.5) {};
		\node [style=none] (77) at (5.85, -3.5) {};
		\node [style=none] (79) at (8, -4) {};
		\node [style=none] (80) at (11, -4) {};
		\node [style=none] (81) at (11.5, -4) {};
		\node [style=none] (82) at (7.5, -4) {};
		\node [style=none] (83) at (7.5, -6) {};
		\node [style=none] (84) at (11.5, -6) {};
		\node [style=none] (85) at (9.5, -5) {Is In$^\dagger$};
		\node [style=none] (87) at (11, -6) {};
		\node [style=none] (89) at (8, -6) {};
		\node [style=none] (90) at (9.5, -3) {};
		\node [style=none] (91) at (9.5, -3) {};
		\node [style=none] (92) at (9.85, -3) {};
		\node [style=none] (93) at (9.125, -3) {};
		\node [style=none] (94) at (9.75, -3.25) {};
		\node [style=none] (95) at (9.25, -3.25) {};
		\node [style=none] (96) at (9.6, -3.5) {};
		\node [style=none] (97) at (9.35, -3.5) {};
		\node [style=none] (98) at (7, -2.5) {$\dots$};
		\node [style=none] (99) at (20.5, 0) {};
		\node [style=none] (100) at (20.5, -2) {};
		\node [style=none] (101) at (14.5, -2) {};
		\node [style=none] (102) at (14.5, 0) {};
		\node [style=none] (103) at (17.5, -1) {Text};
		\node [style=none] (105) at (20, 0) {};
		\node [style=none] (107) at (17, 0) {};
		\node [style=none] (109) at (15, 0) {};
		\node [style=none] (112) at (18.5, 0) {};
		\node [style=none] (113) at (20, -2) {};
		\node [style=none] (114) at (18.5, -2) {};
		\node [style=none] (115) at (17, -2) {};
		\node [style=none] (116) at (15, -2) {};
		\node [style=none] (117) at (17, -3) {};
		\node [style=none] (118) at (17, -3) {};
		\node [style=none] (119) at (17.35, -3) {};
		\node [style=none] (120) at (16.625, -3) {};
		\node [style=none] (121) at (17.25, -3.25) {};
		\node [style=none] (122) at (16.75, -3.25) {};
		\node [style=none] (123) at (17.1, -3.5) {};
		\node [style=none] (124) at (16.85, -3.5) {};
		\node [style=none] (125) at (15, -4) {};
		\node [style=none] (126) at (20, -4) {};
		\node [style=none] (127) at (20.5, -4) {};
		\node [style=none] (128) at (14.5, -4) {};
		\node [style=none] (129) at (14.5, -6) {};
		\node [style=none] (130) at (20.5, -6) {};
		\node [style=none] (131) at (17.5, -5) {Is In$^\dagger$};
		\node [style=none] (133) at (20, -6) {};
		\node [style=none] (135) at (15, -6) {};
		\node [style=none] (136) at (18.5, -3) {};
		\node [style=none] (137) at (18.5, -3) {};
		\node [style=none] (138) at (18.85, -3) {};
		\node [style=none] (139) at (18.125, -3) {};
		\node [style=none] (140) at (18.75, -3.25) {};
		\node [style=none] (141) at (18.25, -3.25) {};
		\node [style=none] (142) at (18.6, -3.5) {};
		\node [style=none] (143) at (18.35, -3.5) {};
		\node [style=none] (144) at (16, -2.5) {$\dots$};
	\end{pgfonlayer}
	\begin{pgfonlayer}{edgelayer}
		\draw (0.center) to (1.center);
		\draw (1.center) to (2.center);
		\draw (2.center) to (3.center);
		\draw (0.center) to (3.center);
		\draw (16.center) to (19.center);
		\draw (18.center) to (17.center);
		\draw (20.center) to (21.center);
		\draw (22.center) to (23.center);
		\draw (24.center) to (25.center);
		\draw (27.center) to (28.center);
		\draw (29.center) to (30.center);
		\draw (31.center) to (32.center);
		\draw (39.center) to (36.center);
		\draw (39.center) to (38.center);
		\draw (38.center) to (37.center);
		\draw (37.center) to (36.center);
		\draw (35.center) to (14.center);
		\draw (15.center) to (34.center);
		\draw (45.center) to (46.center);
		\draw (46.center) to (47.center);
		\draw (47.center) to (48.center);
		\draw (45.center) to (48.center);
		\draw (63.center) to (62.center);
		\draw (72.center) to (73.center);
		\draw (74.center) to (75.center);
		\draw (76.center) to (77.center);
		\draw (84.center) to (81.center);
		\draw (84.center) to (83.center);
		\draw (83.center) to (82.center);
		\draw (82.center) to (81.center);
		\draw (80.center) to (59.center);
		\draw (79.center) to (61.center);
		\draw (92.center) to (93.center);
		\draw (94.center) to (95.center);
		\draw (96.center) to (97.center);
		\draw (91.center) to (60.center);
		\draw (99.center) to (100.center);
		\draw (100.center) to (101.center);
		\draw (101.center) to (102.center);
		\draw (99.center) to (102.center);
		\draw (119.center) to (120.center);
		\draw (121.center) to (122.center);
		\draw (123.center) to (124.center);
		\draw (130.center) to (127.center);
		\draw (130.center) to (129.center);
		\draw (129.center) to (128.center);
		\draw (128.center) to (127.center);
		\draw (126.center) to (113.center);
		\draw (138.center) to (139.center);
		\draw (140.center) to (141.center);
		\draw (142.center) to (143.center);
		\draw (137.center) to (114.center);
		\draw (116.center) to (125.center);
		\draw (118.center) to (115.center);
	\end{pgfonlayer}
\end{tikzpicture}

%% file: figures/babi1-split.tikz
\begin{tikzpicture}
	\begin{pgfonlayer}{nodelayer}
		\node [style=none] (0) at (2.75, 2) {};
		\node [style=none] (1) at (2.75, 0) {};
		\node [style=none] (2) at (-3.25, 0) {};
		\node [style=none] (3) at (-3.25, 2) {};
		\node [style=none] (4) at (-0.25, 1) {Text};
		\node [style=none] (6) at (2.25, 2) {};
		\node [style=none] (8) at (-0.75, 2) {};
		\node [style=none] (10) at (-2.75, 2) {};
		\node [style=none] (13) at (0.75, 2) {};
		\node [style=none] (14) at (2.25, 0) {};
		\node [style=none] (15) at (0.75, 0) {};
		\node [style=none] (16) at (-0.75, 0) {};
		\node [style=none] (17) at (-2.75, 0) {};
		\node [style=none] (18) at (-2.75, -4) {};
		\node [style=none] (19) at (-0.75, -4) {};
		\node [style=none] (20) at (-0.4, -4) {};
		\node [style=none] (21) at (-1.125, -4) {};
		\node [style=none] (22) at (-0.5, -4.25) {};
		\node [style=none] (23) at (-1, -4.25) {};
		\node [style=none] (24) at (-0.65, -4.5) {};
		\node [style=none] (25) at (-0.9, -4.5) {};
		\node [style=none] (26) at (-2.75, -4) {};
		\node [style=none] (27) at (-2.4, -4) {};
		\node [style=none] (28) at (-3.125, -4) {};
		\node [style=none] (29) at (-2.5, -4.25) {};
		\node [style=none] (30) at (-3, -4.25) {};
		\node [style=none] (31) at (-2.65, -4.5) {};
		\node [style=none] (32) at (-2.9, -4.5) {};
		\node [style=none] (33) at (-1.75, -3.5) {$\dots$};
		\node [style=none] (34) at (0.75, -6) {};
		\node [style=none] (35) at (2.25, -6) {};
		\node [style=none] (36) at (2.75, -6) {};
		\node [style=none] (37) at (0.25, -6) {};
		\node [style=none] (38) at (0.25, -8) {};
		\node [style=none] (39) at (2.75, -8) {};
		\node [style=none] (40) at (1.5, -7) {Is In$^\dagger$};
		\node [style=none] (42) at (2.25, -8) {};
		\node [style=none] (44) at (0.75, -8) {};
		\node [style=none] (45) at (1.25, -1) {};
		\node [style=none] (46) at (1.25, -3) {};
		\node [style=none] (47) at (-3.25, -3) {};
		\node [style=none] (48) at (-3.25, -1) {};
		\node [style=none] (49) at (-0.75, -1) {};
		\node [style=none] (50) at (0.75, -1) {};
		\node [style=none] (51) at (-2.75, -1) {};
		\node [style=none] (52) at (-0.75, -3) {};
		\node [style=none] (53) at (-2.75, -3) {};
		\node [style=none] (54) at (0.75, -3) {};
		\node [style=none] (55) at (4, -5) {};
		\node [style=none] (56) at (-4, -5) {};
		\node [style=none] (57) at (-1, -2) {$P_i$};
	\end{pgfonlayer}
	\begin{pgfonlayer}{edgelayer}
		\draw (0.center) to (1.center);
		\draw (1.center) to (2.center);
		\draw (2.center) to (3.center);
		\draw (0.center) to (3.center);
		\draw (20.center) to (21.center);
		\draw (22.center) to (23.center);
		\draw (24.center) to (25.center);
		\draw (27.center) to (28.center);
		\draw (29.center) to (30.center);
		\draw (31.center) to (32.center);
		\draw (39.center) to (36.center);
		\draw (39.center) to (38.center);
		\draw (38.center) to (37.center);
		\draw (37.center) to (36.center);
		\draw (35.center) to (14.center);
		\draw (46.center) to (45.center);
		\draw (48.center) to (45.center);
		\draw (48.center) to (47.center);
		\draw (47.center) to (46.center);
		\draw (51.center) to (17.center);
		\draw (16.center) to (49.center);
		\draw (50.center) to (15.center);
		\draw (53.center) to (26.center);
		\draw (19.center) to (52.center);
		\draw (54.center) to (34.center);
		\draw [style=box edge] (56.center) to (55.center);
	\end{pgfonlayer}
\end{tikzpicture}

%% file: figures/babi1-oracles.tikz
\begin{tikzpicture}
	\begin{pgfonlayer}{nodelayer}
		\node [style=none] (0) at (2, -0.75) {};
		\node [style=none] (1) at (2, 0.75) {};
		\node [style=none] (2) at (2, 1.25) {};
		\node [style=none] (3) at (2, -1.25) {};
		\node [style=none] (4) at (0, -1.25) {};
		\node [style=none] (5) at (0, 1.25) {};
		\node [style=none] (6) at (1, 0) {Is In};
		\node [style=none] (7) at (-0.5, 0.75) {};
		\node [style=none] (8) at (0, 0.75) {};
		\node [style=none] (9) at (-0.5, -0.75) {};
		\node [style=none] (10) at (0, -0.75) {};
		\node [style=none] (11) at (2.5, 0.75) {};
		\node [style=none] (12) at (2.5, -0.75) {};
		\node [style=none] (13) at (-3, 0) {$\mathcal{V}|0\rangle~~=~~$};
		\node [style=none] (27) at (8, 0) {$\mathcal{W}|i\rangle|0\rangle~~=~~$};
		\node [style=none] (28) at (11.75, 2.75) {};
		\node [style=none] (29) at (13.75, 2.75) {};
		\node [style=none] (30) at (13.75, -3.25) {};
		\node [style=none] (31) at (11.75, -3.25) {};
		\node [style=none] (32) at (12.75, -0.25) {Text};
		\node [style=none] (33) at (10.75, 2.25) {};
		\node [style=none] (34) at (11.75, 2.25) {};
		\node [style=none] (35) at (10.75, -0.75) {};
		\node [style=none] (36) at (11.75, -0.75) {};
		\node [style=none] (37) at (10.75, -2.75) {};
		\node [style=none] (38) at (11.75, -2.75) {};
		\node [style=none] (39) at (11.25, -1.5) {$\vdots$};
		\node [style=none] (40) at (10.75, 0.75) {};
		\node [style=none] (41) at (11.75, 0.75) {};
		\node [style=none] (42) at (13.75, 2.25) {};
		\node [style=none] (43) at (13.75, 0.75) {};
		\node [style=none] (44) at (13.75, -0.75) {};
		\node [style=none] (45) at (13.75, -2.75) {};
		\node [style=none] (46) at (17.75, -2.75) {};
		\node [style=none] (47) at (17.75, -0.75) {};
		\node [style=none] (48) at (17.75, -0.4) {};
		\node [style=none] (49) at (17.75, -1.125) {};
		\node [style=none] (50) at (18, -0.5) {};
		\node [style=none] (51) at (18, -1) {};
		\node [style=none] (52) at (18.25, -0.65) {};
		\node [style=none] (53) at (18.25, -0.9) {};
		\node [style=none] (54) at (17.75, -2.75) {};
		\node [style=none] (55) at (17.75, -2.4) {};
		\node [style=none] (56) at (17.75, -3.125) {};
		\node [style=none] (57) at (18, -2.5) {};
		\node [style=none] (58) at (18, -3) {};
		\node [style=none] (59) at (18.25, -2.65) {};
		\node [style=none] (60) at (18.25, -2.9) {};
		\node [style=none] (61) at (18, -1.5) {$\vdots$};
		\node [style=none] (62) at (18.25, 0.75) {};
		\node [style=none] (63) at (18.25, 2.25) {};
		\node [style=none] (66) at (14.75, 1.25) {};
		\node [style=none] (67) at (16.75, 1.25) {};
		\node [style=none] (68) at (16.75, -3.25) {};
		\node [style=none] (69) at (14.75, -3.25) {};
		\node [style=none] (70) at (14.75, -0.75) {};
		\node [style=none] (71) at (14.75, 0.75) {};
		\node [style=none] (72) at (14.75, -2.75) {};
		\node [style=none] (73) at (16.75, -0.75) {};
		\node [style=none] (74) at (16.75, -2.75) {};
		\node [style=none] (75) at (16.75, 0.75) {};
		\node [style=none] (77) at (15.75, -1) {$P$};
		\node [style=none] (78) at (10.75, 4) {};
		\node [style=cnot ctrl] (79) at (15.75, 4) {};
		\node [style=none] (80) at (18.25, 4) {};
		\node [style=none] (81) at (15.75, 1.25) {};
		\node [style=none] (82) at (-1, 0.75) {$|0\rangle$};
		\node [style=none] (83) at (-1, -0.75) {$|0\rangle$};
		\node [style=none] (84) at (10.25, 2.25) {$|0\rangle$};
		\node [style=none] (85) at (10.25, -0.75) {$|0\rangle$};
		\node [style=none] (86) at (10.25, -2.75) {$|0\rangle$};
		\node [style=none] (87) at (10.25, 0.75) {$|0\rangle$};
		\node [style=none] (88) at (10.25, 4) {$|i\rangle$};
	\end{pgfonlayer}
	\begin{pgfonlayer}{edgelayer}
		\draw (5.center) to (2.center);
		\draw (5.center) to (4.center);
		\draw (4.center) to (3.center);
		\draw (3.center) to (2.center);
		\draw (8.center) to (7.center);
		\draw (10.center) to (9.center);
		\draw (12.center) to (0.center);
		\draw (1.center) to (11.center);
		\draw (34.center) to (33.center);
		\draw (28.center) to (29.center);
		\draw (29.center) to (30.center);
		\draw (30.center) to (31.center);
		\draw (28.center) to (31.center);
		\draw (36.center) to (35.center);
		\draw (38.center) to (37.center);
		\draw (41.center) to (40.center);
		\draw (48.center) to (49.center);
		\draw (50.center) to (51.center);
		\draw (52.center) to (53.center);
		\draw (55.center) to (56.center);
		\draw (57.center) to (58.center);
		\draw (59.center) to (60.center);
		\draw (63.center) to (42.center);
		\draw (67.center) to (66.center);
		\draw (69.center) to (66.center);
		\draw (69.center) to (68.center);
		\draw (68.center) to (67.center);
		\draw (72.center) to (45.center);
		\draw (44.center) to (70.center);
		\draw (71.center) to (43.center);
		\draw (74.center) to (54.center);
		\draw (47.center) to (73.center);
		\draw (75.center) to (62.center);
		\draw (81.center) to (79);
		\draw (80.center) to (79);
		\draw (79) to (78.center);
	\end{pgfonlayer}
\end{tikzpicture}

%% file: figures/multiplexer-option1.tikz
\begin{tikzpicture}
	\begin{pgfonlayer}{nodelayer}
		\node [style=none] (0) at (-1.75, 1.75) {};
		\node [style=none] (1) at (-1.75, -1.25) {};
		\node [style=none] (2) at (-1.75, -3.25) {};
		\node [style=none] (3) at (1.75, -1.25) {};
		\node [style=none] (4) at (1.75, -3.25) {};
		\node [style=none] (5) at (1.75, 1.75) {};
		\node [style=none] (6) at (-0.75, 2.25) {};
		\node [style=none] (7) at (0.75, 2.25) {};
		\node [style=none] (8) at (0.75, -3.75) {};
		\node [style=none] (9) at (-0.75, -3.75) {};
		\node [style=none] (10) at (-0.75, -1.25) {};
		\node [style=none] (11) at (-0.75, 1.75) {};
		\node [style=none] (12) at (-0.75, -3.25) {};
		\node [style=none] (13) at (0.75, -1.25) {};
		\node [style=none] (14) at (0.75, -3.25) {};
		\node [style=none] (15) at (0.75, 1.75) {};
		\node [style=none] (16) at (0, -0.5) {$P$};
		\node [style=cnot ctrl] (18) at (0, 3) {};
		\node [style=none] (19) at (1.75, 3) {};
		\node [style=none] (20) at (0, 2.25) {};
		\node [style=none] (21) at (-1.25, -2) {$\vdots$};
		\node [style=none] (22) at (1.25, -2) {$\vdots$};
		\node [style=none] (24) at (1.75, 4.5) {};
		\node [style=none] (25) at (0.75, 4) {$\vdots$};
		\node [style=none] (26) at (-0.75, 4) {$\vdots$};
		\node [style=none] (27) at (-1.75, 3) {};
		\node [style=none] (28) at (-1.75, 4.5) {};
		\node [style=cnot ctrl] (29) at (0, 4.5) {};
		\node [style=none] (30) at (-2.5, 4.5) {$|i_1\rangle$};
		\node [style=none] (31) at (-2.5, 3) {$|i_b\rangle$};
		\node [style=none] (32) at (5.5, 1.75) {};
		\node [style=none] (33) at (5.5, 0.25) {};
		\node [style=none] (35) at (11, 0.25) {};
		\node [style=none] (37) at (11, 1.75) {};
		\node [style=none] (50) at (11, 3) {};
		\node [style=none] (52) at (6, -2.25) {$\vdots$};
		\node [style=none] (53) at (10.5, -2.25) {$\vdots$};
		\node [style=none] (54) at (11, 4.5) {};
		\node [style=none] (57) at (5.5, 3) {};
		\node [style=none] (58) at (5.5, 4.5) {};
		\node [style=none] (60) at (4.75, 4.5) {$|i_1\rangle$};
		\node [style=none] (61) at (4.75, 3) {$|i_b\rangle$};
		\node [style=cnot ctrl] (62) at (6.5, 3) {};
		\node [style=none] (63) at (6.5, 1.75) {$\times$};
		\node [style=none] (64) at (6.5, 0.25) {$\times$};
		\node [style=cnot ctrl] (65) at (6.5, 4.5) {};
		\node [style=cnot ctrl] (66) at (7.5, 3) {};
		\node [style=none] (67) at (7.5, 1.75) {$\times$};
		\node [style=none] (68) at (7.5, -1.25) {$\times$};
		\node [style=vertex set] (69) at (7.5, 4.5) {};
		\node [style=none] (70) at (-1.75, 0.25) {};
		\node [style=none] (71) at (1.75, 0.25) {};
		\node [style=none] (72) at (-0.75, 0.25) {};
		\node [style=none] (73) at (0.75, 0.25) {};
		\node [style=none] (74) at (5.5, -1.25) {};
		\node [style=none] (76) at (11, -1.25) {};
		\node [style=none] (78) at (5.5, -3.25) {};
		\node [style=none] (79) at (11, -3.25) {};
		\node [style=vertex set] (80) at (9.5, 3) {};
		\node [style=none] (81) at (9.5, 1.75) {$\times$};
		\node [style=none] (82) at (9.5, -3.25) {$\times$};
		\node [style=vertex set] (83) at (9.5, 4.5) {};
		\node [style=none] (84) at (8.5, 3.75) {$\dots$};
		\node [style=none] (85) at (8.5, -0.5) {$\dots$};
		\node [style=none] (86) at (3.75, 0) {$=$};
	\end{pgfonlayer}
	\begin{pgfonlayer}{edgelayer}
		\draw (7.center) to (6.center);
		\draw (9.center) to (6.center);
		\draw (9.center) to (8.center);
		\draw (8.center) to (7.center);
		\draw (12.center) to (2.center);
		\draw (1.center) to (10.center);
		\draw (11.center) to (0.center);
		\draw (14.center) to (4.center);
		\draw (3.center) to (13.center);
		\draw (15.center) to (5.center);
		\draw (20.center) to (18);
		\draw (19.center) to (18);
		\draw (27.center) to (18);
		\draw (29) to (28.center);
		\draw (29) to (24.center);
		\draw (29) to (18);
		\draw (63.center) to (32.center);
		\draw (64.center) to (33.center);
		\draw (63.center) to (64.center);
		\draw (63.center) to (62);
		\draw (65) to (62);
		\draw (67.center) to (68.center);
		\draw (67.center) to (66);
		\draw (69) to (66);
		\draw (64.center) to (35.center);
		\draw (67.center) to (63.center);
		\draw (67.center) to (37.center);
		\draw (72.center) to (70.center);
		\draw (73.center) to (71.center);
		\draw (74.center) to (76.center);
		\draw (78.center) to (79.center);
		\draw (81.center) to (82.center);
		\draw (81.center) to (80);
		\draw (83) to (80);
		\draw (54.center) to (83);
		\draw (83) to (69);
		\draw (69) to (65);
		\draw (65) to (58.center);
		\draw (57.center) to (62);
		\draw (62) to (66);
		\draw (66) to (80);
		\draw (80) to (50.center);
	\end{pgfonlayer}
\end{tikzpicture}

%% file: figures/multiplexer-option2.tikz
\begin{tikzpicture}
	\begin{pgfonlayer}{nodelayer}
		\node [style=none] (0) at (-1.75, 2.25) {};
		\node [style=none] (1) at (-1.75, 0.25) {};
		\node [style=none] (2) at (-1.75, -0.75) {};
		\node [style=none] (3) at (1.75, 0.25) {};
		\node [style=none] (4) at (1.75, -0.75) {};
		\node [style=none] (5) at (1.75, 2.25) {};
		\node [style=none] (6) at (-0.75, 2.75) {};
		\node [style=none] (7) at (0.75, 2.75) {};
		\node [style=none] (8) at (0.75, -5.25) {};
		\node [style=none] (9) at (-0.75, -5.25) {};
		\node [style=none] (10) at (-0.75, 0.25) {};
		\node [style=none] (11) at (-0.75, 2.25) {};
		\node [style=none] (12) at (-0.75, -0.75) {};
		\node [style=none] (13) at (0.75, 0.25) {};
		\node [style=none] (14) at (0.75, -0.75) {};
		\node [style=none] (15) at (0.75, 2.25) {};
		\node [style=none] (16) at (0, -1.25) {$P$};
		\node [style=cnot ctrl] (17) at (0, 3.5) {};
		\node [style=none] (18) at (1.75, 3.5) {};
		\node [style=none] (19) at (0, 2.75) {};
		\node [style=none] (22) at (1.75, 5.5) {};
		\node [style=none] (25) at (-1.75, 3.5) {};
		\node [style=none] (26) at (-1.75, 5.5) {};
		\node [style=cnot ctrl] (27) at (0, 5.5) {};
		\node [style=none] (28) at (-2.5, 5.5) {$|i_1\rangle$};
		\node [style=none] (29) at (-2.5, 3.5) {$|i_3\rangle$};
		\node [style=none] (30) at (-1.75, 1.25) {};
		\node [style=none] (31) at (1.75, 1.25) {};
		\node [style=none] (32) at (-0.75, 1.25) {};
		\node [style=none] (33) at (0.75, 1.25) {};
		\node [style=none] (67) at (1.75, 4.5) {};
		\node [style=none] (68) at (-1.75, 4.5) {};
		\node [style=cnot ctrl] (69) at (0, 4.5) {};
		\node [style=none] (70) at (-2.5, 4.5) {$|i_2\rangle$};
		\node [style=none] (71) at (5, 2.25) {};
		\node [style=none] (72) at (5, 0.25) {};
		\node [style=none] (73) at (5, -0.75) {};
		\node [style=none] (74) at (10, 0.25) {};
		\node [style=none] (75) at (10, -0.75) {};
		\node [style=none] (76) at (10, 2.25) {};
		\node [style=none] (89) at (10, 3.5) {};
		\node [style=none] (91) at (10, 5.5) {};
		\node [style=none] (92) at (5, 3.5) {};
		\node [style=none] (93) at (5, 5.5) {};
		\node [style=none] (95) at (4.25, 5.5) {$|i_1\rangle$};
		\node [style=none] (96) at (4.25, 3.5) {$|i_3\rangle$};
		\node [style=none] (97) at (5, 1.25) {};
		\node [style=none] (98) at (10, 1.25) {};
		\node [style=none] (101) at (10, 4.5) {};
		\node [style=none] (102) at (5, 4.5) {};
		\node [style=none] (104) at (4.25, 4.5) {$|i_2\rangle$};
		\node [style=none] (105) at (-1.75, -1.75) {};
		\node [style=none] (106) at (-1.75, -3.75) {};
		\node [style=none] (107) at (-1.75, -4.75) {};
		\node [style=none] (108) at (1.75, -3.75) {};
		\node [style=none] (109) at (1.75, -4.75) {};
		\node [style=none] (110) at (1.75, -1.75) {};
		\node [style=none] (111) at (-0.75, -3.75) {};
		\node [style=none] (112) at (-0.75, -1.75) {};
		\node [style=none] (113) at (-0.75, -4.75) {};
		\node [style=none] (114) at (0.75, -3.75) {};
		\node [style=none] (115) at (0.75, -4.75) {};
		\node [style=none] (116) at (0.75, -1.75) {};
		\node [style=none] (118) at (-1.75, -2.75) {};
		\node [style=none] (119) at (1.75, -2.75) {};
		\node [style=none] (120) at (-0.75, -2.75) {};
		\node [style=none] (121) at (0.75, -2.75) {};
		\node [style=none] (122) at (5, -1.75) {};
		\node [style=none] (123) at (5, -3.75) {};
		\node [style=none] (124) at (5, -4.75) {};
		\node [style=none] (125) at (10, -3.75) {};
		\node [style=none] (126) at (10, -4.75) {};
		\node [style=none] (127) at (10, -1.75) {};
		\node [style=none] (128) at (5, -2.75) {};
		\node [style=none] (129) at (10, -2.75) {};
		\node [style=none] (130) at (6, 2.25) {$\times$};
		\node [style=none] (131) at (6, -1.75) {$\times$};
		\node [style=none] (132) at (6.5, 1.25) {$\times$};
		\node [style=none] (133) at (6.5, -2.75) {$\times$};
		\node [style=none] (134) at (7, 0.25) {$\times$};
		\node [style=none] (135) at (7, -3.75) {$\times$};
		\node [style=none] (136) at (7.5, -0.75) {$\times$};
		\node [style=none] (137) at (7.5, -4.75) {$\times$};
		\node [style=none] (138) at (8, 2.25) {$\times$};
		\node [style=none] (139) at (8, 0.25) {$\times$};
		\node [style=none] (140) at (8.5, 1.25) {$\times$};
		\node [style=none] (141) at (8.5, -0.75) {$\times$};
		\node [style=none] (142) at (9, 2.25) {$\times$};
		\node [style=none] (143) at (9, 1.25) {$\times$};
		\node [style=small black dot] (144) at (6, 5.5) {};
		\node [style=small black dot] (145) at (6.5, 5.5) {};
		\node [style=small black dot] (146) at (7, 5.5) {};
		\node [style=small black dot] (147) at (7.5, 5.5) {};
		\node [style=small black dot] (148) at (8, 4.5) {};
		\node [style=small black dot] (149) at (8.5, 4.5) {};
		\node [style=small black dot] (150) at (9, 3.5) {};
		\node [style=small black dot] (151) at (7, 5.5) {};
		\node [style=none] (152) at (3.25, 0) {$=$};
	\end{pgfonlayer}
	\begin{pgfonlayer}{edgelayer}
		\draw (7.center) to (6.center);
		\draw (9.center) to (6.center);
		\draw (9.center) to (8.center);
		\draw (8.center) to (7.center);
		\draw (12.center) to (2.center);
		\draw (1.center) to (10.center);
		\draw (11.center) to (0.center);
		\draw (14.center) to (4.center);
		\draw (3.center) to (13.center);
		\draw (15.center) to (5.center);
		\draw (19.center) to (17);
		\draw (18.center) to (17);
		\draw (25.center) to (17);
		\draw (27) to (26.center);
		\draw (27) to (22.center);
		\draw (27) to (17);
		\draw (32.center) to (30.center);
		\draw (33.center) to (31.center);
		\draw (69) to (68.center);
		\draw (69) to (67.center);
		\draw (113.center) to (107.center);
		\draw (106.center) to (111.center);
		\draw (112.center) to (105.center);
		\draw (115.center) to (109.center);
		\draw (108.center) to (114.center);
		\draw (116.center) to (110.center);
		\draw (120.center) to (118.center);
		\draw (121.center) to (119.center);
		\draw (130.center) to (131.center);
		\draw (132.center) to (133.center);
		\draw (134.center) to (135.center);
		\draw (136.center) to (137.center);
		\draw (138.center) to (139.center);
		\draw (140.center) to (141.center);
		\draw (142.center) to (143.center);
		\draw (144) to (130.center);
		\draw (145) to (132.center);
		\draw (146) to (134.center);
		\draw (147) to (136.center);
		\draw (148) to (138.center);
		\draw (149) to (140.center);
		\draw (150) to (142.center);
		\draw (93.center) to (144);
		\draw (144) to (145);
		\draw (145) to (151);
		\draw (151) to (147);
		\draw (147) to (91.center);
		\draw (101.center) to (149);
		\draw (149) to (148);
		\draw (148) to (102.center);
		\draw (92.center) to (150);
		\draw (150) to (89.center);
		\draw (71.center) to (130.center);
		\draw (130.center) to (138.center);
		\draw (138.center) to (142.center);
		\draw (142.center) to (76.center);
		\draw (97.center) to (132.center);
		\draw (132.center) to (140.center);
		\draw (140.center) to (143.center);
		\draw (143.center) to (98.center);
		\draw (72.center) to (134.center);
		\draw (134.center) to (139.center);
		\draw (139.center) to (74.center);
		\draw (73.center) to (136.center);
		\draw (136.center) to (141.center);
		\draw (141.center) to (75.center);
		\draw (122.center) to (131.center);
		\draw (131.center) to (127.center);
		\draw (128.center) to (133.center);
		\draw (133.center) to (129.center);
		\draw (123.center) to (135.center);
		\draw (135.center) to (125.center);
		\draw (124.center) to (137.center);
		\draw (137.center) to (126.center);
	\end{pgfonlayer}
\end{tikzpicture}

%% file: figures/multiplexer-option2-construction.tikz
\begin{tikzpicture}
	\begin{pgfonlayer}{nodelayer}
		\node [style=none] (0) at (-5.5, 1.75) {};
		\node [style=none] (1) at (-5.5, -1) {};
		\node [style=none] (2) at (-5.5, -2.75) {};
		\node [style=none] (3) at (-2, -1) {};
		\node [style=none] (4) at (-2, -2.75) {};
		\node [style=none] (5) at (-2, 1.75) {};
		\node [style=none] (6) at (-4.5, 2.25) {};
		\node [style=none] (7) at (-3, 2.25) {};
		\node [style=none] (8) at (-3, -3.25) {};
		\node [style=none] (9) at (-4.5, -3.25) {};
		\node [style=none] (10) at (-4.5, -1) {};
		\node [style=none] (11) at (-4.5, 1.75) {};
		\node [style=none] (12) at (-4.5, -2.75) {};
		\node [style=none] (13) at (-3, -1) {};
		\node [style=none] (14) at (-3, -2.75) {};
		\node [style=none] (15) at (-3, 1.75) {};
		\node [style=none] (16) at (-3.75, -0.5) {$P$};
		\node [style=cnot ctrl] (17) at (-3.75, 3) {};
		\node [style=none] (18) at (-2, 3) {};
		\node [style=none] (19) at (-3.75, 2.25) {};
		\node [style=none] (20) at (-5, -1.75) {$\vdots$};
		\node [style=none] (21) at (-2.5, -1.75) {$\vdots$};
		\node [style=none] (22) at (-2, 5.5) {};
		\node [style=none] (23) at (-2.5, 4) {$\vdots$};
		\node [style=none] (24) at (-5, 4) {$\vdots$};
		\node [style=none] (25) at (-5.5, 3) {};
		\node [style=none] (26) at (-5.5, 5.5) {};
		\node [style=cnot ctrl] (27) at (-3.75, 5.5) {};
		\node [style=none] (28) at (-6.25, 5.5) {$|i_1\rangle$};
		\node [style=none] (29) at (-6.25, 3) {$|i_b\rangle$};
		\node [style=none] (30) at (1.75, 1.75) {};
		\node [style=none] (31) at (1.75, 0) {};
		\node [style=none] (32) at (8.75, 0) {};
		\node [style=none] (33) at (8.75, 1.75) {};
		\node [style=none] (34) at (8.75, 3) {};
		\node [style=none] (35) at (2.25, -1.75) {$\vdots$};
		\node [style=none] (36) at (8.25, -1.75) {$\vdots$};
		\node [style=none] (37) at (8.75, 5.5) {};
		\node [style=none] (38) at (1.75, 3) {};
		\node [style=none] (39) at (1.75, 5.5) {};
		\node [style=none] (40) at (1, 5.5) {$|i_1\rangle$};
		\node [style=none] (41) at (1, 3) {$|i_b\rangle$};
		\node [style=none] (43) at (3, 1.75) {$\times$};
		\node [style=none] (44) at (3, -1) {$\times$};
		\node [style=none] (47) at (5, 0) {$\times$};
		\node [style=none] (48) at (5, -2.75) {$\times$};
		\node [style=none] (50) at (-5.5, 0) {};
		\node [style=none] (51) at (-2, 0) {};
		\node [style=none] (52) at (-4.5, 0) {};
		\node [style=none] (53) at (-3, 0) {};
		\node [style=none] (54) at (1.75, -1) {};
		\node [style=none] (55) at (8.75, -1) {};
		\node [style=none] (56) at (1.75, -2.75) {};
		\node [style=none] (57) at (8.75, -2.75) {};
		\node [style=none] (64) at (0, -0.25) {$=$};
		\node [style=none] (65) at (-5, 1) {$\vdots$};
		\node [style=none] (66) at (-2.5, 1) {$\vdots$};
		\node [style=none] (67) at (-2, 4.5) {};
		\node [style=none] (68) at (-5.5, 4.5) {};
		\node [style=cnot ctrl] (69) at (-3.75, 4.5) {};
		\node [style=none] (70) at (-6.25, 4.5) {$|i_2\rangle$};
		\node [style=none] (71) at (8.75, 4.5) {};
		\node [style=none] (72) at (1.75, 4.5) {};
		\node [style=none] (73) at (1, 4.5) {$|i_2\rangle$};
		\node [style=none] (74) at (2.25, 1) {$\vdots$};
		\node [style=none] (75) at (8.25, 1) {$\vdots$};
		\node [style=none] (76) at (6, 2.25) {};
		\node [style=none] (77) at (7.5, 2.25) {};
		\node [style=none] (78) at (7.5, -0.5) {};
		\node [style=none] (79) at (6, -0.5) {};
		\node [style=none] (81) at (6, 1.75) {};
		\node [style=none] (85) at (7.5, 1.75) {};
		\node [style=none] (86) at (6.75, 1) {$P$};
		\node [style=none] (87) at (6.75, 2.25) {};
		\node [style=none] (88) at (6, 0) {};
		\node [style=none] (89) at (7.5, 0) {};
		\node [style=small black dot] (90) at (3, 5.5) {};
		\node [style=small black dot] (91) at (5, 5.5) {};
		\node [style=small black dot] (92) at (6.75, 4.5) {};
		\node [style=small black dot] (93) at (6.75, 3) {};
		\node [style=none] (94) at (4, 1) {$\ddots$};
		\node [style=none] (95) at (4, -1.75) {$\ddots$};
		\node [style=none] (96) at (12.75, 0.5) {};
		\node [style=none] (101) at (16.25, 0.5) {};
		\node [style=none] (102) at (13.75, 1) {};
		\node [style=none] (103) at (15.25, 1) {};
		\node [style=none] (104) at (15.25, -1) {};
		\node [style=none] (105) at (13.75, -1) {};
		\node [style=none] (107) at (13.75, 0.5) {};
		\node [style=none] (111) at (15.25, 0.5) {};
		\node [style=none] (112) at (14.5, 0) {$P$};
		\node [style=cnot ctrl] (113) at (14.5, 1.75) {};
		\node [style=none] (114) at (16.25, 1.75) {};
		\node [style=none] (115) at (14.5, 1) {};
		\node [style=none] (121) at (12.75, 1.75) {};
		\node [style=none] (125) at (12, 1.75) {$|i_1\rangle$};
		\node [style=none] (126) at (12.75, -0.5) {};
		\node [style=none] (127) at (16.25, -0.5) {};
		\node [style=none] (128) at (13.75, -0.5) {};
		\node [style=none] (129) at (15.25, -0.5) {};
		\node [style=none] (130) at (17.5, 0) {$=$};
		\node [style=none] (131) at (19.25, 0.5) {};
		\node [style=none] (132) at (22.75, 0.5) {};
		\node [style=cnot ctrl] (140) at (21, 1.75) {};
		\node [style=none] (141) at (22.75, 1.75) {};
		\node [style=none] (143) at (19.25, 1.75) {};
		\node [style=none] (144) at (18.5, 1.75) {$|i_1\rangle$};
		\node [style=none] (145) at (19.25, -0.5) {};
		\node [style=none] (146) at (22.75, -0.5) {};
		\node [style=none] (147) at (21, 0.5) {$\times$};
		\node [style=none] (148) at (21, -0.5) {$\times$};
		\node [style=none] (149) at (2.25, 4) {$\vdots$};
		\node [style=none] (150) at (8.25, 4) {$\vdots$};
	\end{pgfonlayer}
	\begin{pgfonlayer}{edgelayer}
		\draw (7.center) to (6.center);
		\draw (9.center) to (6.center);
		\draw (9.center) to (8.center);
		\draw (8.center) to (7.center);
		\draw (12.center) to (2.center);
		\draw (1.center) to (10.center);
		\draw (11.center) to (0.center);
		\draw (14.center) to (4.center);
		\draw (3.center) to (13.center);
		\draw (15.center) to (5.center);
		\draw (19.center) to (17);
		\draw (18.center) to (17);
		\draw (25.center) to (17);
		\draw (27) to (26.center);
		\draw (27) to (22.center);
		\draw (43.center) to (44.center);
		\draw (47.center) to (48.center);
		\draw (52.center) to (50.center);
		\draw (53.center) to (51.center);
		\draw (69) to (68.center);
		\draw (69) to (67.center);
		\draw (27) to (69);
		\draw (69) to (17);
		\draw (48.center) to (56.center);
		\draw (44.center) to (54.center);
		\draw (47.center) to (31.center);
		\draw (44.center) to (55.center);
		\draw (48.center) to (57.center);
		\draw (77.center) to (76.center);
		\draw (79.center) to (76.center);
		\draw (79.center) to (78.center);
		\draw (78.center) to (77.center);
		\draw (91) to (47.center);
		\draw (90) to (43.center);
		\draw (39.center) to (90);
		\draw (90) to (91);
		\draw (91) to (37.center);
		\draw (93) to (87.center);
		\draw (92) to (93);
		\draw (92) to (72.center);
		\draw (92) to (71.center);
		\draw (34.center) to (93);
		\draw (93) to (38.center);
		\draw (30.center) to (43.center);
		\draw (43.center) to (81.center);
		\draw (85.center) to (33.center);
		\draw (32.center) to (89.center);
		\draw (88.center) to (47.center);
		\draw (103.center) to (102.center);
		\draw (105.center) to (102.center);
		\draw (105.center) to (104.center);
		\draw (104.center) to (103.center);
		\draw (107.center) to (96.center);
		\draw (111.center) to (101.center);
		\draw (115.center) to (113);
		\draw (114.center) to (113);
		\draw (121.center) to (113);
		\draw (128.center) to (126.center);
		\draw (129.center) to (127.center);
		\draw (141.center) to (140);
		\draw (143.center) to (140);
		\draw (148.center) to (146.center);
		\draw (132.center) to (147.center);
		\draw (147.center) to (131.center);
		\draw (145.center) to (148.center);
		\draw (148.center) to (147.center);
		\draw (147.center) to (140);
	\end{pgfonlayer}
\end{tikzpicture}

%% file: figures/tt-def.tikz
\begin{tikzpicture}
	\begin{pgfonlayer}{nodelayer}
		\node [style=none] (0) at (0, -1.25) {T$_0^\dagger$};
		\node [style=none] (1) at (0, 1.75) {T$_i$};
		\node [style=none] (2) at (-1.25, -0.25) {};
		\node [style=none] (3) at (-1.25, 0.75) {};
		\node [style=none] (4) at (1.25, 0.75) {};
		\node [style=none] (5) at (1.25, -0.25) {};
		\node [style=none] (6) at (0, 0.25) {$\dots$};
		\node [style=none] (11) at (-1.25, 2.75) {};
		\node [style=none] (12) at (1.25, 2.75) {};
		\node [style=none] (13) at (-1.25, -2.25) {};
		\node [style=none] (14) at (1.25, -2.25) {};
		\node [style=none] (15) at (-1.75, 2.75) {};
		\node [style=none] (16) at (-1.75, 0.75) {};
		\node [style=none] (17) at (1.75, 0.75) {};
		\node [style=none] (18) at (1.75, 2.75) {};
		\node [style=none] (19) at (-1.75, -0.25) {};
		\node [style=none] (20) at (-1.75, -2.25) {};
		\node [style=none] (21) at (1.75, -2.25) {};
		\node [style=none] (22) at (1.75, -0.25) {};
	\end{pgfonlayer}
	\begin{pgfonlayer}{edgelayer}
		\draw (5.center) to (4.center);
		\draw (3.center) to (2.center);
		\draw (18.center) to (15.center);
		\draw (15.center) to (16.center);
		\draw (16.center) to (17.center);
		\draw (17.center) to (18.center);
		\draw (22.center) to (19.center);
		\draw (19.center) to (20.center);
		\draw (20.center) to (21.center);
		\draw (21.center) to (22.center);
	\end{pgfonlayer}
\end{tikzpicture}

%% file: figures/tt-oracle/pcrz.tikz
\begin{tikzpicture}
	\begin{pgfonlayer}{nodelayer}
		\node [style=none] (0) at (0, 0.5) {};
		\node [style=none] (1) at (0, -1) {};
		\node [style=none] (2) at (0, -2) {};
		\node [style=none] (3) at (0.25, 0) {$\vdots$};
		\node [style=cnot ctrl] (4) at (1, 1.5) {};
		\node [style=cnot ctrl] (5) at (5.75, -1) {};
		\node [style=gate] (6) at (1, -2) {Z};
		\node [style=gate] (7) at (5.75, -2) {R$_z(\frac{\pi}{2^k})$};
		\node [style=none] (8) at (0, 1.5) {};
		\node [style=cnot ctrl] (9) at (2.25, 0.5) {};
		\node [style=gate] (10) at (2.25, -2) {S};
		\node [style=none] (11) at (3.5, 0) {\rotatebox{10}{$\ddots$}};
		\node [style=none] (12) at (6.75, -2) {};
		\node [style=none] (13) at (6.75, -1) {};
		\node [style=none] (14) at (6.75, 0.5) {};
		\node [style=none] (15) at (6.75, 1.5) {};
		\node [style=none] (19) at (-2, 0) {$=$};
		\node [style=none] (20) at (-7, 0.5) {};
		\node [style=none] (21) at (-7, -1) {};
		\node [style=none] (22) at (-7, -2) {};
		\node [style=none] (23) at (-6.5, 0) {$\vdots$};
		\node [style=none] (28) at (-7, 1.5) {};
		\node [style=none] (32) at (-4, -2) {};
		\node [style=none] (33) at (-4, -1) {};
		\node [style=none] (34) at (-4, 0.5) {};
		\node [style=none] (35) at (-4, 1.5) {};
		\node [style=gate] (39) at (-5.5, -2) {R$_z$};
		\node [style=none] (40) at (0, 2.5) {};
		\node [style=none] (41) at (6.75, 2.5) {};
		\node [style=cnot ctrl] (42) at (1, 2.5) {};
		\node [style=cnot ctrl] (43) at (2.25, 2.5) {};
		\node [style=cnot ctrl] (44) at (5.75, 2.5) {};
		\node [style=cnot ctrl] (45) at (-5.5, 2.5) {};
		\node [style=none] (46) at (-7, 2.5) {};
		\node [style=none] (47) at (-4, 2.5) {};
		\node [style=none] (48) at (-6, 1.75) {};
		\node [style=none] (49) at (-5, 1.75) {};
		\node [style=none] (50) at (-5, -1.25) {};
		\node [style=none] (51) at (-6, -1.25) {};
		\node [style=none] (52) at (-5.5, 0.25) {\rotatebox{-90}{Angle}};
		\node [style=none] (53) at (-5.5, 1.75) {};
		\node [style=none] (54) at (-5.5, -1.25) {};
		\node [style=none] (55) at (-6, 1.5) {};
		\node [style=none] (56) at (-6, 0.5) {};
		\node [style=none] (57) at (-6, -1) {};
		\node [style=none] (58) at (-5, 1.5) {};
		\node [style=none] (59) at (-5, 0.5) {};
		\node [style=none] (60) at (-5, -1) {};
		\node [style=none] (61) at (6.5, 0) {$\vdots$};
		\node [style=none] (62) at (-4.5, 0) {$\vdots$};
		\node [style=none] (63) at (3.5, -2) {$\dots$};
		\node [style=none] (64) at (3, -2) {};
		\node [style=none] (65) at (4, -2) {};
		\node [style=none] (66) at (3.5, 2.5) {$\dots$};
		\node [style=none] (67) at (3, 2.5) {};
		\node [style=none] (68) at (4, 2.5) {};
	\end{pgfonlayer}
	\begin{pgfonlayer}{edgelayer}
		\draw (4) to (6);
		\draw (2.center) to (6);
		\draw (5) to (7);
		\draw (5) to (1.center);
		\draw (9) to (10);
		\draw [in=180, out=0] (6) to (10);
		\draw (0.center) to (9);
		\draw (4) to (8.center);
		\draw (15.center) to (4);
		\draw (9) to (14.center);
		\draw (13.center) to (5);
		\draw [in=180, out=0] (7) to (12.center);
		\draw (44) to (41.center);
		\draw (44) to (5);
		\draw (43) to (9);
		\draw (43) to (42);
		\draw (42) to (4);
		\draw (40.center) to (42);
		\draw (45) to (46.center);
		\draw (47.center) to (45);
		\draw (51.center) to (50.center);
		\draw (50.center) to (49.center);
		\draw (49.center) to (48.center);
		\draw (48.center) to (51.center);
		\draw (53.center) to (45);
		\draw (54.center) to (39);
		\draw (57.center) to (21.center);
		\draw (56.center) to (20.center);
		\draw (55.center) to (28.center);
		\draw (60.center) to (33.center);
		\draw (34.center) to (59.center);
		\draw (58.center) to (35.center);
		\draw (32.center) to (39);
		\draw (39) to (22.center);
		\draw (65.center) to (7);
		\draw (64.center) to (10);
		\draw (68.center) to (44);
		\draw (67.center) to (43);
	\end{pgfonlayer}
\end{tikzpicture}

%% file: figures/tt-oracle/pcansatz1.tikz
\begin{tikzpicture}
	\begin{pgfonlayer}{nodelayer}
		\node [style=none] (0) at (9.5, 0) {};
		\node [style=gate] (12) at (11, 0) {H};
		\node [style=gate] (16) at (12.5, 0) {R$_z$};
		\node [style=cnot ctrl] (17) at (12.5, 1) {};
		\node [style=none] (18) at (9.5, 1) {};
		\node [style=none] (19) at (12.25, 3.75) {};
		\node [style=none] (20) at (12.75, 3.75) {};
		\node [style=none] (21) at (12.75, 1.75) {};
		\node [style=none] (22) at (12.25, 1.75) {};
		\node [style=none] (23) at (12.5, 2.75) {\rotatebox{-90}{\tiny Angle}};
		\node [style=none] (24) at (12.5, 3.75) {};
		\node [style=none] (25) at (12.5, 1.75) {};
		\node [style=none] (26) at (12.25, 3.5) {};
		\node [style=none] (27) at (12.25, 2) {};
		\node [style=none] (28) at (12.75, 3.5) {};
		\node [style=none] (29) at (12.75, 2) {};
		\node [style=none] (58) at (10.5, 3.5) {};
		\node [style=none] (59) at (10.5, 2) {};
		\node [style=none] (60) at (11, 3) {$\vdots$};
		\node [style=none] (62) at (21.5, 2) {};
		\node [style=none] (63) at (21.5, 3.5) {};
		\node [style=none] (66) at (21.5, 1) {};
		\node [style=none] (78) at (21, 3) {$\vdots$};
		\node [style=none] (79) at (21.5, 0) {};
		\node [style=gate] (87) at (14, 0) {H};
		\node [style=cnot targ] (92) at (11.5, 3.5) {};
		\node [style=cnot targ] (93) at (11.5, 2) {};
		\node [style=none] (94) at (11.2, 4.25) {};
		\node [style=none] (95) at (11.8, 4.25) {};
		\node [style=none] (96) at (11.2, 6.25) {};
		\node [style=none] (97) at (11.8, 6.25) {};
		\node [style=none] (98) at (11.5, 4.25) {};
		\node [style=none] (99) at (11.5, 5.25) {\rotatebox{-90}{\tiny Param$_{j1}$}};
		\node [style=none] (100) at (9.5, 6) {};
		\node [style=none] (101) at (9.5, 4.5) {};
		\node [style=none] (102) at (10, 5.5) {$\vdots$};
		\node [style=none] (103) at (11.2, 6) {};
		\node [style=none] (104) at (11.2, 4.5) {};
		\node [style=none] (105) at (11.8, 4.5) {};
		\node [style=none] (106) at (11.8, 6) {};
		\node [style=cnot targ] (107) at (13.5, 3.5) {};
		\node [style=cnot targ] (108) at (13.5, 2) {};
		\node [style=none] (109) at (13.2, 4.25) {};
		\node [style=none] (110) at (13.8, 4.25) {};
		\node [style=none] (111) at (13.2, 6.25) {};
		\node [style=none] (112) at (13.8, 6.25) {};
		\node [style=none] (113) at (13.5, 4.25) {};
		\node [style=none] (114) at (13.5, 5.25) {\rotatebox{-90}{\tiny Param$_{j1}$}};
		\node [style=none] (115) at (13.2, 6) {};
		\node [style=none] (116) at (13.2, 4.5) {};
		\node [style=none] (117) at (13.8, 4.5) {};
		\node [style=none] (118) at (13.8, 6) {};
		\node [style=none] (179) at (21.5, 6) {};
		\node [style=none] (180) at (21.5, 4.5) {};
		\node [style=none] (181) at (21, 5.5) {$\vdots$};
		\node [style=none] (194) at (10, 3.5) {$|0\rangle$};
		\node [style=none] (195) at (10, 2) {$|0\rangle$};
		\node [style=none] (198) at (0.5, 0) {};
		\node [style=none] (202) at (0.5, 1) {};
		\node [style=none] (205) at (1.5, 3.5) {};
		\node [style=none] (206) at (1.5, 2) {};
		\node [style=none] (207) at (2, 3) {$\vdots$};
		\node [style=none] (209) at (6.5, 2) {};
		\node [style=none] (210) at (6.5, 3.5) {};
		\node [style=none] (211) at (6.5, 1) {};
		\node [style=none] (212) at (6, 3) {$\vdots$};
		\node [style=none] (213) at (6.5, 0) {};
		\node [style=none] (214) at (0.5, 6) {};
		\node [style=none] (215) at (0.5, 4.5) {};
		\node [style=none] (216) at (1, 5.5) {$\vdots$};
		\node [style=none] (217) at (6.5, 6) {};
		\node [style=none] (218) at (6.5, 4.5) {};
		\node [style=none] (219) at (6, 5.5) {$\vdots$};
		\node [style=none] (220) at (1, 3.5) {$|0\rangle$};
		\node [style=none] (221) at (1, 2) {$|0\rangle$};
		\node [style=none] (224) at (2, 0.5) {};
		\node [style=none] (225) at (2, -0.5) {};
		\node [style=none] (226) at (5, -0.5) {};
		\node [style=none] (227) at (5, 0.5) {};
		\node [style=none] (228) at (2, 0) {};
		\node [style=none] (235) at (5, 0) {};
		\node [style=none] (236) at (3.5, 0) {\rotatebox{0}{Ansatz$_j^1$}};
		\node [style=cnot ctrl] (237) at (3.5, 1) {};
		\node [style=none] (238) at (3.5, 0.5) {};
		\node [style=none] (239) at (2.5, 6.25) {};
		\node [style=none] (240) at (4.5, 6.25) {};
		\node [style=none] (241) at (2.5, 4.25) {};
		\node [style=none] (242) at (4.5, 4.25) {};
		\node [style=none] (243) at (3.5, 5.25) {Index};
		\node [style=none] (244) at (3.5, 4.25) {};
		\node [style=none] (245) at (2.5, 6) {};
		\node [style=none] (246) at (2.5, 4.5) {};
		\node [style=none] (247) at (4.5, 4.5) {};
		\node [style=none] (248) at (4.5, 6) {};
		\node [style=none] (249) at (8, 2.75) {$=$};
		\node [style=gate] (251) at (15.55, 0) {R$_z$};
		\node [style=cnot ctrl] (252) at (15.55, 1) {};
		\node [style=none] (253) at (15.3, 3.75) {};
		\node [style=none] (254) at (15.8, 3.75) {};
		\node [style=none] (255) at (15.8, 1.75) {};
		\node [style=none] (256) at (15.3, 1.75) {};
		\node [style=none] (257) at (15.55, 2.75) {\rotatebox{-90}{\tiny Angle}};
		\node [style=none] (258) at (15.55, 3.75) {};
		\node [style=none] (259) at (15.55, 1.75) {};
		\node [style=none] (260) at (15.3, 3.5) {};
		\node [style=none] (261) at (15.3, 2) {};
		\node [style=none] (262) at (15.8, 3.5) {};
		\node [style=none] (263) at (15.8, 2) {};
		\node [style=cnot targ] (265) at (14.55, 3.5) {};
		\node [style=cnot targ] (266) at (14.55, 2) {};
		\node [style=none] (267) at (14.25, 4.25) {};
		\node [style=none] (268) at (14.85, 4.25) {};
		\node [style=none] (269) at (14.25, 6.25) {};
		\node [style=none] (270) at (14.85, 6.25) {};
		\node [style=none] (271) at (14.55, 4.25) {};
		\node [style=none] (272) at (14.55, 5.25) {\rotatebox{-90}{\tiny Param$_{j2}$}};
		\node [style=none] (273) at (14.25, 6) {};
		\node [style=none] (274) at (14.25, 4.5) {};
		\node [style=none] (275) at (14.85, 4.5) {};
		\node [style=none] (276) at (14.85, 6) {};
		\node [style=cnot targ] (277) at (16.55, 3.5) {};
		\node [style=cnot targ] (278) at (16.55, 2) {};
		\node [style=none] (279) at (16.25, 4.25) {};
		\node [style=none] (280) at (16.85, 4.25) {};
		\node [style=none] (281) at (16.25, 6.25) {};
		\node [style=none] (282) at (16.85, 6.25) {};
		\node [style=none] (283) at (16.55, 4.25) {};
		\node [style=none] (284) at (16.55, 5.25) {\rotatebox{-90}{\tiny Param$_{j2}$}};
		\node [style=none] (285) at (16.25, 6) {};
		\node [style=none] (286) at (16.25, 4.5) {};
		\node [style=none] (287) at (16.85, 4.5) {};
		\node [style=none] (288) at (16.85, 6) {};
		\node [style=gate] (289) at (17, 0) {H};
		\node [style=gate] (290) at (18.55, 0) {R$_z$};
		\node [style=cnot ctrl] (291) at (18.55, 1) {};
		\node [style=none] (292) at (18.3, 3.75) {};
		\node [style=none] (293) at (18.8, 3.75) {};
		\node [style=none] (294) at (18.8, 1.75) {};
		\node [style=none] (295) at (18.3, 1.75) {};
		\node [style=none] (296) at (18.55, 2.75) {\rotatebox{-90}{\tiny Angle}};
		\node [style=none] (297) at (18.55, 3.75) {};
		\node [style=none] (298) at (18.55, 1.75) {};
		\node [style=none] (299) at (18.3, 3.5) {};
		\node [style=none] (300) at (18.3, 2) {};
		\node [style=none] (301) at (18.8, 3.5) {};
		\node [style=none] (302) at (18.8, 2) {};
		\node [style=gate] (303) at (20, 0) {H};
		\node [style=cnot targ] (304) at (17.55, 3.5) {};
		\node [style=cnot targ] (305) at (17.55, 2) {};
		\node [style=none] (306) at (17.25, 4.25) {};
		\node [style=none] (307) at (17.85, 4.25) {};
		\node [style=none] (308) at (17.25, 6.25) {};
		\node [style=none] (309) at (17.85, 6.25) {};
		\node [style=none] (310) at (17.55, 4.25) {};
		\node [style=none] (311) at (17.55, 5.25) {\rotatebox{-90}{\tiny Param$_{j3}$}};
		\node [style=none] (312) at (17.25, 6) {};
		\node [style=none] (313) at (17.25, 4.5) {};
		\node [style=none] (314) at (17.85, 4.5) {};
		\node [style=none] (315) at (17.85, 6) {};
		\node [style=cnot targ] (316) at (19.55, 3.5) {};
		\node [style=cnot targ] (317) at (19.55, 2) {};
		\node [style=none] (318) at (19.25, 4.25) {};
		\node [style=none] (319) at (19.85, 4.25) {};
		\node [style=none] (320) at (19.25, 6.25) {};
		\node [style=none] (321) at (19.85, 6.25) {};
		\node [style=none] (322) at (19.55, 4.25) {};
		\node [style=none] (323) at (19.55, 5.25) {\rotatebox{-90}{\tiny Param$_{j3}$}};
		\node [style=none] (324) at (19.25, 6) {};
		\node [style=none] (325) at (19.25, 4.5) {};
		\node [style=none] (326) at (19.85, 4.5) {};
		\node [style=none] (327) at (19.85, 6) {};
		\node [style=cnot ctrl] (328) at (11, 1) {};
		\node [style=cnot ctrl] (329) at (14, 1) {};
		\node [style=cnot ctrl] (330) at (17, 1) {};
		\node [style=cnot ctrl] (331) at (20, 1) {};
	\end{pgfonlayer}
	\begin{pgfonlayer}{edgelayer}
		\draw (17) to (18.center);
		\draw (22.center) to (21.center);
		\draw (21.center) to (20.center);
		\draw (20.center) to (19.center);
		\draw (19.center) to (22.center);
		\draw (12) to (16);
		\draw (16) to (17);
		\draw (17) to (25.center);
		\draw (93) to (92);
		\draw [in=90, out=-90] (97.center) to (95.center);
		\draw (95.center) to (94.center);
		\draw (94.center) to (96.center);
		\draw (96.center) to (97.center);
		\draw (98.center) to (92);
		\draw (103.center) to (100.center);
		\draw (101.center) to (104.center);
		\draw (58.center) to (92);
		\draw (93) to (59.center);
		\draw (92) to (26.center);
		\draw (27.center) to (93);
		\draw (108) to (107);
		\draw (112.center) to (110.center);
		\draw (110.center) to (109.center);
		\draw (109.center) to (111.center);
		\draw (111.center) to (112.center);
		\draw (113.center) to (107);
		\draw (116.center) to (105.center);
		\draw (106.center) to (115.center);
		\draw (107) to (28.center);
		\draw (29.center) to (108);
		\draw (227.center) to (224.center);
		\draw (224.center) to (225.center);
		\draw (225.center) to (226.center);
		\draw (226.center) to (227.center);
		\draw (213.center) to (235.center);
		\draw (198.center) to (228.center);
		\draw (237) to (211.center);
		\draw (237) to (202.center);
		\draw (238.center) to (237);
		\draw (242.center) to (240.center);
		\draw (240.center) to (239.center);
		\draw (239.center) to (241.center);
		\draw (241.center) to (242.center);
		\draw (244.center) to (237);
		\draw (248.center) to (217.center);
		\draw (218.center) to (247.center);
		\draw (246.center) to (215.center);
		\draw (214.center) to (245.center);
		\draw (205.center) to (210.center);
		\draw (209.center) to (206.center);
		\draw (87) to (16);
		\draw (256.center) to (255.center);
		\draw (255.center) to (254.center);
		\draw (254.center) to (253.center);
		\draw (253.center) to (256.center);
		\draw (251) to (252);
		\draw (252) to (259.center);
		\draw (266) to (265);
		\draw [in=90, out=-90] (270.center) to (268.center);
		\draw (268.center) to (267.center);
		\draw (267.center) to (269.center);
		\draw (269.center) to (270.center);
		\draw (271.center) to (265);
		\draw (265) to (260.center);
		\draw (261.center) to (266);
		\draw (278) to (277);
		\draw (282.center) to (280.center);
		\draw (280.center) to (279.center);
		\draw (279.center) to (281.center);
		\draw (281.center) to (282.center);
		\draw (283.center) to (277);
		\draw (286.center) to (275.center);
		\draw (276.center) to (285.center);
		\draw (277) to (262.center);
		\draw (263.center) to (278);
		\draw (295.center) to (294.center);
		\draw (294.center) to (293.center);
		\draw (293.center) to (292.center);
		\draw (292.center) to (295.center);
		\draw (289) to (290);
		\draw (290) to (291);
		\draw (291) to (298.center);
		\draw (305) to (304);
		\draw [in=90, out=-90] (309.center) to (307.center);
		\draw (307.center) to (306.center);
		\draw (306.center) to (308.center);
		\draw (308.center) to (309.center);
		\draw (310.center) to (304);
		\draw (304) to (299.center);
		\draw (300.center) to (305);
		\draw (317) to (316);
		\draw (321.center) to (319.center);
		\draw (319.center) to (318.center);
		\draw (318.center) to (320.center);
		\draw (320.center) to (321.center);
		\draw (322.center) to (316);
		\draw (325.center) to (314.center);
		\draw (315.center) to (324.center);
		\draw (316) to (301.center);
		\draw (302.center) to (317);
		\draw (303) to (290);
		\draw (79.center) to (303);
		\draw (66.center) to (291);
		\draw (291) to (252);
		\draw (252) to (17);
		\draw [in=360, out=180] (12) to (0.center);
		\draw (87) to (251);
		\draw (251) to (289);
		\draw (62.center) to (317);
		\draw (316) to (63.center);
		\draw (304) to (277);
		\draw (278) to (305);
		\draw (266) to (108);
		\draw (107) to (265);
		\draw (274.center) to (117.center);
		\draw (118.center) to (273.center);
		\draw (288.center) to (312.center);
		\draw (287.center) to (313.center);
		\draw (326.center) to (180.center);
		\draw (179.center) to (327.center);
		\draw (331) to (303);
		\draw (329) to (87);
		\draw (330) to (289);
		\draw (328) to (12);
	\end{pgfonlayer}
\end{tikzpicture}

%% file: figures/tt-oracle/pcansatz.tikz
\begin{tikzpicture}
	\begin{pgfonlayer}{nodelayer}
		\node [style=none] (0) at (1, 0) {};
		\node [style=none] (1) at (1, -1) {};
		\node [style=none] (2) at (1, -2) {};
		\node [style=none] (3) at (1, -3) {};
		\node [style=gate] (4) at (2, -2) {H};
		\node [style=gate] (5) at (6, -2) {H};
		\node [style=gate] (6) at (2, -1) {H};
		\node [style=gate] (7) at (6, -1) {H};
		\node [style=gate] (8) at (2, 0) {H};
		\node [style=cnot ctrl] (9) at (3, -3) {};
		\node [style=cnot ctrl] (10) at (4, -2) {};
		\node [style=cnot ctrl] (11) at (5, -1) {};
		\node [style=gate] (12) at (6, 0) {H};
		\node [style=cnot ctrl] (13) at (5, 0) {};
		\node [style=cnot ctrl] (14) at (4, -1) {};
		\node [style=cnot ctrl] (15) at (3, -2) {};
		\node [style=gate] (25) at (7.5, 0) {R$_z$};
		\node [style=cnot ctrl] (26) at (7.5, 1) {};
		\node [style=none] (27) at (1, 1) {};
		\node [style=none] (29) at (7.25, 3.75) {};
		\node [style=none] (30) at (7.75, 3.75) {};
		\node [style=none] (31) at (7.75, 1.75) {};
		\node [style=none] (32) at (7.25, 1.75) {};
		\node [style=none] (33) at (7.5, 2.75) {\rotatebox{-90}{\tiny Angle}};
		\node [style=none] (34) at (7.5, 3.75) {};
		\node [style=none] (35) at (7.5, 1.75) {};
		\node [style=none] (37) at (7.25, 3.5) {};
		\node [style=none] (38) at (7.25, 2) {};
		\node [style=none] (40) at (7.75, 3.5) {};
		\node [style=none] (41) at (7.75, 2) {};
		\node [style=none] (47) at (19.5, -1) {};
		\node [style=gate] (50) at (10.5, -1) {R$_z$};
		\node [style=cnot ctrl] (51) at (10.5, 1) {};
		\node [style=none] (54) at (10.25, 3.75) {};
		\node [style=none] (55) at (10.75, 3.75) {};
		\node [style=none] (56) at (10.75, 1.75) {};
		\node [style=none] (57) at (10.25, 1.75) {};
		\node [style=none] (58) at (10.5, 2.75) {\rotatebox{-90}{\tiny Angle}};
		\node [style=none] (59) at (10.5, 3.75) {};
		\node [style=none] (60) at (10.5, 1.75) {};
		\node [style=none] (61) at (10.25, 3.5) {};
		\node [style=none] (62) at (10.25, 2) {};
		\node [style=none] (63) at (10.75, 3.5) {};
		\node [style=none] (64) at (10.75, 2) {};
		\node [style=none] (70) at (19.5, -2) {};
		\node [style=gate] (73) at (13.5, -2) {R$_z$};
		\node [style=cnot ctrl] (74) at (13.5, 1) {};
		\node [style=none] (77) at (13.25, 3.75) {};
		\node [style=none] (78) at (13.75, 3.75) {};
		\node [style=none] (79) at (13.75, 1.75) {};
		\node [style=none] (80) at (13.25, 1.75) {};
		\node [style=none] (81) at (13.5, 2.75) {\rotatebox{-90}{\tiny Angle}};
		\node [style=none] (82) at (13.5, 3.75) {};
		\node [style=none] (83) at (13.5, 1.75) {};
		\node [style=none] (84) at (13.25, 3.5) {};
		\node [style=none] (85) at (13.25, 2) {};
		\node [style=none] (86) at (13.75, 3.5) {};
		\node [style=none] (87) at (13.75, 2) {};
		\node [style=none] (89) at (2, 3.5) {};
		\node [style=none] (90) at (2, 2) {};
		\node [style=none] (92) at (2.5, 3) {$\vdots$};
		\node [style=none] (93) at (19.5, -3) {};
		\node [style=none] (94) at (19.5, 2) {};
		\node [style=none] (95) at (19.5, 3.5) {};
		\node [style=gate] (96) at (16.5, -3) {R$_z$};
		\node [style=cnot ctrl] (97) at (16.5, 1) {};
		\node [style=none] (99) at (19.5, 1) {};
		\node [style=none] (100) at (16.25, 3.75) {};
		\node [style=none] (101) at (16.75, 3.75) {};
		\node [style=none] (102) at (16.75, 1.75) {};
		\node [style=none] (103) at (16.25, 1.75) {};
		\node [style=none] (104) at (16.5, 2.75) {\rotatebox{-90}{\tiny Angle}};
		\node [style=none] (105) at (16.5, 3.75) {};
		\node [style=none] (106) at (16.5, 1.75) {};
		\node [style=none] (107) at (16.25, 3.5) {};
		\node [style=none] (108) at (16.25, 2) {};
		\node [style=none] (109) at (16.75, 3.5) {};
		\node [style=none] (110) at (16.75, 2) {};
		\node [style=none] (111) at (19, 3) {$\vdots$};
		\node [style=none] (112) at (19.5, 0) {};
		\node [style=cnot ctrl] (113) at (4, 1) {};
		\node [style=cnot ctrl] (115) at (5, 1) {};
		\node [style=cnot ctrl] (118) at (2, 1) {};
		\node [style=cnot ctrl] (119) at (3, 1) {};
		\node [style=cnot ctrl] (120) at (6, 1) {};
		\node [style=gate] (121) at (18, -2) {H};
		\node [style=gate] (122) at (18, -1) {H};
		\node [style=gate] (123) at (18, 0) {H};
		\node [style=cnot ctrl] (124) at (18, 1) {};
		\node [style=gate] (125) at (18, -3) {H};
		\node [style=gate] (126) at (2, -3) {H};
		\node [style=gate] (127) at (6, -3) {H};
		\node [style=cnot targ] (128) at (6.5, 3.5) {};
		\node [style=cnot targ] (129) at (6.5, 2) {};
		\node [style=none] (138) at (6.2, 4.25) {};
		\node [style=none] (139) at (6.8, 4.25) {};
		\node [style=none] (140) at (6.2, 6.25) {};
		\node [style=none] (141) at (6.8, 6.25) {};
		\node [style=none] (142) at (6.5, 4.25) {};
		\node [style=none] (143) at (6.5, 5.25) {\rotatebox{-90}{\tiny Param$_{j1}$}};
		\node [style=none] (144) at (1, 6) {};
		\node [style=none] (145) at (1, 4.5) {};
		\node [style=none] (146) at (1.5, 5.5) {$\vdots$};
		\node [style=none] (147) at (6.2, 6) {};
		\node [style=none] (148) at (6.2, 4.5) {};
		\node [style=none] (149) at (6.8, 4.5) {};
		\node [style=none] (150) at (6.8, 6) {};
		\node [style=cnot targ] (151) at (8.5, 3.5) {};
		\node [style=cnot targ] (152) at (8.5, 2) {};
		\node [style=none] (153) at (8.2, 4.25) {};
		\node [style=none] (154) at (8.8, 4.25) {};
		\node [style=none] (155) at (8.2, 6.25) {};
		\node [style=none] (156) at (8.8, 6.25) {};
		\node [style=none] (157) at (8.5, 4.25) {};
		\node [style=none] (158) at (8.5, 5.25) {\rotatebox{-90}{\tiny Param$_{j1}$}};
		\node [style=none] (159) at (8.2, 6) {};
		\node [style=none] (160) at (8.2, 4.5) {};
		\node [style=none] (161) at (8.8, 4.5) {};
		\node [style=none] (162) at (8.8, 6) {};
		\node [style=cnot targ] (163) at (9.5, 3.5) {};
		\node [style=cnot targ] (164) at (9.5, 2) {};
		\node [style=none] (165) at (9.2, 4.25) {};
		\node [style=none] (166) at (9.8, 4.25) {};
		\node [style=none] (167) at (9.2, 6.25) {};
		\node [style=none] (168) at (9.8, 6.25) {};
		\node [style=none] (169) at (9.5, 4.25) {};
		\node [style=none] (170) at (9.5, 5.25) {\rotatebox{-90}{\tiny Param$_{j2}$}};
		\node [style=none] (171) at (9.2, 6) {};
		\node [style=none] (172) at (9.2, 4.5) {};
		\node [style=none] (173) at (9.8, 4.5) {};
		\node [style=none] (174) at (9.8, 6) {};
		\node [style=cnot targ] (175) at (11.5, 3.5) {};
		\node [style=cnot targ] (176) at (11.5, 2) {};
		\node [style=none] (177) at (11.2, 4.25) {};
		\node [style=none] (178) at (11.8, 4.25) {};
		\node [style=none] (179) at (11.2, 6.25) {};
		\node [style=none] (180) at (11.8, 6.25) {};
		\node [style=none] (181) at (11.5, 4.25) {};
		\node [style=none] (182) at (11.5, 5.25) {\rotatebox{-90}{\tiny Param$_{j2}$}};
		\node [style=none] (183) at (11.2, 6) {};
		\node [style=none] (184) at (11.2, 4.5) {};
		\node [style=none] (185) at (11.8, 4.5) {};
		\node [style=none] (186) at (11.8, 6) {};
		\node [style=cnot targ] (187) at (12.5, 3.5) {};
		\node [style=cnot targ] (188) at (12.5, 2) {};
		\node [style=none] (189) at (12.2, 4.25) {};
		\node [style=none] (190) at (12.8, 4.25) {};
		\node [style=none] (191) at (12.2, 6.25) {};
		\node [style=none] (192) at (12.8, 6.25) {};
		\node [style=none] (193) at (12.5, 4.25) {};
		\node [style=none] (194) at (12.5, 5.25) {\rotatebox{-90}{\tiny Param$_{j3}$}};
		\node [style=none] (195) at (12.2, 6) {};
		\node [style=none] (196) at (12.2, 4.5) {};
		\node [style=none] (197) at (12.8, 4.5) {};
		\node [style=none] (198) at (12.8, 6) {};
		\node [style=cnot targ] (199) at (14.5, 3.5) {};
		\node [style=cnot targ] (200) at (14.5, 2) {};
		\node [style=none] (201) at (14.2, 4.25) {};
		\node [style=none] (202) at (14.8, 4.25) {};
		\node [style=none] (203) at (14.2, 6.25) {};
		\node [style=none] (204) at (14.8, 6.25) {};
		\node [style=none] (205) at (14.5, 4.25) {};
		\node [style=none] (206) at (14.5, 5.25) {\rotatebox{-90}{\tiny Param$_{j3}$}};
		\node [style=none] (207) at (14.2, 6) {};
		\node [style=none] (208) at (14.2, 4.5) {};
		\node [style=none] (209) at (14.8, 4.5) {};
		\node [style=none] (210) at (14.8, 6) {};
		\node [style=cnot targ] (211) at (15.5, 3.5) {};
		\node [style=cnot targ] (212) at (15.5, 2) {};
		\node [style=none] (213) at (15.2, 4.25) {};
		\node [style=none] (214) at (15.8, 4.25) {};
		\node [style=none] (215) at (15.2, 6.25) {};
		\node [style=none] (216) at (15.8, 6.25) {};
		\node [style=none] (217) at (15.5, 4.25) {};
		\node [style=none] (218) at (15.45, 5.25) {\rotatebox{-90}{\tiny Param$_{j4}$}};
		\node [style=none] (219) at (15.2, 6) {};
		\node [style=none] (220) at (15.2, 4.5) {};
		\node [style=none] (221) at (15.8, 4.5) {};
		\node [style=none] (222) at (15.8, 6) {};
		\node [style=none] (223) at (19.5, 6) {};
		\node [style=none] (224) at (19.5, 4.5) {};
		\node [style=none] (225) at (19, 5.5) {$\vdots$};
		\node [style=cnot targ] (226) at (17.5, 3.5) {};
		\node [style=cnot targ] (227) at (17.5, 2) {};
		\node [style=none] (228) at (17.2, 4.25) {};
		\node [style=none] (229) at (17.8, 4.25) {};
		\node [style=none] (230) at (17.2, 6.25) {};
		\node [style=none] (231) at (17.8, 6.25) {};
		\node [style=none] (232) at (17.5, 4.25) {};
		\node [style=none] (233) at (17.5, 5.25) {\rotatebox{-90}{\tiny Param$_{j4}$}};
		\node [style=none] (234) at (17.2, 6) {};
		\node [style=none] (235) at (17.2, 4.5) {};
		\node [style=none] (236) at (17.8, 4.5) {};
		\node [style=none] (237) at (17.8, 6) {};
		\node [style=none] (238) at (1.5, 3.5) {$|0\rangle$};
		\node [style=none] (239) at (1.5, 2) {$|0\rangle$};
		\node [style=none] (242) at (-8, 0) {};
		\node [style=none] (243) at (-8, -1) {};
		\node [style=none] (244) at (-8, -2) {};
		\node [style=none] (245) at (-8, -3) {};
		\node [style=none] (260) at (-8, 1) {};
		\node [style=none] (272) at (-2, -1) {};
		\node [style=none] (286) at (-2, -2) {};
		\node [style=none] (300) at (-7, 3.5) {};
		\node [style=none] (301) at (-7, 2) {};
		\node [style=none] (302) at (-6.5, 3) {$\vdots$};
		\node [style=none] (303) at (-2, -3) {};
		\node [style=none] (304) at (-2, 2) {};
		\node [style=none] (305) at (-2, 3.5) {};
		\node [style=none] (308) at (-2, 1) {};
		\node [style=none] (320) at (-2.5, 3) {$\vdots$};
		\node [style=none] (321) at (-2, 0) {};
		\node [style=none] (342) at (-8, 6) {};
		\node [style=none] (343) at (-8, 4.5) {};
		\node [style=none] (344) at (-7.5, 5.5) {$\vdots$};
		\node [style=none] (421) at (-2, 6) {};
		\node [style=none] (422) at (-2, 4.5) {};
		\node [style=none] (423) at (-2.5, 5.5) {$\vdots$};
		\node [style=none] (436) at (-7.5, 3.5) {$|0\rangle$};
		\node [style=none] (437) at (-7.5, 2) {$|0\rangle$};
		\node [style=none] (440) at (-6, 0.25) {};
		\node [style=none] (441) at (-6, -3.25) {};
		\node [style=none] (442) at (-4, -3.25) {};
		\node [style=none] (443) at (-4, 0.25) {};
		\node [style=none] (445) at (-6, 0) {};
		\node [style=none] (446) at (-6, -1) {};
		\node [style=none] (447) at (-6, -2) {};
		\node [style=none] (448) at (-6, -3) {};
		\node [style=none] (449) at (-4, -3) {};
		\node [style=none] (450) at (-4, -2) {};
		\node [style=none] (451) at (-4, -1) {};
		\node [style=none] (452) at (-4, 0) {};
		\node [style=none] (454) at (-5, -1.5) {\rotatebox{-90}{Ansatz$_j^4$}};
		\node [style=cnot ctrl] (455) at (-5, 1) {};
		\node [style=none] (456) at (-5, 0.25) {};
		\node [style=none] (457) at (-6, 6.25) {};
		\node [style=none] (458) at (-4, 6.25) {};
		\node [style=none] (459) at (-6, 4.25) {};
		\node [style=none] (460) at (-4, 4.25) {};
		\node [style=none] (461) at (-5, 5.25) {Index};
		\node [style=none] (462) at (-5, 4.25) {};
		\node [style=none] (463) at (-6, 6) {};
		\node [style=none] (464) at (-6, 4.5) {};
		\node [style=none] (465) at (-4, 4.5) {};
		\node [style=none] (466) at (-4, 6) {};
		\node [style=none] (467) at (-0.5, 1) {$=$};
	\end{pgfonlayer}
	\begin{pgfonlayer}{edgelayer}
		\draw (9) to (15);
		\draw (10) to (14);
		\draw (11) to (13);
		\draw (13) to (8);
		\draw (6) to (1.center);
		\draw (8) to (0.center);
		\draw (4) to (2.center);
		\draw (4) to (15);
		\draw (15) to (5);
		\draw (5) to (10);
		\draw (6) to (14);
		\draw (14) to (7);
		\draw (7) to (11);
		\draw (13) to (12);
		\draw (9) to (3.center);
		\draw (26) to (27.center);
		\draw (32.center) to (31.center);
		\draw (31.center) to (30.center);
		\draw (30.center) to (29.center);
		\draw (29.center) to (32.center);
		\draw (57.center) to (56.center);
		\draw (56.center) to (55.center);
		\draw (55.center) to (54.center);
		\draw (54.center) to (57.center);
		\draw (47.center) to (50);
		\draw (80.center) to (79.center);
		\draw (79.center) to (78.center);
		\draw (78.center) to (77.center);
		\draw (77.center) to (80.center);
		\draw (70.center) to (73);
		\draw (99.center) to (97);
		\draw (103.center) to (102.center);
		\draw (102.center) to (101.center);
		\draw (101.center) to (100.center);
		\draw (100.center) to (103.center);
		\draw (93.center) to (96);
		\draw (12) to (25);
		\draw (11) to (50);
		\draw (73) to (10);
		\draw (96) to (9);
		\draw (26) to (51);
		\draw (51) to (74);
		\draw (74) to (97);
		\draw (112.center) to (25);
		\draw (118) to (119);
		\draw (119) to (120);
		\draw (119) to (15);
		\draw [in=90, out=-90] (120) to (5);
		\draw (113) to (14);
		\draw (115) to (13);
		\draw (118) to (8);
		\draw (8) to (6);
		\draw (6) to (4);
		\draw (124) to (123);
		\draw [in=90, out=-90] (123) to (122);
		\draw (122) to (121);
		\draw (125) to (121);
		\draw (25) to (26);
		\draw (51) to (50);
		\draw (74) to (73);
		\draw (97) to (96);
		\draw (97) to (106.center);
		\draw (74) to (83.center);
		\draw (51) to (60.center);
		\draw (26) to (35.center);
		\draw (127) to (5);
		\draw (4) to (126);
		\draw (113) to (120);
		\draw (120) to (115);
		\draw (115) to (120);
		\draw (120) to (7);
		\draw (120) to (12);
		\draw (129) to (128);
		\draw [in=90, out=-90] (141.center) to (139.center);
		\draw (139.center) to (138.center);
		\draw (138.center) to (140.center);
		\draw (140.center) to (141.center);
		\draw (142.center) to (128);
		\draw (147.center) to (144.center);
		\draw (145.center) to (148.center);
		\draw (89.center) to (128);
		\draw (129) to (90.center);
		\draw (128) to (37.center);
		\draw (38.center) to (129);
		\draw (152) to (151);
		\draw (156.center) to (154.center);
		\draw (154.center) to (153.center);
		\draw (153.center) to (155.center);
		\draw (155.center) to (156.center);
		\draw (157.center) to (151);
		\draw (164) to (163);
		\draw [in=90, out=-90] (168.center) to (166.center);
		\draw (166.center) to (165.center);
		\draw (165.center) to (167.center);
		\draw (167.center) to (168.center);
		\draw (169.center) to (163);
		\draw (176) to (175);
		\draw (180.center) to (178.center);
		\draw (178.center) to (177.center);
		\draw (177.center) to (179.center);
		\draw (179.center) to (180.center);
		\draw (181.center) to (175);
		\draw (188) to (187);
		\draw (192.center) to (190.center);
		\draw (190.center) to (189.center);
		\draw (189.center) to (191.center);
		\draw (191.center) to (192.center);
		\draw (193.center) to (187);
		\draw (200) to (199);
		\draw (204.center) to (202.center);
		\draw (202.center) to (201.center);
		\draw (201.center) to (203.center);
		\draw (203.center) to (204.center);
		\draw (205.center) to (199);
		\draw (212) to (211);
		\draw (216.center) to (214.center);
		\draw (214.center) to (213.center);
		\draw (213.center) to (215.center);
		\draw (215.center) to (216.center);
		\draw (217.center) to (211);
		\draw (220.center) to (209.center);
		\draw (208.center) to (197.center);
		\draw (196.center) to (185.center);
		\draw (184.center) to (173.center);
		\draw (172.center) to (161.center);
		\draw (160.center) to (149.center);
		\draw (150.center) to (159.center);
		\draw (162.center) to (171.center);
		\draw (174.center) to (183.center);
		\draw (186.center) to (195.center);
		\draw (198.center) to (207.center);
		\draw (210.center) to (219.center);
		\draw (227) to (226);
		\draw (231.center) to (229.center);
		\draw (229.center) to (228.center);
		\draw (228.center) to (230.center);
		\draw (230.center) to (231.center);
		\draw (232.center) to (226);
		\draw (221.center) to (235.center);
		\draw (236.center) to (224.center);
		\draw (223.center) to (237.center);
		\draw (234.center) to (222.center);
		\draw (95.center) to (226);
		\draw (226) to (109.center);
		\draw (107.center) to (211);
		\draw (211) to (199);
		\draw (199) to (86.center);
		\draw (84.center) to (187);
		\draw (187) to (175);
		\draw (175) to (63.center);
		\draw (61.center) to (163);
		\draw (163) to (151);
		\draw (151) to (40.center);
		\draw (41.center) to (152);
		\draw (152) to (164);
		\draw (164) to (62.center);
		\draw (64.center) to (176);
		\draw (176) to (188);
		\draw (87.center) to (200);
		\draw (200) to (212);
		\draw (212) to (108.center);
		\draw (110.center) to (227);
		\draw (227) to (94.center);
		\draw (85.center) to (188);
		\draw (443.center) to (440.center);
		\draw (440.center) to (441.center);
		\draw (441.center) to (442.center);
		\draw (442.center) to (443.center);
		\draw (321.center) to (452.center);
		\draw (451.center) to (272.center);
		\draw (286.center) to (450.center);
		\draw (449.center) to (303.center);
		\draw (242.center) to (445.center);
		\draw (446.center) to (243.center);
		\draw (244.center) to (447.center);
		\draw (448.center) to (245.center);
		\draw (455) to (308.center);
		\draw (455) to (260.center);
		\draw (456.center) to (455);
		\draw (460.center) to (458.center);
		\draw (458.center) to (457.center);
		\draw (457.center) to (459.center);
		\draw (459.center) to (460.center);
		\draw (462.center) to (455);
		\draw (466.center) to (421.center);
		\draw (422.center) to (465.center);
		\draw (464.center) to (343.center);
		\draw (342.center) to (463.center);
		\draw (300.center) to (305.center);
		\draw (304.center) to (301.center);
	\end{pgfonlayer}
\end{tikzpicture}

%% file: figures/tt-oracle/pansatz-naive.tikz
\begin{tikzpicture}
	\begin{pgfonlayer}{nodelayer}
		\node [style=none] (0) at (0.5, -3) {};
		\node [style=none] (3) at (0.5, -6) {};
		\node [style=none] (6) at (1.5, -0.5) {};
		\node [style=none] (7) at (1.5, -2) {};
		\node [style=none] (8) at (2, -1) {$\vdots$};
		\node [style=none] (9) at (6.5, -6) {};
		\node [style=none] (10) at (6.5, -2) {};
		\node [style=none] (11) at (6.5, -0.5) {};
		\node [style=none] (12) at (6, -1) {$\vdots$};
		\node [style=none] (13) at (6.5, -3) {};
		\node [style=none] (14) at (0.5, 4.5) {};
		\node [style=none] (15) at (0.5, 3) {};
		\node [style=none] (16) at (1, 4) {$\vdots$};
		\node [style=none] (17) at (6.5, 4.5) {};
		\node [style=none] (18) at (6.5, 3) {};
		\node [style=none] (19) at (6, 4) {$\vdots$};
		\node [style=none] (20) at (1, -0.5) {$|0\rangle$};
		\node [style=none] (21) at (1, -2) {$|0\rangle$};
		\node [style=none] (24) at (2.5, -2.75) {};
		\node [style=none] (25) at (2.5, -6.25) {};
		\node [style=none] (26) at (4.5, -6.25) {};
		\node [style=none] (27) at (4.5, -2.75) {};
		\node [style=none] (28) at (2.5, -3) {};
		\node [style=none] (31) at (2.5, -6) {};
		\node [style=none] (32) at (4.5, -6) {};
		\node [style=none] (35) at (4.5, -3) {};
		\node [style=none] (36) at (3.5, -4.5) {\rotatebox{-90}{Ansatz$_j$}};
		\node [style=none] (37) at (3.5, -2.75) {};
		\node [style=none] (38) at (2.5, 4.75) {};
		\node [style=none] (39) at (4.5, 4.75) {};
		\node [style=none] (40) at (2.5, 2.75) {};
		\node [style=none] (41) at (4.5, 2.75) {};
		\node [style=none] (42) at (3.5, 3.75) {Index};
		\node [style=none] (43) at (3.5, 2.75) {};
		\node [style=none] (44) at (2.5, 4.5) {};
		\node [style=none] (45) at (2.5, 3) {};
		\node [style=none] (46) at (4.5, 3) {};
		\node [style=none] (47) at (4.5, 4.5) {};
		\node [style=none] (48) at (1.5, 2) {};
		\node [style=none] (49) at (1.5, 0.5) {};
		\node [style=none] (50) at (6.5, 0.5) {};
		\node [style=none] (51) at (6.5, 2) {};
		\node [style=none] (52) at (1, 2) {$|0\rangle$};
		\node [style=none] (53) at (1, 0.5) {$|0\rangle$};
		\node [style=none] (56) at (8.5, -1) {$=$};
		\node [style=none] (57) at (10.5, -3) {};
		\node [style=none] (60) at (10.5, -6) {};
		\node [style=none] (63) at (11.5, -0.5) {};
		\node [style=none] (64) at (11.5, -2) {};
		\node [style=none] (65) at (12, -1) {$\vdots$};
		\node [style=none] (66) at (23.5, -6) {};
		\node [style=none] (67) at (23.5, -2) {};
		\node [style=none] (68) at (23.5, -0.5) {};
		\node [style=none] (69) at (23, -1) {$\vdots$};
		\node [style=none] (70) at (23.5, -3) {};
		\node [style=none] (71) at (10.5, 4.5) {};
		\node [style=none] (72) at (10.5, 3) {};
		\node [style=none] (73) at (11, 4) {$\vdots$};
		\node [style=none] (74) at (23.5, 4.5) {};
		\node [style=none] (75) at (23.5, 3) {};
		\node [style=none] (76) at (23, 4) {$\vdots$};
		\node [style=none] (77) at (11, -0.5) {$|0\rangle$};
		\node [style=none] (78) at (11, -2) {$|0\rangle$};
		\node [style=none] (81) at (11.5, 2) {};
		\node [style=none] (82) at (11.5, 0.5) {};
		\node [style=none] (83) at (23.5, 0.5) {};
		\node [style=none] (84) at (23.5, 2) {};
		\node [style=none] (85) at (11, 2) {$|0\rangle$};
		\node [style=none] (86) at (11, 0.5) {$|0\rangle$};
		\node [style=cnot targ] (89) at (12.55, 2) {};
		\node [style=cnot targ] (90) at (12.55, 0.5) {};
		\node [style=none] (91) at (12.25, 2.75) {};
		\node [style=none] (92) at (12.85, 2.75) {};
		\node [style=none] (93) at (12.25, 4.75) {};
		\node [style=none] (94) at (12.85, 4.75) {};
		\node [style=none] (95) at (12.55, 2.75) {};
		\node [style=none] (96) at (12.55, 3.75) {\rotatebox{-90}{\tiny Width$_{j}$}};
		\node [style=none] (97) at (12.25, 4.5) {};
		\node [style=none] (98) at (12.25, 3) {};
		\node [style=none] (99) at (12.85, 3) {};
		\node [style=none] (100) at (12.85, 4.5) {};
		\node [style=cnot targ] (101) at (21.55, 2) {};
		\node [style=cnot targ] (102) at (21.55, 0.5) {};
		\node [style=none] (103) at (21.25, 2.75) {};
		\node [style=none] (104) at (21.85, 2.75) {};
		\node [style=none] (105) at (21.25, 4.75) {};
		\node [style=none] (106) at (21.85, 4.75) {};
		\node [style=none] (107) at (21.55, 2.75) {};
		\node [style=none] (108) at (21.55, 3.75) {\rotatebox{-90}{\tiny Width$_{j}$}};
		\node [style=none] (109) at (21.25, 4.5) {};
		\node [style=none] (110) at (21.25, 3) {};
		\node [style=none] (111) at (21.85, 3) {};
		\node [style=none] (112) at (21.85, 4.5) {};
		\node [style=none] (135) at (14, -2.5) {};
		\node [style=none] (136) at (14, -6.25) {};
		\node [style=none] (137) at (15, -6.25) {};
		\node [style=none] (138) at (15, -2.5) {};
		\node [style=none] (139) at (14, -3) {};
		\node [style=none] (142) at (14, -6) {};
		\node [style=none] (143) at (15, -6) {};
		\node [style=none] (146) at (15, -3) {};
		\node [style=none] (147) at (14.5, -4.5) {\rotatebox{-90}{\tiny Ansatz$_j^d$}};
		\node [style=cnot ctrl] (148) at (14.5, 0.5) {};
		\node [style=none] (149) at (14.5, -2.5) {};
		\node [style=none] (150) at (14, 4.75) {};
		\node [style=none] (151) at (15, 4.75) {};
		\node [style=none] (152) at (14, 2.75) {};
		\node [style=none] (153) at (15, 2.75) {};
		\node [style=none] (154) at (14.5, 3.75) {\rotatebox{-90}{\tiny Index}};
		\node [style=none] (155) at (14.5, 2.75) {};
		\node [style=none] (156) at (14, 4.5) {};
		\node [style=none] (157) at (14, 3) {};
		\node [style=none] (158) at (15, 3) {};
		\node [style=none] (159) at (15, 4.5) {};
		\node [style=none] (160) at (14.5, 4.75) {};
		\node [style=none] (185) at (18.5, -2.5) {};
		\node [style=none] (186) at (18.5, -3.5) {};
		\node [style=none] (187) at (20.5, -3.5) {};
		\node [style=none] (188) at (20.5, -2.5) {};
		\node [style=none] (189) at (18.5, -3) {};
		\node [style=none] (190) at (20.5, -3) {};
		\node [style=none] (191) at (19.5, -3) {\rotatebox{0}{\tiny Ansatz$_j^1$}};
		\node [style=vertex set] (192) at (19.5, 0.5) {};
		\node [style=none] (193) at (19.5, -2.5) {};
		\node [style=none] (194) at (19, 4.75) {};
		\node [style=none] (195) at (20, 4.75) {};
		\node [style=none] (196) at (19, 2.75) {};
		\node [style=none] (197) at (20, 2.75) {};
		\node [style=none] (198) at (19.5, 3.75) {\rotatebox{-90}{\tiny Index}};
		\node [style=none] (199) at (19.5, 2.75) {};
		\node [style=none] (200) at (19, 4.5) {};
		\node [style=none] (201) at (19, 3) {};
		\node [style=none] (202) at (20, 3) {};
		\node [style=none] (203) at (20, 4.5) {};
		\node [style=none] (204) at (19.5, 4.75) {};
		\node [style=none] (231) at (23, -3.2) {};
		\node [style=none] (232) at (23.2, -2.8) {};
		\node [style=none] (237) at (23, -6.2) {};
		\node [style=none] (238) at (23.2, -5.8) {};
		\node [style=none] (239) at (10.75, -3.2) {};
		\node [style=none] (240) at (10.95, -2.8) {};
		\node [style=none] (245) at (10.75, -6.2) {};
		\node [style=none] (246) at (10.95, -5.8) {};
		\node [style=none] (247) at (0.75, -3.2) {};
		\node [style=none] (248) at (0.95, -2.8) {};
		\node [style=none] (253) at (0.75, -6.2) {};
		\node [style=none] (254) at (0.95, -5.8) {};
		\node [style=none] (255) at (6, -3.2) {};
		\node [style=none] (256) at (6.2, -2.8) {};
		\node [style=none] (261) at (6, -6.2) {};
		\node [style=none] (262) at (6.2, -5.8) {};
		\node [style=cnot ctrl] (263) at (14.5, 2) {};
		\node [style=vertex set] (265) at (19.5, 2) {};
		\node [style=none] (266) at (1, -4.25) {$\vdots$};
		\node [style=none] (267) at (6, -4.25) {$\vdots$};
		\node [style=none] (268) at (23, -4.25) {$\vdots$};
		\node [style=none] (269) at (11, -4.25) {$\vdots$};
		\node [style=none] (270) at (17.5, -4) {\rotatebox{70}{$\ddots$}};
		\node [style=none] (271) at (18, 3.75) {$\dots$};
		\node [style=none] (272) at (18, 1.25) {$\dots$};
		\node [style=none] (273) at (16, -2.5) {};
		\node [style=none] (274) at (16, -5.25) {};
		\node [style=none] (275) at (17, -5.25) {};
		\node [style=none] (276) at (17, -2.5) {};
		\node [style=none] (277) at (16, -3) {};
		\node [style=none] (278) at (16, -5) {};
		\node [style=none] (279) at (17, -5) {};
		\node [style=none] (280) at (17, -3) {};
		\node [style=none] (281) at (16.5, -4) {\rotatebox{-90}{\tiny Ansatz$_j^{d-1}$}};
		\node [style=vertex set] (282) at (16.5, 0.5) {};
		\node [style=none] (283) at (16.5, -2.5) {};
		\node [style=none] (284) at (16, 4.75) {};
		\node [style=none] (285) at (17, 4.75) {};
		\node [style=none] (286) at (16, 2.75) {};
		\node [style=none] (287) at (17, 2.75) {};
		\node [style=none] (288) at (16.5, 3.75) {\rotatebox{-90}{\tiny Index}};
		\node [style=none] (289) at (16.5, 2.75) {};
		\node [style=none] (290) at (16, 4.5) {};
		\node [style=none] (291) at (16, 3) {};
		\node [style=none] (292) at (17, 3) {};
		\node [style=none] (293) at (17, 4.5) {};
		\node [style=none] (294) at (16.5, 4.75) {};
		\node [style=cnot ctrl] (295) at (16.5, 2) {};
		\node [style=none] (296) at (10.5, -5) {};
		\node [style=none] (297) at (23.5, -5) {};
		\node [style=none] (300) at (14, -5) {};
		\node [style=none] (301) at (15, -5) {};
		\node [style=none] (302) at (23, -5.2) {};
		\node [style=none] (303) at (23.2, -4.8) {};
		\node [style=none] (304) at (10.75, -5.2) {};
		\node [style=none] (305) at (10.95, -4.8) {};
		\node [style=none] (306) at (2, 1.5) {$\vdots$};
		\node [style=none] (307) at (6, 1.5) {$\vdots$};
		\node [style=none] (308) at (12, 1.5) {$\vdots$};
		\node [style=none] (309) at (23, 1.5) {$\vdots$};
	\end{pgfonlayer}
	\begin{pgfonlayer}{edgelayer}
		\draw (27.center) to (24.center);
		\draw (24.center) to (25.center);
		\draw (25.center) to (26.center);
		\draw (26.center) to (27.center);
		\draw (13.center) to (35.center);
		\draw (32.center) to (9.center);
		\draw (0.center) to (28.center);
		\draw (31.center) to (3.center);
		\draw (41.center) to (39.center);
		\draw (39.center) to (38.center);
		\draw (38.center) to (40.center);
		\draw (40.center) to (41.center);
		\draw (47.center) to (17.center);
		\draw (18.center) to (46.center);
		\draw (45.center) to (15.center);
		\draw (14.center) to (44.center);
		\draw (6.center) to (11.center);
		\draw (10.center) to (7.center);
		\draw (48.center) to (51.center);
		\draw (50.center) to (49.center);
		\draw (43.center) to (37.center);
		\draw (90) to (89);
		\draw [in=90, out=-90] (94.center) to (92.center);
		\draw (92.center) to (91.center);
		\draw (91.center) to (93.center);
		\draw (93.center) to (94.center);
		\draw (95.center) to (89);
		\draw (98.center) to (72.center);
		\draw (71.center) to (97.center);
		\draw (102) to (101);
		\draw [in=90, out=-90] (106.center) to (104.center);
		\draw (104.center) to (103.center);
		\draw (103.center) to (105.center);
		\draw (105.center) to (106.center);
		\draw (107.center) to (101);
		\draw (112.center) to (74.center);
		\draw (75.center) to (111.center);
		\draw (102) to (83.center);
		\draw (81.center) to (89);
		\draw (90) to (82.center);
		\draw (101) to (84.center);
		\draw (138.center) to (135.center);
		\draw (135.center) to (136.center);
		\draw (136.center) to (137.center);
		\draw (137.center) to (138.center);
		\draw (153.center) to (151.center);
		\draw (151.center) to (150.center);
		\draw (150.center) to (152.center);
		\draw (152.center) to (153.center);
		\draw (188.center) to (185.center);
		\draw (185.center) to (186.center);
		\draw (186.center) to (187.center);
		\draw (187.center) to (188.center);
		\draw (197.center) to (195.center);
		\draw (195.center) to (194.center);
		\draw (194.center) to (196.center);
		\draw (196.center) to (197.center);
		\draw (139.center) to (57.center);
		\draw (142.center) to (60.center);
		\draw (148) to (149.center);
		\draw (192) to (193.center);
		\draw (202.center) to (110.center);
		\draw (109.center) to (203.center);
		\draw (156.center) to (100.center);
		\draw (99.center) to (157.center);
		\draw (190.center) to (70.center);
		\draw (68.center) to (63.center);
		\draw (64.center) to (67.center);
		\draw (232.center) to (231.center);
		\draw (238.center) to (237.center);
		\draw (240.center) to (239.center);
		\draw (246.center) to (245.center);
		\draw (248.center) to (247.center);
		\draw (254.center) to (253.center);
		\draw (256.center) to (255.center);
		\draw (262.center) to (261.center);
		\draw (89) to (263);
		\draw (90) to (148);
		\draw (148) to (192);
		\draw (263) to (265);
		\draw (265) to (101);
		\draw (102) to (192);
		\draw (192) to (265);
		\draw (265) to (199.center);
		\draw (155.center) to (263);
		\draw (263) to (148);
		\draw (276.center) to (273.center);
		\draw (273.center) to (274.center);
		\draw (274.center) to (275.center);
		\draw (275.center) to (276.center);
		\draw (287.center) to (285.center);
		\draw (285.center) to (284.center);
		\draw (284.center) to (286.center);
		\draw (286.center) to (287.center);
		\draw (282) to (283.center);
		\draw (289.center) to (295);
		\draw (295) to (282);
		\draw (159.center) to (290.center);
		\draw (158.center) to (291.center);
		\draw (292.center) to (201.center);
		\draw (200.center) to (293.center);
		\draw (146.center) to (277.center);
		\draw (143.center) to (66.center);
		\draw (300.center) to (296.center);
		\draw (303.center) to (302.center);
		\draw (305.center) to (304.center);
		\draw (301.center) to (278.center);
		\draw (279.center) to (297.center);
		\draw (280.center) to (189.center);
	\end{pgfonlayer}
\end{tikzpicture}

%% file: figures/tt-oracle/unary-iteration-0.tikz
\begin{tikzpicture}
	\begin{pgfonlayer}{nodelayer}
		\node [style=vertex set] (0) at (-11, 2) {};
		\node [style=vertex set] (1) at (-11, -0.5) {};
		\node [style=vertex set] (2) at (-9, 2) {};
		\node [style=cnot ctrl] (3) at (-9, -0.5) {};
		\node [style=cnot ctrl] (4) at (-7, 2) {};
		\node [style=cnot ctrl] (5) at (-7, -0.5) {};
		\node [style=gate] (6) at (-11, -1.5) {$U_{0\cdots 0}$};
		\node [style=gate] (7) at (-9, -1.5) {$U_{0\cdots 1}$};
		\node [style=gate] (8) at (-7, -1.5) {$U_{1\cdots 1}$};
		\node [style=none] (9) at (-6, 2) {};
		\node [style=none] (10) at (-6, -0.5) {};
		\node [style=none] (11) at (-12, -0.5) {};
		\node [style=none] (12) at (-12, 2) {};
		\node [style=none] (13) at (-12, -1.5) {};
		\node [style=none] (14) at (-6, -1.5) {};
		\node [style=none] (15) at (-8.5, -0.5) {};
		\node [style=none] (16) at (-7.5, -0.5) {};
		\node [style=none] (17) at (-8.5, 2) {};
		\node [style=none] (18) at (-7.5, 2) {};
		\node [style=none] (19) at (-8, -0.5) {$\cdots$};
		\node [style=none] (20) at (-8, 2) {$\cdots$};
		\node [style=none] (29) at (-11.5, 0.5) {$\vdots$};
		\node [style=none] (30) at (-6.5, 0.5) {$\vdots$};
		\node [style=none] (31) at (-5, 0.25) {$=$};
		\node [style=none] (35) at (-4, 2) {};
		\node [style=vertex set] (37) at (-11, 1) {};
		\node [style=vertex set] (38) at (-9, 1) {};
		\node [style=cnot ctrl] (39) at (-7, 1) {};
		\node [style=none] (40) at (-6, 1) {};
		\node [style=none] (41) at (-12, 1) {};
		\node [style=none] (42) at (-8.5, 1) {};
		\node [style=none] (43) at (-7.5, 1) {};
		\node [style=none] (44) at (-8, 1) {$\cdots$};
		\node [style=vertex set] (45) at (-2, -0.5) {};
		\node [style=cnot ctrl] (47) at (0.5, -0.5) {};
		\node [style=gate] (48) at (-2, -1.5) {$U_{1\cdots 0}$};
		\node [style=gate] (50) at (0.5, -1.5) {$U_{1\cdots 1}$};
		\node [style=none] (51) at (7, -0.5) {};
		\node [style=none] (52) at (-4, -0.5) {};
		\node [style=none] (53) at (-4, -1.5) {};
		\node [style=none] (54) at (7, -1.5) {};
		\node [style=none] (55) at (-1.5, -0.5) {};
		\node [style=none] (56) at (0, -0.5) {};
		\node [style=none] (57) at (-0.75, -0.5) {$\cdots$};
		\node [style=none] (58) at (-3, 0.5) {$\vdots$};
		\node [style=none] (59) at (6.5, 0.5) {$\vdots$};
		\node [style=vertex set] (60) at (-2, 1) {};
		\node [style=cnot ctrl] (62) at (0.5, 1) {};
		\node [style=none] (63) at (7, 1) {};
		\node [style=none] (64) at (-4, 1) {};
		\node [style=none] (65) at (-1.5, 1) {};
		\node [style=none] (66) at (0, 1) {};
		\node [style=none] (67) at (-0.75, 1) {$\cdots$};
		\node [style=cnot ctrl] (68) at (-2, 2) {};
		\node [style=cnot ctrl] (69) at (0.5, 2) {};
		\node [style=none] (124) at (7, 2) {};
		\node [style=cnot targ] (125) at (1.5, 2) {};
		\node [style=vertex set] (126) at (2.5, -0.5) {};
		\node [style=cnot ctrl] (127) at (5, -0.5) {};
		\node [style=gate] (128) at (2.5, -1.5) {$U_{0\cdots 0}$};
		\node [style=gate] (129) at (5, -1.5) {$U_{0\cdots 1}$};
		\node [style=none] (130) at (3, -0.5) {};
		\node [style=none] (131) at (4.5, -0.5) {};
		\node [style=none] (132) at (3.75, -0.5) {$\cdots$};
		\node [style=vertex set] (133) at (2.5, 1) {};
		\node [style=cnot ctrl] (134) at (5, 1) {};
		\node [style=none] (135) at (3, 1) {};
		\node [style=none] (136) at (4.5, 1) {};
		\node [style=none] (137) at (3.75, 1) {$\cdots$};
		\node [style=cnot ctrl] (138) at (2.5, 2) {};
		\node [style=cnot ctrl] (139) at (5, 2) {};
		\node [style=cnot targ] (140) at (6, 2) {};
		\node [style=none] (141) at (-1.5, 2) {};
		\node [style=none] (142) at (0, 2) {};
		\node [style=none] (143) at (-0.75, 2) {$\cdots$};
		\node [style=none] (144) at (3, 2) {};
		\node [style=none] (145) at (4.5, 2) {};
		\node [style=none] (146) at (3.75, 2) {$\cdots$};
	\end{pgfonlayer}
	\begin{pgfonlayer}{edgelayer}
		\draw (0) to (1);
		\draw (1) to (6);
		\draw (2) to (3);
		\draw (3) to (7);
		\draw (4) to (5);
		\draw (5) to (8);
		\draw (3) to (1);
		\draw (0) to (2);
		\draw (14.center) to (8);
		\draw (13.center) to (6);
		\draw (6) to (7);
		\draw (7) to (8);
		\draw (10.center) to (5);
		\draw (1) to (11.center);
		\draw (12.center) to (0);
		\draw (4) to (9.center);
		\draw (18.center) to (4);
		\draw (17.center) to (2);
		\draw (3) to (15.center);
		\draw (16.center) to (5);
		\draw (38) to (37);
		\draw (40.center) to (39);
		\draw (37) to (41.center);
		\draw (38) to (42.center);
		\draw (43.center) to (39);
		\draw (45) to (48);
		\draw (47) to (50);
		\draw (53.center) to (48);
		\draw (45) to (52.center);
		\draw (56.center) to (47);
		\draw (60) to (64.center);
		\draw (66.center) to (62);
		\draw (60) to (45);
		\draw (62) to (47);
		\draw (45) to (55.center);
		\draw (65.center) to (60);
		\draw (48) to (50);
		\draw (68) to (60);
		\draw (69) to (62);
		\draw (126) to (128);
		\draw (127) to (129);
		\draw (131.center) to (127);
		\draw (136.center) to (134);
		\draw (133) to (126);
		\draw (134) to (127);
		\draw (126) to (130.center);
		\draw (135.center) to (133);
		\draw (128) to (129);
		\draw (138) to (133);
		\draw (139) to (134);
		\draw (35.center) to (68);
		\draw (69) to (125);
		\draw (125) to (138);
		\draw (140) to (139);
		\draw (140) to (124.center);
		\draw (63.center) to (134);
		\draw (62) to (133);
		\draw (47) to (126);
		\draw (127) to (51.center);
		\draw (54.center) to (129);
		\draw (50) to (128);
		\draw (142.center) to (69);
		\draw (141.center) to (68);
		\draw (145.center) to (139);
		\draw (144.center) to (138);
	\end{pgfonlayer}
\end{tikzpicture}

%% file: figures/tt-oracle/unary-iteration-2.tikz
\begin{tikzpicture}
	\begin{pgfonlayer}{nodelayer}
		\node [style=vertex set] (0) at (-8, 1) {};
		\node [style=vertex set] (1) at (-8, -2.5) {};
		\node [style=vertex set] (2) at (-6, 1) {};
		\node [style=cnot ctrl] (3) at (-6, -2.5) {};
		\node [style=cnot ctrl] (4) at (-4, 1) {};
		\node [style=cnot ctrl] (5) at (-4, -2.5) {};
		\node [style=gate] (6) at (-8, -3.5) {$U_{0\cdots 0}$};
		\node [style=gate] (7) at (-6, -3.5) {$U_{0\cdots 1}$};
		\node [style=gate] (8) at (-4, -3.5) {$U_{1\cdots 1}$};
		\node [style=none] (9) at (-3, 1) {};
		\node [style=none] (10) at (-3, -2.5) {};
		\node [style=none] (11) at (-9, -2.5) {};
		\node [style=none] (12) at (-9, 1) {};
		\node [style=none] (13) at (-9, -3.5) {};
		\node [style=none] (14) at (-3, -3.5) {};
		\node [style=none] (15) at (-5.5, -2.5) {};
		\node [style=none] (16) at (-4.5, -2.5) {};
		\node [style=none] (17) at (-5.5, 1) {};
		\node [style=none] (18) at (-4.5, 1) {};
		\node [style=none] (19) at (-5, -2.5) {$\cdots$};
		\node [style=none] (20) at (-5, 1) {$\cdots$};
		\node [style=none] (29) at (-8.5, -1.5) {$\vdots$};
		\node [style=none] (30) at (-3.5, -1.5) {$\vdots$};
		\node [style=vertex set] (31) at (-8, -1) {};
		\node [style=vertex set] (32) at (-6, -1) {};
		\node [style=cnot ctrl] (33) at (-4, -1) {};
		\node [style=none] (34) at (-3, -1) {};
		\node [style=none] (35) at (-9, -1) {};
		\node [style=none] (36) at (-5.5, -1) {};
		\node [style=none] (37) at (-4.5, -1) {};
		\node [style=none] (38) at (-5, -1) {$\cdots$};
		\node [style=cnot ctrl] (39) at (-8, 2) {};
		\node [style=cnot ctrl] (40) at (-6, 2) {};
		\node [style=cnot ctrl] (41) at (-4, 2) {};
		\node [style=none] (42) at (-3, 2) {};
		\node [style=none] (43) at (-9, 2) {};
		\node [style=none] (44) at (-5.5, 2) {};
		\node [style=none] (45) at (-4.5, 2) {};
		\node [style=none] (46) at (-5, 2) {$\cdots$};
		\node [style=vertex set] (48) at (1.75, -2.5) {};
		\node [style=cnot ctrl] (52) at (3.75, -2.5) {};
		\node [style=gate] (53) at (1.75, -3.5) {$U_{1\cdots 0}$};
		\node [style=gate] (55) at (3.75, -3.5) {$U_{1\cdots 1}$};
		\node [style=none] (57) at (10.75, -2.5) {};
		\node [style=none] (58) at (-1, -2.5) {};
		\node [style=none] (60) at (-1, -3.5) {};
		\node [style=none] (61) at (10.75, -3.5) {};
		\node [style=none] (62) at (2.25, -2.5) {};
		\node [style=none] (63) at (3.25, -2.5) {};
		\node [style=none] (66) at (2.75, -2.5) {$\cdots$};
		\node [style=none] (68) at (-0.25, -1.5) {$\vdots$};
		\node [style=none] (69) at (10.25, -1.5) {$\vdots$};
		\node [style=vertex set] (70) at (1.75, -1) {};
		\node [style=cnot ctrl] (72) at (3.75, -1) {};
		\node [style=none] (73) at (10.75, -1) {};
		\node [style=none] (74) at (-1, -1) {};
		\node [style=none] (75) at (2.25, -1) {};
		\node [style=none] (76) at (3.25, -1) {};
		\node [style=none] (77) at (2.75, -1) {$\cdots$};
		\node [style=cnot ctrl] (78) at (1.75, 0) {};
		\node [style=cnot ctrl] (80) at (3.75, 0) {};
		\node [style=cnot targ] (82) at (0.75, 0) {};
		\node [style=none] (83) at (2.25, 0) {};
		\node [style=none] (84) at (3.25, 0) {};
		\node [style=none] (85) at (2.75, 0) {$\cdots$};
		\node [style=none] (86) at (-2, -0.5) {$=$};
		\node [style=cnot ctrl] (87) at (0.75, 2) {};
		\node [style=cnot ctrl] (88) at (0.75, 1) {};
		\node [style=none] (90) at (-1, 2) {};
		\node [style=none] (91) at (-1, 1) {};
		\node [style=vertex set] (92) at (5.75, -2.5) {};
		\node [style=cnot ctrl] (93) at (7.75, -2.5) {};
		\node [style=gate] (94) at (5.75, -3.5) {$U_{0\cdots 0}$};
		\node [style=gate] (95) at (7.75, -3.5) {$U_{0\cdots 1}$};
		\node [style=none] (96) at (6.25, -2.5) {};
		\node [style=none] (97) at (7.25, -2.5) {};
		\node [style=none] (98) at (6.75, -2.5) {$\cdots$};
		\node [style=vertex set] (99) at (5.75, -1) {};
		\node [style=cnot ctrl] (100) at (7.75, -1) {};
		\node [style=none] (101) at (6.25, -1) {};
		\node [style=none] (102) at (7.25, -1) {};
		\node [style=none] (103) at (6.75, -1) {$\cdots$};
		\node [style=cnot ctrl] (104) at (5.75, 0) {};
		\node [style=cnot ctrl] (105) at (7.75, 0) {};
		\node [style=none] (106) at (6.25, 0) {};
		\node [style=none] (107) at (7.25, 0) {};
		\node [style=none] (108) at (6.75, 0) {$\cdots$};
		\node [style=cnot targ] (109) at (4.75, 0) {};
		\node [style=cnot ctrl] (110) at (4.75, 2) {};
		\node [style=cnot targ] (111) at (8.75, 0) {};
		\node [style=cnot ctrl] (112) at (8.75, 2) {};
		\node [style=cnot ctrl] (113) at (9.75, 2) {};
		\node [style=cnot ctrl] (114) at (9.75, 1) {};
		\node [style=cnot targ] (115) at (9.75, 0) {};
		\node [style=none] (116) at (10.75, 2) {};
		\node [style=none] (117) at (10.75, 1) {};
		\node [style=none] (118) at (10.75, 0) {};
		\node [style=none] (119) at (0, 0) {};
		\node [style=none] (120) at (-0.5, 0) {$|0\rangle$};
	\end{pgfonlayer}
	\begin{pgfonlayer}{edgelayer}
		\draw (0) to (1);
		\draw (1) to (6);
		\draw (2) to (3);
		\draw (3) to (7);
		\draw (4) to (5);
		\draw (5) to (8);
		\draw (3) to (1);
		\draw (0) to (2);
		\draw (14.center) to (8);
		\draw (13.center) to (6);
		\draw (6) to (7);
		\draw (7) to (8);
		\draw (10.center) to (5);
		\draw (1) to (11.center);
		\draw (12.center) to (0);
		\draw (4) to (9.center);
		\draw (18.center) to (4);
		\draw (17.center) to (2);
		\draw (3) to (15.center);
		\draw (16.center) to (5);
		\draw (32) to (31);
		\draw (34.center) to (33);
		\draw (31) to (35.center);
		\draw (32) to (36.center);
		\draw (37.center) to (33);
		\draw (42.center) to (41);
		\draw (41) to (45.center);
		\draw (44.center) to (40);
		\draw (40) to (39);
		\draw (39) to (43.center);
		\draw (39) to (0);
		\draw (40) to (2);
		\draw (41) to (4);
		\draw (48) to (53);
		\draw (52) to (55);
		\draw (60.center) to (53);
		\draw (48) to (58.center);
		\draw (63.center) to (52);
		\draw (70) to (74.center);
		\draw (76.center) to (72);
		\draw (80) to (84.center);
		\draw (78) to (82);
		\draw (87) to (88);
		\draw (91.center) to (88);
		\draw (87) to (90.center);
		\draw (70) to (75.center);
		\draw (83.center) to (78);
		\draw (48) to (62.center);
		\draw (55) to (53);
		\draw (78) to (70);
		\draw (70) to (48);
		\draw (80) to (72);
		\draw (72) to (52);
		\draw (82) to (88);
		\draw (92) to (94);
		\draw (93) to (95);
		\draw (97.center) to (93);
		\draw (102.center) to (100);
		\draw (105) to (107.center);
		\draw (99) to (101.center);
		\draw (106.center) to (104);
		\draw (92) to (96.center);
		\draw (95) to (94);
		\draw (104) to (99);
		\draw (99) to (92);
		\draw (105) to (100);
		\draw (100) to (93);
		\draw (110) to (109);
		\draw (110) to (87);
		\draw (80) to (109);
		\draw (109) to (104);
		\draw (112) to (111);
		\draw (113) to (114);
		\draw (114) to (115);
		\draw (117.center) to (114);
		\draw (113) to (116.center);
		\draw (115) to (111);
		\draw (113) to (112);
		\draw (111) to (105);
		\draw (112) to (110);
		\draw (73.center) to (100);
		\draw (93) to (57.center);
		\draw (61.center) to (95);
		\draw (92) to (52);
		\draw (72) to (99);
		\draw (55) to (94);
		\draw (88) to (114);
		\draw (115) to (118.center);
		\draw (119.center) to (82);
	\end{pgfonlayer}
\end{tikzpicture}

%% file: figures/tt-oracle/pansatz.tikz
\begin{tikzpicture}
	\begin{pgfonlayer}{nodelayer}
		\node [style=none] (0) at (-3, -3) {};
		\node [style=none] (1) at (-3, -4) {};
		\node [style=none] (2) at (-3, -5) {};
		\node [style=none] (3) at (-3, -6) {};
		\node [style=none] (5) at (3, -4) {};
		\node [style=none] (6) at (3, -5) {};
		\node [style=none] (7) at (-2, -0.5) {};
		\node [style=none] (8) at (-2, -2) {};
		\node [style=none] (9) at (-1.5, -1) {$\vdots$};
		\node [style=none] (10) at (3, -6) {};
		\node [style=none] (11) at (3, -2) {};
		\node [style=none] (12) at (3, -0.5) {};
		\node [style=none] (14) at (2.5, -1) {$\vdots$};
		\node [style=none] (15) at (3, -3) {};
		\node [style=none] (16) at (-3, 5) {};
		\node [style=none] (17) at (-3, 3.5) {};
		\node [style=none] (18) at (-2.5, 4.5) {$\vdots$};
		\node [style=none] (19) at (3, 5) {};
		\node [style=none] (20) at (3, 3.5) {};
		\node [style=none] (21) at (2.5, 4.5) {$\vdots$};
		\node [style=none] (22) at (-2.5, -0.5) {$|0\rangle$};
		\node [style=none] (23) at (-2.5, -2) {$|0\rangle$};
		\node [style=none] (26) at (-1, -2.75) {};
		\node [style=none] (27) at (-1, -6.25) {};
		\node [style=none] (28) at (1, -6.25) {};
		\node [style=none] (29) at (1, -2.75) {};
		\node [style=none] (30) at (-1, -3) {};
		\node [style=none] (31) at (-1, -4) {};
		\node [style=none] (32) at (-1, -5) {};
		\node [style=none] (33) at (-1, -6) {};
		\node [style=none] (34) at (1, -6) {};
		\node [style=none] (35) at (1, -5) {};
		\node [style=none] (36) at (1, -4) {};
		\node [style=none] (37) at (1, -3) {};
		\node [style=none] (38) at (0, -4.5) {\rotatebox{-90}{Ansatz$_j$}};
		\node [style=none] (40) at (0, -2.75) {};
		\node [style=none] (41) at (-1, 5.25) {};
		\node [style=none] (42) at (1, 5.25) {};
		\node [style=none] (43) at (-1, 3.25) {};
		\node [style=none] (44) at (1, 3.25) {};
		\node [style=none] (45) at (0, 4.25) {Index};
		\node [style=none] (46) at (0, 3.25) {};
		\node [style=none] (47) at (-1, 5) {};
		\node [style=none] (48) at (-1, 3.5) {};
		\node [style=none] (49) at (1, 3.5) {};
		\node [style=none] (50) at (1, 5) {};
		\node [style=none] (51) at (-2, 2.5) {};
		\node [style=none] (52) at (-2, 1.5) {};
		\node [style=none] (54) at (3, 1.5) {};
		\node [style=none] (55) at (3, 2.5) {};
		\node [style=none] (57) at (-2.5, 2.5) {$|0\rangle$};
		\node [style=none] (58) at (-2.5, 1.5) {$|0\rangle$};
		\node [style=none] (61) at (5, -1) {$=$};
		\node [style=none] (62) at (7, -3) {};
		\node [style=none] (63) at (7, -4) {};
		\node [style=none] (64) at (7, -5) {};
		\node [style=none] (65) at (7, -6) {};
		\node [style=none] (66) at (22, -4) {};
		\node [style=none] (67) at (22, -5) {};
		\node [style=none] (68) at (8, -0.5) {};
		\node [style=none] (69) at (8, -2) {};
		\node [style=none] (70) at (8.5, -1) {$\vdots$};
		\node [style=none] (71) at (22, -6) {};
		\node [style=none] (72) at (22, -2) {};
		\node [style=none] (73) at (22, -0.5) {};
		\node [style=none] (74) at (21.5, -1) {$\vdots$};
		\node [style=none] (75) at (22, -3) {};
		\node [style=none] (76) at (7, 5) {};
		\node [style=none] (77) at (7, 3.5) {};
		\node [style=none] (78) at (7.5, 4.5) {$\vdots$};
		\node [style=none] (79) at (22, 5) {};
		\node [style=none] (80) at (22, 3.5) {};
		\node [style=none] (81) at (21.5, 4.5) {$\vdots$};
		\node [style=none] (82) at (7.5, -0.5) {$|0\rangle$};
		\node [style=none] (83) at (7.5, -2) {$|0\rangle$};
		\node [style=none] (110) at (8, 2.5) {};
		\node [style=none] (111) at (8, 1.5) {};
		\node [style=none] (112) at (22, 1.5) {};
		\node [style=none] (113) at (22, 2.5) {};
		\node [style=none] (114) at (7.5, 2.5) {$|0\rangle$};
		\node [style=none] (115) at (7.5, 1.5) {$|0\rangle$};
		\node [style=cnot targ] (118) at (9.05, 2.5) {};
		\node [style=cnot targ] (119) at (9.05, 1.5) {};
		\node [style=none] (120) at (8.75, 3.25) {};
		\node [style=none] (121) at (9.35, 3.25) {};
		\node [style=none] (122) at (8.75, 5.25) {};
		\node [style=none] (123) at (9.35, 5.25) {};
		\node [style=none] (124) at (9.05, 3.25) {};
		\node [style=none] (125) at (9.05, 4.25) {\rotatebox{-90}{\tiny Width$_{j}$}};
		\node [style=none] (126) at (8.75, 5) {};
		\node [style=none] (127) at (8.75, 3.5) {};
		\node [style=none] (128) at (9.35, 3.5) {};
		\node [style=none] (129) at (9.35, 5) {};
		\node [style=cnot targ] (130) at (20.05, 2.5) {};
		\node [style=cnot targ] (131) at (20.05, 1.5) {};
		\node [style=none] (132) at (19.75, 3.25) {};
		\node [style=none] (133) at (20.35, 3.25) {};
		\node [style=none] (134) at (19.75, 5.25) {};
		\node [style=none] (135) at (20.35, 5.25) {};
		\node [style=none] (136) at (20.05, 3.25) {};
		\node [style=none] (137) at (20.05, 4.25) {\rotatebox{-90}{\tiny Width$_{j}$}};
		\node [style=none] (138) at (19.75, 5) {};
		\node [style=none] (139) at (19.75, 3.5) {};
		\node [style=none] (140) at (20.35, 3.5) {};
		\node [style=none] (141) at (20.35, 5) {};
		\node [style=cnot ctrl] (142) at (10, 2.5) {};
		\node [style=cnot ctrl] (143) at (10, 1.5) {};
		\node [style=cnot targ] (144) at (10, 0.5) {};
		\node [style=cnot targ] (145) at (14.25, 0.5) {};
		\node [style=cnot ctrl] (146) at (14.25, 1.5) {};
		\node [style=cnot ctrl] (147) at (19, 2.5) {};
		\node [style=cnot ctrl] (148) at (19, 1.5) {};
		\node [style=cnot targ] (149) at (19, 0.5) {};
		\node [style=cnot targ] (150) at (12, 0.5) {};
		\node [style=cnot targ] (151) at (16, 0.5) {};
		\node [style=cnot ctrl] (152) at (16, 2.5) {};
		\node [style=cnot ctrl] (153) at (12, 2.5) {};
		\node [style=none] (154) at (-2, 0.5) {};
		\node [style=none] (155) at (3, 0.5) {};
		\node [style=none] (156) at (-2.5, 0.5) {$|0\rangle$};
		\node [style=none] (158) at (8, 0.5) {};
		\node [style=none] (159) at (7.5, 0.5) {$|0\rangle$};
		\node [style=none] (160) at (22, 0.5) {};
		\node [style=cnot targ] (162) at (18, 0.5) {};
		\node [style=cnot ctrl] (163) at (18, 1.5) {};
		\node [style=none] (165) at (10.5, -2.5) {};
		\node [style=none] (166) at (10.5, -6.25) {};
		\node [style=none] (167) at (11.5, -6.25) {};
		\node [style=none] (168) at (11.5, -2.5) {};
		\node [style=none] (169) at (10.5, -3) {};
		\node [style=none] (170) at (10.5, -4) {};
		\node [style=none] (171) at (10.5, -5) {};
		\node [style=none] (172) at (10.5, -6) {};
		\node [style=none] (173) at (11.5, -6) {};
		\node [style=none] (174) at (11.5, -5) {};
		\node [style=none] (175) at (11.5, -4) {};
		\node [style=none] (176) at (11.5, -3) {};
		\node [style=none] (177) at (11, -4.5) {\rotatebox{-90}{\tiny Ansatz$_j^4$}};
		\node [style=cnot ctrl] (178) at (11, 0.5) {};
		\node [style=none] (179) at (11, -2.5) {};
		\node [style=none] (180) at (10.5, 5.25) {};
		\node [style=none] (181) at (11.5, 5.25) {};
		\node [style=none] (182) at (10.5, 3.25) {};
		\node [style=none] (183) at (11.5, 3.25) {};
		\node [style=none] (184) at (11, 4.25) {\rotatebox{-90}{\tiny Index}};
		\node [style=none] (185) at (11, 3.25) {};
		\node [style=none] (186) at (10.5, 5) {};
		\node [style=none] (187) at (10.5, 3.5) {};
		\node [style=none] (188) at (11.5, 3.5) {};
		\node [style=none] (189) at (11.5, 5) {};
		\node [style=none] (190) at (11, 5.25) {};
		\node [style=none] (191) at (12.5, -2.5) {};
		\node [style=none] (192) at (12.5, -5.5) {};
		\node [style=none] (193) at (13.5, -5.5) {};
		\node [style=none] (194) at (13.5, -2.5) {};
		\node [style=none] (195) at (12.5, -3) {};
		\node [style=none] (196) at (12.5, -4) {};
		\node [style=none] (197) at (12.5, -5) {};
		\node [style=none] (200) at (13.5, -5) {};
		\node [style=none] (201) at (13.5, -4) {};
		\node [style=none] (202) at (13.5, -3) {};
		\node [style=none] (203) at (13, -4) {\rotatebox{-90}{\tiny Ansatz$_j^3$}};
		\node [style=cnot ctrl] (204) at (13, 0.5) {};
		\node [style=none] (205) at (13, -2.5) {};
		\node [style=none] (206) at (12.5, 5.25) {};
		\node [style=none] (207) at (13.5, 5.25) {};
		\node [style=none] (208) at (12.5, 3.25) {};
		\node [style=none] (209) at (13.5, 3.25) {};
		\node [style=none] (210) at (13, 4.25) {\rotatebox{-90}{\tiny Index}};
		\node [style=none] (211) at (13, 3.25) {};
		\node [style=none] (212) at (12.5, 5) {};
		\node [style=none] (213) at (12.5, 3.5) {};
		\node [style=none] (214) at (13.5, 3.5) {};
		\node [style=none] (215) at (13.5, 5) {};
		\node [style=none] (216) at (13, 5.25) {};
		\node [style=none] (217) at (16, -2.5) {};
		\node [style=none] (218) at (16, -3.5) {};
		\node [style=none] (219) at (18, -3.5) {};
		\node [style=none] (220) at (18, -2.5) {};
		\node [style=none] (221) at (16, -3) {};
		\node [style=none] (226) at (18, -3) {};
		\node [style=none] (227) at (17, -3) {\rotatebox{0}{\tiny Ansatz$_j^1$}};
		\node [style=cnot ctrl] (228) at (17, 0.5) {};
		\node [style=none] (229) at (17, -2.5) {};
		\node [style=none] (230) at (16.5, 5.25) {};
		\node [style=none] (231) at (17.5, 5.25) {};
		\node [style=none] (232) at (16.5, 3.25) {};
		\node [style=none] (233) at (17.5, 3.25) {};
		\node [style=none] (234) at (17, 4.25) {\rotatebox{-90}{\tiny Index}};
		\node [style=none] (235) at (17, 3.25) {};
		\node [style=none] (236) at (16.5, 5) {};
		\node [style=none] (237) at (16.5, 3.5) {};
		\node [style=none] (238) at (17.5, 3.5) {};
		\node [style=none] (239) at (17.5, 5) {};
		\node [style=none] (240) at (17, 5.25) {};
		\node [style=none] (241) at (14.5, -2.5) {};
		\node [style=none] (242) at (14.5, -4.5) {};
		\node [style=none] (243) at (15.5, -4.5) {};
		\node [style=none] (244) at (15.5, -2.5) {};
		\node [style=none] (245) at (14.5, -3) {};
		\node [style=none] (246) at (14.5, -4) {};
		\node [style=none] (249) at (15.5, -4) {};
		\node [style=none] (250) at (15.5, -3) {};
		\node [style=none] (251) at (15, -3.5) {\rotatebox{-90}{\tiny Ansatz$_j^2$}};
		\node [style=cnot ctrl] (252) at (15, 0.5) {};
		\node [style=none] (253) at (15, -2.5) {};
		\node [style=none] (254) at (14.5, 5.25) {};
		\node [style=none] (255) at (15.5, 5.25) {};
		\node [style=none] (256) at (14.5, 3.25) {};
		\node [style=none] (257) at (15.5, 3.25) {};
		\node [style=none] (258) at (15, 4.25) {\rotatebox{-90}{\tiny Index}};
		\node [style=none] (259) at (15, 3.25) {};
		\node [style=none] (260) at (14.5, 5) {};
		\node [style=none] (261) at (14.5, 3.5) {};
		\node [style=none] (262) at (15.5, 3.5) {};
		\node [style=none] (263) at (15.5, 5) {};
		\node [style=none] (264) at (15, 5.25) {};
		\node [style=cnot targ] (265) at (13.75, 0.5) {};
		\node [style=cnot ctrl] (266) at (13.75, 2.5) {};
		\node [style=cnot targ] (267) at (18.5, 0.5) {};
		\node [style=cnot ctrl] (268) at (18.5, 2.5) {};
		\node [style=none] (269) at (21.5, -3.2) {};
		\node [style=none] (270) at (21.7, -2.8) {};
		\node [style=none] (271) at (21.5, -4.2) {};
		\node [style=none] (272) at (21.7, -3.8) {};
		\node [style=none] (273) at (21.5, -5.2) {};
		\node [style=none] (274) at (21.7, -4.8) {};
		\node [style=none] (275) at (21.5, -6.2) {};
		\node [style=none] (276) at (21.7, -5.8) {};
		\node [style=none] (277) at (7.25, -3.2) {};
		\node [style=none] (278) at (7.45, -2.8) {};
		\node [style=none] (279) at (7.25, -4.2) {};
		\node [style=none] (280) at (7.45, -3.8) {};
		\node [style=none] (281) at (7.25, -5.2) {};
		\node [style=none] (282) at (7.45, -4.8) {};
		\node [style=none] (283) at (7.25, -6.2) {};
		\node [style=none] (284) at (7.45, -5.8) {};
		\node [style=none] (285) at (-2.75, -3.2) {};
		\node [style=none] (286) at (-2.55, -2.8) {};
		\node [style=none] (287) at (-2.75, -4.2) {};
		\node [style=none] (288) at (-2.55, -3.8) {};
		\node [style=none] (289) at (-2.75, -5.2) {};
		\node [style=none] (290) at (-2.55, -4.8) {};
		\node [style=none] (291) at (-2.75, -6.2) {};
		\node [style=none] (292) at (-2.55, -5.8) {};
		\node [style=none] (293) at (2.5, -3.2) {};
		\node [style=none] (294) at (2.7, -2.8) {};
		\node [style=none] (295) at (2.5, -4.2) {};
		\node [style=none] (296) at (2.7, -3.8) {};
		\node [style=none] (297) at (2.5, -5.2) {};
		\node [style=none] (298) at (2.7, -4.8) {};
		\node [style=none] (299) at (2.5, -6.2) {};
		\node [style=none] (300) at (2.7, -5.8) {};
	\end{pgfonlayer}
	\begin{pgfonlayer}{edgelayer}
		\draw (29.center) to (26.center);
		\draw (26.center) to (27.center);
		\draw (27.center) to (28.center);
		\draw (28.center) to (29.center);
		\draw (15.center) to (37.center);
		\draw (36.center) to (5.center);
		\draw (6.center) to (35.center);
		\draw (34.center) to (10.center);
		\draw (0.center) to (30.center);
		\draw (31.center) to (1.center);
		\draw (2.center) to (32.center);
		\draw (33.center) to (3.center);
		\draw (44.center) to (42.center);
		\draw (42.center) to (41.center);
		\draw (41.center) to (43.center);
		\draw (43.center) to (44.center);
		\draw (50.center) to (19.center);
		\draw (20.center) to (49.center);
		\draw (48.center) to (17.center);
		\draw (16.center) to (47.center);
		\draw (7.center) to (12.center);
		\draw (11.center) to (8.center);
		\draw (51.center) to (55.center);
		\draw (54.center) to (52.center);
		\draw (46.center) to (40.center);
		\draw (119) to (118);
		\draw [in=90, out=-90] (123.center) to (121.center);
		\draw (121.center) to (120.center);
		\draw (120.center) to (122.center);
		\draw (122.center) to (123.center);
		\draw (124.center) to (118);
		\draw (127.center) to (77.center);
		\draw (76.center) to (126.center);
		\draw (131) to (130);
		\draw [in=90, out=-90] (135.center) to (133.center);
		\draw (133.center) to (132.center);
		\draw (132.center) to (134.center);
		\draw (134.center) to (135.center);
		\draw (136.center) to (130);
		\draw (141.center) to (79.center);
		\draw (80.center) to (140.center);
		\draw (144) to (143);
		\draw (143) to (142);
		\draw (145) to (146);
		\draw (149) to (148);
		\draw (148) to (147);
		\draw (142) to (118);
		\draw (119) to (143);
		\draw (153) to (150);
		\draw (152) to (151);
		\draw (148) to (131);
		\draw (131) to (112.center);
		\draw (110.center) to (118);
		\draw (119) to (111.center);
		\draw (130) to (113.center);
		\draw (155.center) to (154.center);
		\draw (158.center) to (144);
		\draw (160.center) to (149);
		\draw (147) to (130);
		\draw (163) to (162);
		\draw (168.center) to (165.center);
		\draw (165.center) to (166.center);
		\draw (166.center) to (167.center);
		\draw (167.center) to (168.center);
		\draw (183.center) to (181.center);
		\draw (181.center) to (180.center);
		\draw (180.center) to (182.center);
		\draw (182.center) to (183.center);
		\draw (194.center) to (191.center);
		\draw (191.center) to (192.center);
		\draw (192.center) to (193.center);
		\draw (193.center) to (194.center);
		\draw (209.center) to (207.center);
		\draw (207.center) to (206.center);
		\draw (206.center) to (208.center);
		\draw (208.center) to (209.center);
		\draw (220.center) to (217.center);
		\draw (217.center) to (218.center);
		\draw (218.center) to (219.center);
		\draw (219.center) to (220.center);
		\draw (233.center) to (231.center);
		\draw (231.center) to (230.center);
		\draw (230.center) to (232.center);
		\draw (232.center) to (233.center);
		\draw (244.center) to (241.center);
		\draw (241.center) to (242.center);
		\draw (242.center) to (243.center);
		\draw (243.center) to (244.center);
		\draw (257.center) to (255.center);
		\draw (255.center) to (254.center);
		\draw (254.center) to (256.center);
		\draw (256.center) to (257.center);
		\draw (169.center) to (62.center);
		\draw (63.center) to (170.center);
		\draw (64.center) to (171.center);
		\draw (172.center) to (65.center);
		\draw (176.center) to (195.center);
		\draw (196.center) to (175.center);
		\draw (174.center) to (197.center);
		\draw (266) to (265);
		\draw (268) to (267);
		\draw (245.center) to (202.center);
		\draw (201.center) to (246.center);
		\draw (250.center) to (221.center);
		\draw (144) to (178);
		\draw (178) to (150);
		\draw (150) to (204);
		\draw (204) to (265);
		\draw (265) to (145);
		\draw (145) to (252);
		\draw (252) to (151);
		\draw (151) to (228);
		\draw (228) to (162);
		\draw (162) to (267);
		\draw (267) to (149);
		\draw (262.center) to (237.center);
		\draw (261.center) to (214.center);
		\draw (213.center) to (188.center);
		\draw (189.center) to (212.center);
		\draw (215.center) to (260.center);
		\draw (263.center) to (236.center);
		\draw (185.center) to (178);
		\draw (178) to (179.center);
		\draw (211.center) to (204);
		\draw (204) to (205.center);
		\draw (259.center) to (252);
		\draw (252) to (253.center);
		\draw (235.center) to (228);
		\draw (228) to (229.center);
		\draw (238.center) to (139.center);
		\draw (138.center) to (239.center);
		\draw (186.center) to (129.center);
		\draw (128.center) to (187.center);
		\draw (226.center) to (75.center);
		\draw (66.center) to (249.center);
		\draw (200.center) to (67.center);
		\draw (173.center) to (71.center);
		\draw (73.center) to (68.center);
		\draw (69.center) to (72.center);
		\draw (142) to (153);
		\draw (153) to (266);
		\draw (266) to (152);
		\draw (152) to (268);
		\draw (268) to (147);
		\draw (148) to (163);
		\draw (163) to (146);
		\draw (146) to (143);
		\draw (270.center) to (269.center);
		\draw (272.center) to (271.center);
		\draw (274.center) to (273.center);
		\draw (276.center) to (275.center);
		\draw (278.center) to (277.center);
		\draw (280.center) to (279.center);
		\draw (282.center) to (281.center);
		\draw (284.center) to (283.center);
		\draw (286.center) to (285.center);
		\draw (288.center) to (287.center);
		\draw (290.center) to (289.center);
		\draw (292.center) to (291.center);
		\draw (294.center) to (293.center);
		\draw (296.center) to (295.center);
		\draw (298.center) to (297.center);
		\draw (300.center) to (299.center);
	\end{pgfonlayer}
\end{tikzpicture}

%% file: figures/tt-oracle/nounmux.tikz
\begin{tikzpicture}
	\begin{pgfonlayer}{nodelayer}
		\node [style=none] (0) at (-3, 1) {};
		\node [style=none] (1) at (-3, -0.5) {};
		\node [style=none] (2) at (-3, -1.5) {};
		\node [style=none] (3) at (-3, -3) {};
		\node [style=none] (4) at (3, -0.5) {};
		\node [style=none] (5) at (3, -1.5) {};
		\node [style=none] (6) at (-2, 3.5) {};
		\node [style=none] (7) at (-2, 2) {};
		\node [style=none] (8) at (-1.5, 3) {$\vdots$};
		\node [style=none] (9) at (3, -3) {};
		\node [style=none] (10) at (3, 2) {};
		\node [style=none] (11) at (3, 3.5) {};
		\node [style=none] (12) at (2.5, 3) {$\vdots$};
		\node [style=none] (13) at (3, 1) {};
		\node [style=none] (14) at (-3, 6) {};
		\node [style=none] (15) at (-3, 4.5) {};
		\node [style=none] (16) at (-2.5, 5.5) {$\vdots$};
		\node [style=none] (17) at (3, 6) {};
		\node [style=none] (18) at (3, 4.5) {};
		\node [style=none] (19) at (2.5, 5.5) {$\vdots$};
		\node [style=none] (20) at (-2.5, 3.5) {$|0\rangle$};
		\node [style=none] (21) at (-2.5, 2) {$|0\rangle$};
		\node [style=none] (24) at (-1, 1.25) {};
		\node [style=none] (25) at (-1, -3.25) {};
		\node [style=none] (26) at (1, -3.25) {};
		\node [style=none] (27) at (1, 1.25) {};
		\node [style=none] (28) at (-1, 1) {};
		\node [style=none] (29) at (-1, -0.5) {};
		\node [style=none] (30) at (-1, -1.5) {};
		\node [style=none] (31) at (-1, -3) {};
		\node [style=none] (32) at (1, -3) {};
		\node [style=none] (33) at (1, -1.5) {};
		\node [style=none] (34) at (1, -0.5) {};
		\node [style=none] (35) at (1, 1) {};
		\node [style=none] (36) at (0, -1) {\rotatebox{-90}{MUX$_j$}};
		\node [style=none] (37) at (0, 1.25) {};
		\node [style=none] (38) at (-1, 6.25) {};
		\node [style=none] (39) at (1, 6.25) {};
		\node [style=none] (40) at (-1, 4.25) {};
		\node [style=none] (41) at (1, 4.25) {};
		\node [style=none] (42) at (0, 5.25) {Index};
		\node [style=none] (43) at (0, 4.25) {};
		\node [style=none] (44) at (-1, 6) {};
		\node [style=none] (45) at (-1, 4.5) {};
		\node [style=none] (46) at (1, 4.5) {};
		\node [style=none] (47) at (1, 6) {};
		\node [style=none] (61) at (6, 1) {};
		\node [style=none] (62) at (6.5, 0.5) {$\vdots$};
		\node [style=none] (63) at (6, -0.5) {};
		\node [style=none] (64) at (16.5, 1) {};
		\node [style=none] (65) at (16, 0.5) {$\vdots$};
		\node [style=none] (66) at (7, 3.5) {};
		\node [style=none] (67) at (7, 2) {};
		\node [style=none] (68) at (7.5, 3) {$\vdots$};
		\node [style=none] (69) at (16.5, -0.5) {};
		\node [style=none] (70) at (16.5, 2) {};
		\node [style=none] (71) at (16.5, 3.5) {};
		\node [style=none] (72) at (16, 3) {$\vdots$};
		\node [style=none] (74) at (6, 6) {};
		\node [style=none] (75) at (6, 4.5) {};
		\node [style=none] (76) at (6.5, 5.5) {$\vdots$};
		\node [style=none] (77) at (16.5, 6) {};
		\node [style=none] (78) at (16.5, 4.5) {};
		\node [style=none] (79) at (16, 5.5) {$\vdots$};
		\node [style=none] (80) at (6.5, 3.5) {$|0\rangle$};
		\node [style=none] (81) at (6.5, 2) {$|0\rangle$};
		\node [style=cnot targ] (84) at (8.05, 3.5) {};
		\node [style=cnot targ] (85) at (8.05, 2) {};
		\node [style=none] (86) at (7.75, 4.25) {};
		\node [style=none] (87) at (8.35, 4.25) {};
		\node [style=none] (88) at (7.75, 6.25) {};
		\node [style=none] (89) at (8.35, 6.25) {};
		\node [style=none] (90) at (8.05, 4.25) {};
		\node [style=none] (91) at (8.05, 5.25) {\rotatebox{-90}{\tiny Index$_{j1}$}};
		\node [style=none] (92) at (7.75, 6) {};
		\node [style=none] (93) at (7.75, 4.5) {};
		\node [style=none] (94) at (8.35, 4.5) {};
		\node [style=none] (95) at (8.35, 6) {};
		\node [style=none] (96) at (6, -1.5) {};
		\node [style=none] (97) at (6.5, -2) {$\vdots$};
		\node [style=none] (98) at (6, -3) {};
		\node [style=none] (99) at (16.5, -1.5) {};
		\node [style=none] (100) at (16, -2) {$\vdots$};
		\node [style=none] (101) at (16.5, -3) {};
		\node [style=none] (102) at (8.25, 1.25) {};
		\node [style=none] (103) at (9.75, 1.25) {};
		\node [style=none] (104) at (8.25, -3.25) {};
		\node [style=none] (105) at (9.75, -3.25) {};
		\node [style=none] (106) at (9, -1) {$P$};
		\node [style=cnot ctrl] (107) at (9, 2) {};
		\node [style=cnot ctrl] (108) at (9, 3.5) {};
		\node [style=none] (109) at (9, 1.25) {};
		\node [style=none] (110) at (8.25, 1) {};
		\node [style=none] (111) at (8.25, -0.5) {};
		\node [style=none] (112) at (8.25, -1.5) {};
		\node [style=none] (113) at (8.25, -3) {};
		\node [style=cnot targ] (114) at (10.05, 3.5) {};
		\node [style=cnot targ] (115) at (10.05, 2) {};
		\node [style=none] (116) at (9.75, 4.25) {};
		\node [style=none] (117) at (10.35, 4.25) {};
		\node [style=none] (118) at (9.75, 6.25) {};
		\node [style=none] (119) at (10.35, 6.25) {};
		\node [style=none] (120) at (10.05, 4.25) {};
		\node [style=none] (121) at (10.05, 5.25) {\rotatebox{-90}{\tiny Index$_{j1}$}};
		\node [style=none] (122) at (9.75, 6) {};
		\node [style=none] (123) at (9.75, 4.5) {};
		\node [style=none] (124) at (10.35, 4.5) {};
		\node [style=none] (125) at (10.35, 6) {};
		\node [style=cnot targ] (126) at (12.55, 3.5) {};
		\node [style=cnot targ] (127) at (12.55, 2) {};
		\node [style=none] (128) at (12.25, 4.25) {};
		\node [style=none] (129) at (12.85, 4.25) {};
		\node [style=none] (130) at (12.25, 6.25) {};
		\node [style=none] (131) at (12.85, 6.25) {};
		\node [style=none] (132) at (12.55, 4.25) {};
		\node [style=none] (133) at (12.55, 5.25) {\rotatebox{-90}{\tiny Index$_{jk}$}};
		\node [style=none] (134) at (12.25, 6) {};
		\node [style=none] (135) at (12.25, 4.5) {};
		\node [style=none] (136) at (12.85, 4.5) {};
		\node [style=none] (137) at (12.85, 6) {};
		\node [style=none] (138) at (12.75, -0.25) {};
		\node [style=none] (139) at (14.25, -0.25) {};
		\node [style=none] (140) at (12.75, -3.25) {};
		\node [style=none] (141) at (14.25, -3.25) {};
		\node [style=none] (142) at (13.5, -1.75) {$P$};
		\node [style=cnot ctrl] (143) at (13.5, 2) {};
		\node [style=cnot ctrl] (144) at (13.5, 3.5) {};
		\node [style=none] (145) at (13.5, -0.25) {};
		\node [style=none] (147) at (12.75, -0.5) {};
		\node [style=none] (148) at (12.75, -1.5) {};
		\node [style=none] (149) at (12.75, -3) {};
		\node [style=cnot targ] (150) at (14.55, 3.5) {};
		\node [style=cnot targ] (151) at (14.55, 2) {};
		\node [style=none] (152) at (14.25, 4.25) {};
		\node [style=none] (153) at (14.85, 4.25) {};
		\node [style=none] (154) at (14.25, 6.25) {};
		\node [style=none] (155) at (14.85, 6.25) {};
		\node [style=none] (156) at (14.55, 4.25) {};
		\node [style=none] (157) at (14.55, 5.25) {\rotatebox{-90}{\tiny Index$_{jk}$}};
		\node [style=none] (158) at (14.25, 6) {};
		\node [style=none] (159) at (14.25, 4.5) {};
		\node [style=none] (160) at (14.85, 4.5) {};
		\node [style=none] (161) at (14.85, 6) {};
		\node [style=none] (162) at (9.75, 1) {};
		\node [style=none] (163) at (9.75, -0.5) {};
		\node [style=none] (164) at (9.75, -1.5) {};
		\node [style=none] (165) at (9.75, -3) {};
		\node [style=none] (167) at (14.25, -0.5) {};
		\node [style=none] (168) at (14.25, -1.5) {};
		\node [style=none] (169) at (14.25, -3) {};
		\node [style=none] (170) at (11.25, 0.5) {\rotatebox{10}{$\ddots$}};
		\node [style=none] (171) at (11.25, 5.25) {$\dots$};
		\node [style=none] (172) at (11.25, -2.25) {$\dots$};
		\node [style=none] (173) at (2.5, 0.5) {$\vdots$};
		\node [style=none] (174) at (2.5, -2) {$\vdots$};
		\node [style=none] (175) at (-2.5, 0.5) {$\vdots$};
		\node [style=none] (176) at (-2.5, -2) {$\vdots$};
		\node [style=none] (177) at (4.5, 1) {$=$};
		\node [style=none] (178) at (11.25, 2.75) {$\dots$};
		\node [style=none] (179) at (-3.5, 1) {};
		\node [style=none] (180) at (-3.5, -0.5) {};
		\node [style=none] (183) at (-4, 0.25) {$k$};
		\node [style=none] (184) at (-2.725, 0.8) {};
		\node [style=none] (185) at (-2.525, 1.2) {};
		\node [style=none] (186) at (6.25, 0.8) {};
		\node [style=none] (187) at (6.45, 1.2) {};
		\node [style=none] (188) at (6.25, -0.7) {};
		\node [style=none] (189) at (6.45, -0.3) {};
		\node [style=none] (190) at (6.25, -1.7) {};
		\node [style=none] (191) at (6.45, -1.3) {};
		\node [style=none] (192) at (6.25, -3.2) {};
		\node [style=none] (193) at (6.45, -2.8) {};
		\node [style=none] (194) at (-2.75, -0.7) {};
		\node [style=none] (195) at (-2.55, -0.3) {};
		\node [style=none] (196) at (-2.75, -1.7) {};
		\node [style=none] (197) at (-2.55, -1.3) {};
		\node [style=none] (198) at (-2.75, -3.2) {};
		\node [style=none] (199) at (-2.55, -2.8) {};
		\node [style=none] (200) at (16.025, 0.8) {};
		\node [style=none] (201) at (16.225, 1.2) {};
		\node [style=none] (202) at (16, -0.7) {};
		\node [style=none] (203) at (16.2, -0.3) {};
		\node [style=none] (204) at (16, -1.7) {};
		\node [style=none] (205) at (16.2, -1.3) {};
		\node [style=none] (206) at (16, -3.2) {};
		\node [style=none] (207) at (16.2, -2.8) {};
		\node [style=none] (208) at (2.525, 0.8) {};
		\node [style=none] (209) at (2.725, 1.2) {};
		\node [style=none] (210) at (2.5, -0.7) {};
		\node [style=none] (211) at (2.7, -0.3) {};
		\node [style=none] (212) at (2.5, -1.7) {};
		\node [style=none] (213) at (2.7, -1.3) {};
		\node [style=none] (214) at (2.5, -3.2) {};
		\node [style=none] (215) at (2.7, -2.8) {};
	\end{pgfonlayer}
	\begin{pgfonlayer}{edgelayer}
		\draw (27.center) to (24.center);
		\draw (24.center) to (25.center);
		\draw (25.center) to (26.center);
		\draw (26.center) to (27.center);
		\draw (13.center) to (35.center);
		\draw (34.center) to (4.center);
		\draw (5.center) to (33.center);
		\draw (32.center) to (9.center);
		\draw (0.center) to (28.center);
		\draw (29.center) to (1.center);
		\draw (2.center) to (30.center);
		\draw (31.center) to (3.center);
		\draw (41.center) to (39.center);
		\draw (39.center) to (38.center);
		\draw (38.center) to (40.center);
		\draw (40.center) to (41.center);
		\draw (47.center) to (17.center);
		\draw (18.center) to (46.center);
		\draw (45.center) to (15.center);
		\draw (14.center) to (44.center);
		\draw (6.center) to (11.center);
		\draw (10.center) to (7.center);
		\draw (43.center) to (37.center);
		\draw (85) to (84);
		\draw [in=90, out=-90] (89.center) to (87.center);
		\draw (87.center) to (86.center);
		\draw (86.center) to (88.center);
		\draw (88.center) to (89.center);
		\draw (90.center) to (84);
		\draw (103.center) to (102.center);
		\draw (102.center) to (104.center);
		\draw (104.center) to (105.center);
		\draw (105.center) to (103.center);
		\draw (108) to (107);
		\draw (109.center) to (107);
		\draw (113.center) to (98.center);
		\draw (96.center) to (112.center);
		\draw (111.center) to (63.center);
		\draw (61.center) to (110.center);
		\draw (115) to (114);
		\draw [in=90, out=-90] (119.center) to (117.center);
		\draw (117.center) to (116.center);
		\draw (116.center) to (118.center);
		\draw (118.center) to (119.center);
		\draw (120.center) to (114);
		\draw (114) to (108);
		\draw (108) to (84);
		\draw (85) to (107);
		\draw (107) to (115);
		\draw (127) to (126);
		\draw [in=90, out=-90] (131.center) to (129.center);
		\draw (129.center) to (128.center);
		\draw (128.center) to (130.center);
		\draw (130.center) to (131.center);
		\draw (132.center) to (126);
		\draw (139.center) to (138.center);
		\draw (138.center) to (140.center);
		\draw (140.center) to (141.center);
		\draw (141.center) to (139.center);
		\draw (144) to (143);
		\draw (145.center) to (143);
		\draw (151) to (150);
		\draw (155.center) to (153.center);
		\draw (153.center) to (152.center);
		\draw (152.center) to (154.center);
		\draw (154.center) to (155.center);
		\draw (156.center) to (150);
		\draw (150) to (144);
		\draw (144) to (126);
		\draw (127) to (143);
		\draw (143) to (151);
		\draw (148.center) to (164.center);
		\draw (149.center) to (165.center);
		\draw (163.center) to (147.center);
		\draw (168.center) to (99.center);
		\draw (101.center) to (169.center);
		\draw (167.center) to (69.center);
		\draw (162.center) to (64.center);
		\draw (126) to (114);
		\draw (127) to (115);
		\draw (150) to (71.center);
		\draw (70.center) to (151);
		\draw (85) to (67.center);
		\draw (66.center) to (84);
		\draw (93.center) to (75.center);
		\draw (92.center) to (74.center);
		\draw (95.center) to (122.center);
		\draw (123.center) to (94.center);
		\draw (124.center) to (135.center);
		\draw (134.center) to (125.center);
		\draw (137.center) to (158.center);
		\draw (161.center) to (77.center);
		\draw (78.center) to (160.center);
		\draw (159.center) to (136.center);
		\draw [style=brace edge] (180.center) to (179.center);
		\draw (185.center) to (184.center);
		\draw (187.center) to (186.center);
		\draw (189.center) to (188.center);
		\draw (191.center) to (190.center);
		\draw (193.center) to (192.center);
		\draw (195.center) to (194.center);
		\draw (197.center) to (196.center);
		\draw (199.center) to (198.center);
		\draw (201.center) to (200.center);
		\draw (203.center) to (202.center);
		\draw (205.center) to (204.center);
		\draw (207.center) to (206.center);
		\draw (209.center) to (208.center);
		\draw (211.center) to (210.center);
		\draw (213.center) to (212.center);
		\draw (215.center) to (214.center);
	\end{pgfonlayer}
\end{tikzpicture}

%% file: figures/tt-oracle/oracle.tikz
\begin{tikzpicture}
	\begin{pgfonlayer}{nodelayer}
		\node [style=none] (0) at (-6, 2.5) {};
		\node [style=none] (1) at (-6, 1) {};
		\node [style=none] (2) at (-6, 0) {};
		\node [style=none] (3) at (-6, -1.5) {};
		\node [style=none] (14) at (-6, 5) {};
		\node [style=none] (15) at (-6, 3.5) {};
		\node [style=none] (16) at (-5.5, 4.5) {$\vdots$};
		\node [style=none] (24) at (-4, 2.75) {};
		\node [style=none] (25) at (-4, -1.75) {};
		\node [style=none] (26) at (-3, -1.75) {};
		\node [style=none] (27) at (-3, 2.75) {};
		\node [style=none] (28) at (-4, 2.5) {};
		\node [style=none] (29) at (-4, 1) {};
		\node [style=none] (30) at (-4, 0) {};
		\node [style=none] (31) at (-4, -1.5) {};
		\node [style=none] (32) at (-3, -1.5) {};
		\node [style=none] (33) at (-3, 0) {};
		\node [style=none] (34) at (-3, 1) {};
		\node [style=none] (35) at (-3, 2.5) {};
		\node [style=none] (36) at (-3.5, 0.5) {\rotatebox{-90}{\tiny MUX$_1$}};
		\node [style=none] (37) at (-3.5, 2.75) {};
		\node [style=none] (38) at (-4, 5.25) {};
		\node [style=none] (39) at (-3, 5.25) {};
		\node [style=none] (40) at (-4, 3.25) {};
		\node [style=none] (41) at (-3, 3.25) {};
		\node [style=none] (42) at (-3.5, 4.25) {\rotatebox{-90}{\tiny Index}};
		\node [style=none] (43) at (-3.5, 3.25) {};
		\node [style=none] (44) at (-4, 5) {};
		\node [style=none] (45) at (-4, 3.5) {};
		\node [style=none] (46) at (-3, 3.5) {};
		\node [style=none] (47) at (-3, 5) {};
		\node [style=none] (50) at (-5.5, 2) {$\vdots$};
		\node [style=none] (51) at (-5.5, -0.5) {$\vdots$};
		\node [style=none] (52) at (-5.725, 2.3) {};
		\node [style=none] (53) at (-5.525, 2.7) {};
		\node [style=none] (54) at (-5.75, 0.8) {};
		\node [style=none] (55) at (-5.55, 1.2) {};
		\node [style=none] (56) at (-5.75, -0.2) {};
		\node [style=none] (57) at (-5.55, 0.2) {};
		\node [style=none] (58) at (-5.75, -1.7) {};
		\node [style=none] (59) at (-5.55, -1.3) {};
		\node [style=none] (72) at (13, 1) {};
		\node [style=none] (73) at (13, 0) {};
		\node [style=none] (77) at (13, -1.5) {};
		\node [style=none] (81) at (13, 2.5) {};
		\node [style=none] (85) at (13, 5) {};
		\node [style=none] (86) at (13, 3.5) {};
		\node [style=none] (87) at (12.5, 4.5) {$\vdots$};
		\node [style=none] (116) at (12.5, 2) {$\vdots$};
		\node [style=none] (117) at (12.5, -0.5) {$\vdots$};
		\node [style=none] (128) at (12.525, 2.3) {};
		\node [style=none] (129) at (12.725, 2.7) {};
		\node [style=none] (130) at (12.5, 0.8) {};
		\node [style=none] (131) at (12.7, 1.2) {};
		\node [style=none] (132) at (12.5, -0.2) {};
		\node [style=none] (133) at (12.7, 0.2) {};
		\node [style=none] (134) at (12.5, -1.7) {};
		\node [style=none] (135) at (12.7, -1.3) {};
		\node [style=none] (160) at (-2.5, 2.75) {};
		\node [style=none] (161) at (-2.5, 0.75) {};
		\node [style=none] (162) at (-1.5, 0.75) {};
		\node [style=none] (163) at (-1.5, 2.75) {};
		\node [style=none] (164) at (-2.5, 2.5) {};
		\node [style=none] (167) at (-2.5, 1) {};
		\node [style=none] (168) at (-1.5, 1) {};
		\node [style=none] (171) at (-1.5, 2.5) {};
		\node [style=none] (172) at (-2, 1.75) {\rotatebox{-90}{\tiny Ansatz$_1$}};
		\node [style=none] (173) at (-2, 2.75) {};
		\node [style=none] (174) at (-2.5, 5.25) {};
		\node [style=none] (175) at (-1.5, 5.25) {};
		\node [style=none] (176) at (-2.5, 3.25) {};
		\node [style=none] (177) at (-1.5, 3.25) {};
		\node [style=none] (178) at (-2, 4.25) {\rotatebox{-90}{\tiny Index}};
		\node [style=none] (179) at (-2, 3.25) {};
		\node [style=none] (180) at (-2.5, 5) {};
		\node [style=none] (181) at (-2.5, 3.5) {};
		\node [style=none] (182) at (-1.5, 3.5) {};
		\node [style=none] (183) at (-1.5, 5) {};
		\node [style=none] (184) at (0.75, 2.75) {};
		\node [style=none] (185) at (0.75, -1.75) {};
		\node [style=none] (186) at (1.75, -1.75) {};
		\node [style=none] (187) at (1.75, 2.75) {};
		\node [style=none] (188) at (0.75, 2.5) {};
		\node [style=none] (189) at (0.75, 1) {};
		\node [style=none] (190) at (0.75, 0) {};
		\node [style=none] (191) at (0.75, -1.5) {};
		\node [style=none] (192) at (1.75, -1.5) {};
		\node [style=none] (193) at (1.75, 0) {};
		\node [style=none] (194) at (1.75, 1) {};
		\node [style=none] (195) at (1.75, 2.5) {};
		\node [style=none] (196) at (1.25, 0.5) {\rotatebox{-90}{\tiny MUX$_2$}};
		\node [style=none] (197) at (1.25, 2.75) {};
		\node [style=none] (198) at (0.75, 5.25) {};
		\node [style=none] (199) at (1.75, 5.25) {};
		\node [style=none] (200) at (0.75, 3.25) {};
		\node [style=none] (201) at (1.75, 3.25) {};
		\node [style=none] (202) at (1.25, 4.25) {\rotatebox{-90}{\tiny Index}};
		\node [style=none] (203) at (1.25, 3.25) {};
		\node [style=none] (204) at (0.75, 5) {};
		\node [style=none] (205) at (0.75, 3.5) {};
		\node [style=none] (206) at (1.75, 3.5) {};
		\node [style=none] (207) at (1.75, 5) {};
		\node [style=none] (232) at (2.25, 2.75) {};
		\node [style=none] (233) at (2.25, 0.75) {};
		\node [style=none] (234) at (3.25, 0.75) {};
		\node [style=none] (235) at (3.25, 2.75) {};
		\node [style=none] (236) at (2.25, 2.5) {};
		\node [style=none] (237) at (2.25, 1) {};
		\node [style=none] (238) at (3.25, 1) {};
		\node [style=none] (239) at (3.25, 2.5) {};
		\node [style=none] (240) at (2.75, 1.75) {\rotatebox{-90}{\tiny Ansatz$_2$}};
		\node [style=none] (241) at (2.75, 2.75) {};
		\node [style=none] (242) at (2.25, 5.25) {};
		\node [style=none] (243) at (3.25, 5.25) {};
		\node [style=none] (244) at (2.25, 3.25) {};
		\node [style=none] (245) at (3.25, 3.25) {};
		\node [style=none] (246) at (2.75, 4.25) {\rotatebox{-90}{\tiny Index}};
		\node [style=none] (247) at (2.75, 3.25) {};
		\node [style=none] (248) at (2.25, 5) {};
		\node [style=none] (249) at (2.25, 3.5) {};
		\node [style=none] (250) at (3.25, 3.5) {};
		\node [style=none] (251) at (3.25, 5) {};
		\node [style=none] (252) at (6.75, 2.75) {};
		\node [style=none] (253) at (6.75, -1.75) {};
		\node [style=none] (254) at (7.75, -1.75) {};
		\node [style=none] (255) at (7.75, 2.75) {};
		\node [style=none] (256) at (6.75, 2.5) {};
		\node [style=none] (257) at (6.75, 1) {};
		\node [style=none] (258) at (6.75, 0) {};
		\node [style=none] (259) at (6.75, -1.5) {};
		\node [style=none] (260) at (7.75, -1.5) {};
		\node [style=none] (261) at (7.75, 0) {};
		\node [style=none] (262) at (7.75, 1) {};
		\node [style=none] (263) at (7.75, 2.5) {};
		\node [style=none] (264) at (7.25, 0.5) {\rotatebox{-90}{\tiny MUX$_m$}};
		\node [style=none] (265) at (7.25, 2.75) {};
		\node [style=none] (266) at (6.75, 5.25) {};
		\node [style=none] (267) at (7.75, 5.25) {};
		\node [style=none] (268) at (6.75, 3.25) {};
		\node [style=none] (269) at (7.75, 3.25) {};
		\node [style=none] (270) at (7.25, 4.25) {\rotatebox{-90}{\tiny Index}};
		\node [style=none] (271) at (7.25, 3.25) {};
		\node [style=none] (272) at (6.75, 5) {};
		\node [style=none] (273) at (6.75, 3.5) {};
		\node [style=none] (274) at (7.75, 3.5) {};
		\node [style=none] (275) at (7.75, 5) {};
		\node [style=none] (276) at (8.25, 2.75) {};
		\node [style=none] (277) at (8.25, 0.75) {};
		\node [style=none] (278) at (9.25, 0.75) {};
		\node [style=none] (279) at (9.25, 2.75) {};
		\node [style=none] (280) at (8.25, 2.5) {};
		\node [style=none] (281) at (8.25, 1) {};
		\node [style=none] (282) at (9.25, 1) {};
		\node [style=none] (283) at (9.25, 2.5) {};
		\node [style=none] (284) at (8.75, 1.75) {\rotatebox{-90}{\tiny Ansatz$_m$}};
		\node [style=none] (285) at (8.75, 2.75) {};
		\node [style=none] (286) at (8.25, 5.25) {};
		\node [style=none] (287) at (9.25, 5.25) {};
		\node [style=none] (288) at (8.25, 3.25) {};
		\node [style=none] (289) at (9.25, 3.25) {};
		\node [style=none] (290) at (8.75, 4.25) {\rotatebox{-90}{\tiny Index}};
		\node [style=none] (291) at (8.75, 3.25) {};
		\node [style=none] (292) at (8.25, 5) {};
		\node [style=none] (293) at (8.25, 3.5) {};
		\node [style=none] (294) at (9.25, 3.5) {};
		\node [style=none] (295) at (9.25, 5) {};
		\node [style=none] (296) at (5.75, 1.75) {$\dots$};
		\node [style=none] (297) at (5.75, 4.25) {$\dots$};
		\node [style=none] (298) at (5.75, -0.75) {$\dots$};
		\node [style=none] (300) at (-12.5, 2.5) {};
		\node [style=none] (301) at (-12.5, 1) {};
		\node [style=none] (302) at (-12.5, 0) {};
		\node [style=none] (303) at (-12.5, -1.5) {};
		\node [style=none] (307) at (-12.5, 5) {};
		\node [style=none] (308) at (-12.5, 3.5) {};
		\node [style=none] (309) at (-12, 4.5) {$\vdots$};
		\node [style=none] (336) at (-12, 2) {$\vdots$};
		\node [style=none] (337) at (-12, -0.5) {$\vdots$};
		\node [style=none] (338) at (-12.225, 2.3) {};
		\node [style=none] (339) at (-12.025, 2.7) {};
		\node [style=none] (340) at (-12.25, 0.8) {};
		\node [style=none] (341) at (-12.05, 1.2) {};
		\node [style=none] (342) at (-12.25, -0.2) {};
		\node [style=none] (343) at (-12.05, 0.2) {};
		\node [style=none] (344) at (-12.25, -1.7) {};
		\node [style=none] (345) at (-12.05, -1.3) {};
		\node [style=none] (346) at (-8, 1) {};
		\node [style=none] (347) at (-8, 0) {};
		\node [style=none] (348) at (-8, -1.5) {};
		\node [style=none] (351) at (-8, 2.5) {};
		\node [style=none] (352) at (-8, 5) {};
		\node [style=none] (353) at (-8, 3.5) {};
		\node [style=none] (354) at (-8.5, 4.5) {$\vdots$};
		\node [style=none] (357) at (-8.5, 2) {$\vdots$};
		\node [style=none] (358) at (-8.5, -0.5) {$\vdots$};
		\node [style=none] (359) at (-8.475, 2.3) {};
		\node [style=none] (360) at (-8.275, 2.7) {};
		\node [style=none] (361) at (-8.5, 0.8) {};
		\node [style=none] (362) at (-8.3, 1.2) {};
		\node [style=none] (363) at (-8.5, -0.2) {};
		\node [style=none] (364) at (-8.3, 0.2) {};
		\node [style=none] (365) at (-8.5, -1.7) {};
		\node [style=none] (366) at (-8.3, -1.3) {};
		\node [style=none] (479) at (-11.25, 5.25) {};
		\node [style=none] (480) at (-9.25, 5.25) {};
		\node [style=none] (481) at (-9.25, 3.25) {};
		\node [style=none] (482) at (-11.25, 3.25) {};
		\node [style=none] (483) at (-10.25, 4.25) {Index};
		\node [style=none] (484) at (-11.25, 2.75) {};
		\node [style=none] (485) at (-9.25, 2.75) {};
		\node [style=none] (486) at (-9.25, -1.75) {};
		\node [style=none] (487) at (-11.25, -1.75) {};
		\node [style=none] (488) at (-10.25, 0.5) {$\mathcal{W}$};
		\node [style=none] (489) at (-9.25, 2.5) {};
		\node [style=none] (490) at (-9.25, 1) {};
		\node [style=none] (491) at (-9.25, 0) {};
		\node [style=none] (492) at (-9.25, -1.5) {};
		\node [style=none] (493) at (-11.25, 2.5) {};
		\node [style=none] (494) at (-11.25, 1) {};
		\node [style=none] (495) at (-11.25, 0) {};
		\node [style=none] (496) at (-11.25, -1.5) {};
		\node [style=none] (497) at (-9.25, 3.5) {};
		\node [style=none] (498) at (-11.25, 3.5) {};
		\node [style=none] (499) at (-11.25, 5) {};
		\node [style=none] (500) at (-9.25, 5) {};
		\node [style=none] (501) at (-7, 1.75) {$=$};
		\node [style=none] (502) at (-10.25, 3.25) {};
		\node [style=none] (503) at (-10.25, 2.75) {};
		\node [style=none] (504) at (-1, 2.75) {};
		\node [style=none] (505) at (-1, -1.75) {};
		\node [style=none] (506) at (0, -1.75) {};
		\node [style=none] (507) at (0, 2.75) {};
		\node [style=none] (508) at (-1, 2.5) {};
		\node [style=none] (509) at (-1, 1) {};
		\node [style=none] (510) at (-1, 0) {};
		\node [style=none] (511) at (-1, -1.5) {};
		\node [style=none] (512) at (0, -1.5) {};
		\node [style=none] (513) at (0, 0) {};
		\node [style=none] (514) at (0, 1) {};
		\node [style=none] (515) at (0, 2.5) {};
		\node [style=none] (516) at (-0.5, 0.5) {\rotatebox{-90}{\tiny MUX$_1^\dagger$}};
		\node [style=none] (517) at (-0.5, 2.75) {};
		\node [style=none] (518) at (-1, 5.25) {};
		\node [style=none] (519) at (0, 5.25) {};
		\node [style=none] (520) at (-1, 3.25) {};
		\node [style=none] (521) at (0, 3.25) {};
		\node [style=none] (522) at (-0.5, 4.25) {\rotatebox{-90}{\tiny Index}};
		\node [style=none] (523) at (-0.5, 3.25) {};
		\node [style=none] (524) at (-1, 5) {};
		\node [style=none] (525) at (-1, 3.5) {};
		\node [style=none] (526) at (0, 3.5) {};
		\node [style=none] (527) at (0, 5) {};
		\node [style=none] (528) at (3.75, 2.75) {};
		\node [style=none] (529) at (3.75, -1.75) {};
		\node [style=none] (530) at (4.75, -1.75) {};
		\node [style=none] (531) at (4.75, 2.75) {};
		\node [style=none] (532) at (3.75, 2.5) {};
		\node [style=none] (533) at (3.75, 1) {};
		\node [style=none] (534) at (3.75, 0) {};
		\node [style=none] (535) at (3.75, -1.5) {};
		\node [style=none] (536) at (4.75, -1.5) {};
		\node [style=none] (537) at (4.75, 0) {};
		\node [style=none] (538) at (4.75, 1) {};
		\node [style=none] (539) at (4.75, 2.5) {};
		\node [style=none] (540) at (4.25, 0.5) {\rotatebox{-90}{\tiny MUX$_2^\dagger$}};
		\node [style=none] (541) at (4.25, 2.75) {};
		\node [style=none] (542) at (3.75, 5.25) {};
		\node [style=none] (543) at (4.75, 5.25) {};
		\node [style=none] (544) at (3.75, 3.25) {};
		\node [style=none] (545) at (4.75, 3.25) {};
		\node [style=none] (546) at (4.25, 4.25) {\rotatebox{-90}{\tiny Index}};
		\node [style=none] (547) at (4.25, 3.25) {};
		\node [style=none] (548) at (3.75, 5) {};
		\node [style=none] (549) at (3.75, 3.5) {};
		\node [style=none] (550) at (4.75, 3.5) {};
		\node [style=none] (551) at (4.75, 5) {};
		\node [style=none] (552) at (9.75, 2.75) {};
		\node [style=none] (553) at (9.75, -1.75) {};
		\node [style=none] (554) at (10.75, -1.75) {};
		\node [style=none] (555) at (10.75, 2.75) {};
		\node [style=none] (556) at (9.75, 2.5) {};
		\node [style=none] (557) at (9.75, 1) {};
		\node [style=none] (558) at (9.75, 0) {};
		\node [style=none] (559) at (9.75, -1.5) {};
		\node [style=none] (560) at (10.75, -1.5) {};
		\node [style=none] (561) at (10.75, 0) {};
		\node [style=none] (562) at (10.75, 1) {};
		\node [style=none] (563) at (10.75, 2.5) {};
		\node [style=none] (564) at (10.25, 0.5) {\rotatebox{-90}{\tiny MUX$_m^\dagger$}};
		\node [style=none] (565) at (10.25, 2.75) {};
		\node [style=none] (566) at (9.75, 5.25) {};
		\node [style=none] (567) at (10.75, 5.25) {};
		\node [style=none] (568) at (9.75, 3.25) {};
		\node [style=none] (569) at (10.75, 3.25) {};
		\node [style=none] (570) at (10.25, 4.25) {\rotatebox{-90}{\tiny Index}};
		\node [style=none] (571) at (10.25, 3.25) {};
		\node [style=none] (572) at (9.75, 5) {};
		\node [style=none] (573) at (9.75, 3.5) {};
		\node [style=none] (574) at (10.75, 3.5) {};
		\node [style=none] (575) at (10.75, 5) {};
		\node [style=none] (577) at (-13.75, 4.25) {$|i\rangle$};
		\node [style=none] (578) at (-13, 5) {};
		\node [style=none] (579) at (-13, 3.5) {};
		\node [style=none] (580) at (-13, 2.5) {};
		\node [style=none] (581) at (-13, -1.5) {};
		\node [style=none] (582) at (-13.75, 0.5) {$|\psi\rangle$};
	\end{pgfonlayer}
	\begin{pgfonlayer}{edgelayer}
		\draw (27.center) to (24.center);
		\draw (24.center) to (25.center);
		\draw (25.center) to (26.center);
		\draw (26.center) to (27.center);
		\draw (0.center) to (28.center);
		\draw (29.center) to (1.center);
		\draw (2.center) to (30.center);
		\draw (31.center) to (3.center);
		\draw (41.center) to (39.center);
		\draw (39.center) to (38.center);
		\draw (38.center) to (40.center);
		\draw (40.center) to (41.center);
		\draw (45.center) to (15.center);
		\draw (14.center) to (44.center);
		\draw (43.center) to (37.center);
		\draw (53.center) to (52.center);
		\draw (55.center) to (54.center);
		\draw (57.center) to (56.center);
		\draw (59.center) to (58.center);
		\draw (129.center) to (128.center);
		\draw (131.center) to (130.center);
		\draw (133.center) to (132.center);
		\draw (135.center) to (134.center);
		\draw (163.center) to (160.center);
		\draw (160.center) to (161.center);
		\draw (161.center) to (162.center);
		\draw (162.center) to (163.center);
		\draw (177.center) to (175.center);
		\draw (175.center) to (174.center);
		\draw (174.center) to (176.center);
		\draw (176.center) to (177.center);
		\draw (179.center) to (173.center);
		\draw (164.center) to (35.center);
		\draw (34.center) to (167.center);
		\draw (47.center) to (180.center);
		\draw (181.center) to (46.center);
		\draw (187.center) to (184.center);
		\draw (184.center) to (185.center);
		\draw (185.center) to (186.center);
		\draw (186.center) to (187.center);
		\draw (201.center) to (199.center);
		\draw (199.center) to (198.center);
		\draw (198.center) to (200.center);
		\draw (200.center) to (201.center);
		\draw (203.center) to (197.center);
		\draw (235.center) to (232.center);
		\draw (232.center) to (233.center);
		\draw (233.center) to (234.center);
		\draw (234.center) to (235.center);
		\draw (245.center) to (243.center);
		\draw (243.center) to (242.center);
		\draw (242.center) to (244.center);
		\draw (244.center) to (245.center);
		\draw (247.center) to (241.center);
		\draw (236.center) to (195.center);
		\draw (194.center) to (237.center);
		\draw (207.center) to (248.center);
		\draw (249.center) to (206.center);
		\draw (255.center) to (252.center);
		\draw (252.center) to (253.center);
		\draw (253.center) to (254.center);
		\draw (254.center) to (255.center);
		\draw (269.center) to (267.center);
		\draw (267.center) to (266.center);
		\draw (266.center) to (268.center);
		\draw (268.center) to (269.center);
		\draw (271.center) to (265.center);
		\draw (279.center) to (276.center);
		\draw (276.center) to (277.center);
		\draw (277.center) to (278.center);
		\draw (278.center) to (279.center);
		\draw (289.center) to (287.center);
		\draw (287.center) to (286.center);
		\draw (286.center) to (288.center);
		\draw (288.center) to (289.center);
		\draw (291.center) to (285.center);
		\draw (280.center) to (263.center);
		\draw (262.center) to (281.center);
		\draw (275.center) to (292.center);
		\draw (293.center) to (274.center);
		\draw (339.center) to (338.center);
		\draw (341.center) to (340.center);
		\draw (343.center) to (342.center);
		\draw (345.center) to (344.center);
		\draw (360.center) to (359.center);
		\draw (362.center) to (361.center);
		\draw (364.center) to (363.center);
		\draw (366.center) to (365.center);
		\draw (482.center) to (479.center);
		\draw (479.center) to (480.center);
		\draw (480.center) to (481.center);
		\draw (481.center) to (482.center);
		\draw (487.center) to (484.center);
		\draw (484.center) to (485.center);
		\draw (485.center) to (486.center);
		\draw (486.center) to (487.center);
		\draw (496.center) to (303.center);
		\draw (495.center) to (302.center);
		\draw (494.center) to (301.center);
		\draw (493.center) to (300.center);
		\draw (489.center) to (351.center);
		\draw (490.center) to (346.center);
		\draw (491.center) to (347.center);
		\draw (492.center) to (348.center);
		\draw (498.center) to (308.center);
		\draw (499.center) to (307.center);
		\draw (500.center) to (352.center);
		\draw (353.center) to (497.center);
		\draw (503.center) to (502.center);
		\draw (507.center) to (504.center);
		\draw (504.center) to (505.center);
		\draw (505.center) to (506.center);
		\draw (506.center) to (507.center);
		\draw (521.center) to (519.center);
		\draw (519.center) to (518.center);
		\draw (518.center) to (520.center);
		\draw (520.center) to (521.center);
		\draw (523.center) to (517.center);
		\draw (531.center) to (528.center);
		\draw (528.center) to (529.center);
		\draw (529.center) to (530.center);
		\draw (530.center) to (531.center);
		\draw (545.center) to (543.center);
		\draw (543.center) to (542.center);
		\draw (542.center) to (544.center);
		\draw (544.center) to (545.center);
		\draw (547.center) to (541.center);
		\draw (258.center) to (537.center);
		\draw (534.center) to (193.center);
		\draw (555.center) to (552.center);
		\draw (552.center) to (553.center);
		\draw (553.center) to (554.center);
		\draw (554.center) to (555.center);
		\draw (569.center) to (567.center);
		\draw (567.center) to (566.center);
		\draw (566.center) to (568.center);
		\draw (568.center) to (569.center);
		\draw (571.center) to (565.center);
		\draw (171.center) to (508.center);
		\draw (509.center) to (168.center);
		\draw (510.center) to (33.center);
		\draw (514.center) to (189.center);
		\draw (188.center) to (515.center);
		\draw (190.center) to (513.center);
		\draw (191.center) to (512.center);
		\draw (511.center) to (32.center);
		\draw (192.center) to (535.center);
		\draw (536.center) to (259.center);
		\draw (257.center) to (538.center);
		\draw (539.center) to (256.center);
		\draw (532.center) to (239.center);
		\draw (238.center) to (533.center);
		\draw (283.center) to (556.center);
		\draw (557.center) to (282.center);
		\draw (261.center) to (558.center);
		\draw (260.center) to (559.center);
		\draw (77.center) to (560.center);
		\draw (561.center) to (73.center);
		\draw (72.center) to (562.center);
		\draw (563.center) to (81.center);
		\draw (86.center) to (574.center);
		\draw (575.center) to (85.center);
		\draw (572.center) to (295.center);
		\draw (294.center) to (573.center);
		\draw (273.center) to (550.center);
		\draw (551.center) to (272.center);
		\draw (549.center) to (250.center);
		\draw (251.center) to (548.center);
		\draw (205.center) to (526.center);
		\draw (527.center) to (204.center);
		\draw (524.center) to (183.center);
		\draw (182.center) to (525.center);
		\draw [style=brace edge] (579.center) to (578.center);
		\draw [style=brace edge] (581.center) to (580.center);
	\end{pgfonlayer}
\end{tikzpicture}